\documentclass[acmsmall,nonacm,authorversion]{acmart}

\setcopyright{none}
\renewcommand\footnotetextcopyrightpermission[1]{}
\usepackage{amsthm}
\usepackage[shortlabels]{enumitem}
\usepackage[ruled,linesnumbered,noend]{algorithm2e}
\usepackage{cleveref}
\usepackage{subcaption} %
\usepackage{todonotes}
\usepackage{thmtools}
\usepackage{etoolbox}  %

\makeatletter
\patchcmd{\@footnotetext}{\@makefntext}{\@@makefntext}{}{}
\patchcmd{\@mpfootnotetext}{\@makefntext}{\@@makefntext}{}{}
\makeatother

\newcommand{\eps}{\varepsilon}
\newcommand{\rank}{\mathrm{rank}}
\newcommand{\estRank}{\widehat{\rank}}

\newcommand{\consolidate}{\textsc{Consolidate}\xspace}

\newcommand{\ignored}[1]{}

\definecolor{RoyalBlue}{RGB}{26, 13, 230}   %
\definecolor{ForestGreen}{RGB}{34, 139, 34} %
\definecolor{BrickRed}{RGB}{203, 65, 84}    %
\definecolor{DarkOrchid}{RGB}{153, 50, 204}

\DeclareMathOperator{\Err}{Err}

\newtheorem{theorem}{Theorem}

\newtheorem{definition}[theorem]{Definition}
\newtheorem{lemma}[theorem]{Lemma}
\newtheorem{observation}[theorem]{Observation}

\theoremstyle{remark}

\title{SplineSketch: Even More Accurate Quantiles with Error Guarantees}

\author{Aleksander Łukasiewicz}
\orcid{0000-0003-1808-8330}
\affiliation{\institution{Computer Science Institute of Charles University}\city{Prague}\country{Czech Republic}}
\email{olekluka@iuuk.mff.cuni.cz}
    
\author{Jakub Tětek}
\orcid{0000-0002-2046-1627}
\affiliation{\institution{INSAIT, University of Sofia "St. Kliment Ohridski"}\city{Sofia}\country{Bulgaria}}
\email{j.tetek@gmail.com}

\author{Pavel Vesel{\'{y}}}
\orcid{0000-0003-1169-7934}
\affiliation{\institution{Computer Science Institute of Charles University}\city{Prague}\country{Czech Republic}}
\email{vesely@iuuk.mff.cuni.cz}

\begin{document}

\begin{abstract}
Space-efficient streaming estimation of quantiles in massive datasets is a fundamental problem with numerous applications in data monitoring and analysis. While theoretical research led to optimal algorithms, such as the Greenwald-Khanna algorithm or the KLL sketch, practitioners often use other sketches that perform significantly better in practice but lack theoretical guarantees. Most notably, the widely used $t$-digest has unbounded worst-case error.

In this paper, we seek to get the best of both worlds. We present a new quantile summary, SplineSketch, for numeric data, offering near-optimal theoretical guarantees, namely uniformly bounded rank error,  and outperforming $t$-digest by a factor of 2--20 on a range of synthetic and real-world datasets. To achieve such performance, we develop a novel approach that maintains a dynamic subdivision of the input range into buckets while fitting the input distribution using monotone cubic spline interpolation.
\end{abstract}



\maketitle

\section{Introduction}

Data sketching has become one of the main tools for dealing with massive data volumes. %
Sketches provide a scalable way to extract key features from large datasets, enabling
real-time analysis at streaming speed or massively parallel processing of distributed datasets, in applications such as network monitoring, machine learning, privacy, or bioinformatics; see %
e.g.~\cite{SifferFTL17,cormodeY20,Cormode21a,RothchildPUISB020,Smith0T20,PaghS21,baker2023genomic,bonnie2024dandd,ClemensSGH23}.

Approximating order statistics %
is a central sketching problem. %
The goal is to summarize a massive, potentially distributed dataset in a small space into a \emph{sketch} that is incrementally updatable, mergeable across data sources, and accurately approximates the data distribution.
Namely, the sketch should provide low-error estimates for the median, percentiles, and their generalization, quantiles,
and should also approximately answer rank queries -- that is, estimate the number of input items no larger than a given element.
Quantile sketches have been applied, e.g., to monitoring latencies~\cite{tene2015_latency_talk}
or for detecting anomalies~\cite{SifferFTL17}.

Streaming quantile estimation has thus received a significant attention in research community.
From the theory point of view, it has been solved to a large extent. Namely, optimal algorithms are known for the uniform error, such as the Greenwald-Khanna (GK) sketch~\cite{greenwald2001space,CormodeV20} and KLL~\cite{KarninLL16}.
However, the accuracy of these theory-based quantile summaries %
on real-world data is only a small constant factor better than the worst-case bound.
Instead, practitioners often use algorithms that perform significantly better in practice but have no guarantees.
Most notably, $t$-digest~\cite{dunning19-t-digest} is widely used --- according to~\cite{dunning21}, it has been adopted by major tech companies (reportedly Microsoft, Facebook, Google) and open-source projects (Elasticsearch \cite{::elasticdocs}, TimescaleDB \cite{::tigerdatadocumentation}, Apache Dubbo \cite{::apachedubbodocumentation}, etc.).
At the same time, no matter how large the $t$-digest is, its error can be arbitrarily bad~\cite{CormodeMRV21}.
Furthermore, to the best of our knowledge, other quantile sketches do not provide both worst-case bounds and error similar to $t$-digest in practice; see Section~\ref{sec:related}.

Here, we introduce a new quantile sketch, called SplineSketch, that not only gets the best of both worlds -- theoretical guarantees and very high accuracy in practice -- but also significantly outperforms the state of the art, including $t$-digest~\cite{dunning19-t-digest};
see Figure~\ref{figTop:error} and Section~\ref{sec:experiments}.
Our sketch processes any numerical input in a streaming fashion and comes with a merge operation,
enabling efficient parallel or distributed processing of large datasets.
Before explaining our contribution in more detail, we describe our setting, including the considered error guarantees.

\begin{figure}[t]
	\vspace{-0.2cm}
	\includegraphics[width=0.9\textwidth]{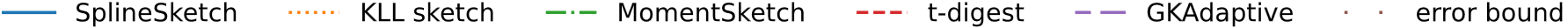}
	\vspace{-0.3cm}
	\centering
	\subcaptionbox{HEPMASS dataset~\cite{hepmass}\label{fig:intro-hepmass}}{%
		\includegraphics[width=0.32\linewidth]{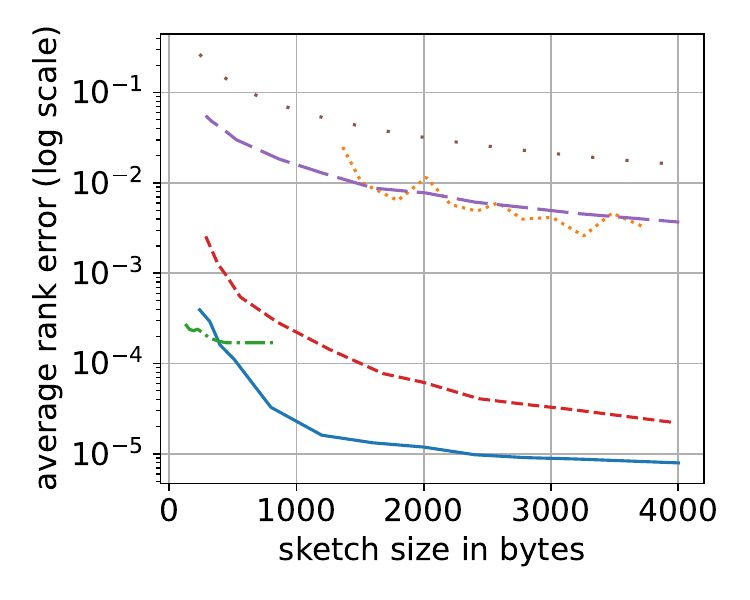}%
	}
	\subcaptionbox{Power dataset~\cite{individual_household_electric_power_consumption_235}\label{fig:intro-power}}{%
		\includegraphics[width=0.32\linewidth]{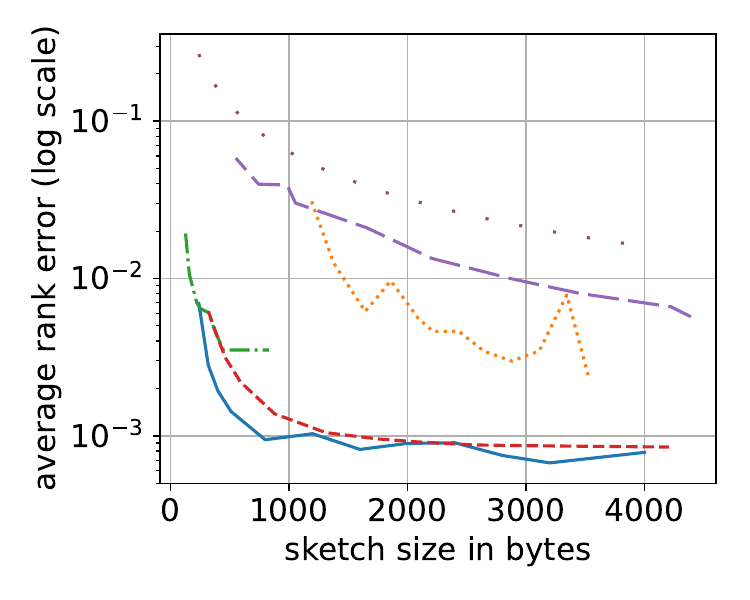}%
	}
	\subcaptionbox{Books dataset~\cite{sosd,MarcusKRSMK0K20sosd}\label{fig:intro-books}}{%
		\includegraphics[width=0.32\linewidth]{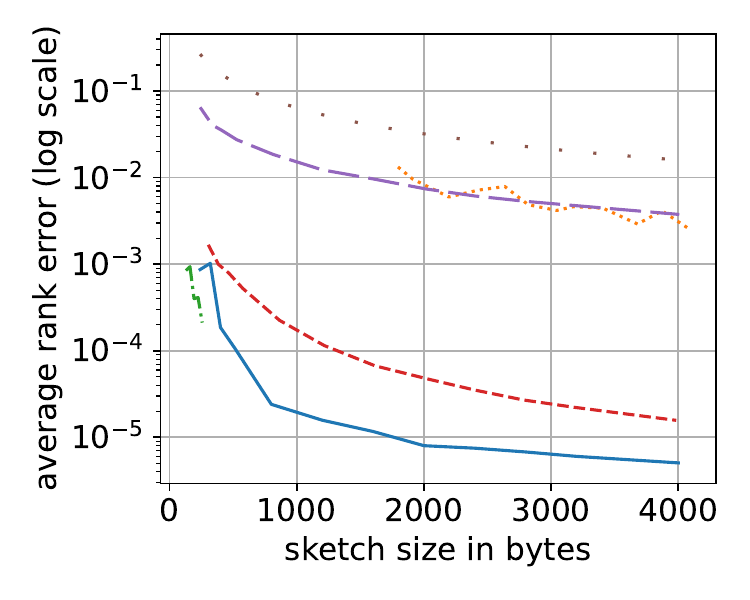}%
	}
	\caption{Average error of SplineSketch and other available sketches on three real-world datasets
		(Section~\ref{sec:experiments}).
		\vspace{-0.3cm}}
	\label{figTop:error}
\end{figure} 

\paragraph{Definitions.}
We use $n$ for the number of items currently summarized by the sketch.
For a multiset $S$ of real numbers and any $x \in \mathbb{R}$, let $\rank(x) = |\{y 
\in S : y \le x\}|$, i.e., the number of input items smaller than or equal to $x$.
For any $q \in (0,1]$, we call $y := \min \{x\in S : \rank(x) \geq q \cdot n\}$ the \emph{$q$-quantile} of $S$,
i.e., the $q$-quantile is the $\lceil q\cdot n\rceil$-th smallest item in $S$ (if there are no duplicates).
This definition implies that a method for estimating ranks may be used for estimating quantiles
with the same accuracy, by performing binary search.
The \emph{aspect ratio} $\alpha$ of $S$ is defined as 
$\alpha = (\max(S) - \min(S)) / \eta$, where $\eta = \min_{i \neq j: x_i \neq x_j} |x_i - x_j|$ 
is the smallest difference of distinct items on input, assuming at least two distinct items are present.
(To avoid formal issues on inputs with $\alpha < 2$, we in fact use $\max(2, \alpha)$
	so that $\log_2 \alpha \ge 1$.)

	We focus on the \emph{streaming and mergeability settings}. The streaming model requires that the dataset $S$ is
	processed in one pass with low memory (in an arbitrary, possibly adversarial, order).
	In the more general setting of mergeable summaries~\cite{AgarwalCHPWY13},
	$S$ is arbitrarily split into a number of parts, each summarized in the streaming fashion,
	and the resulting summaries are combined via pairwise merge operations.
	In other words, besides answering rank and quantile queries,
	the sketch needs to have an update operation, which adds a new observation to the dataset,
	and a merge operation, which takes two sketches and outputs a single sketch of their combined input.

	Following a long line of work, e.g., \cite{shrivastava2004medians,greenwald2001space,AgarwalCHPWY13,KarninLL16,luo16_quantiles_experimental,CormodeV20,AssadiJPS23-GK_weighted,GuptaSW24optimaQS-non-comp,GribelyukSWY24_GK-KG_sketch},
	we study the \emph{uniform rank error}, where, for a given accuracy parameter $\eps > 0$ and any query item $x$ (that may or may not be in the input), we require
	the algorithm to return a rank estimate $\estRank(x)$ such that
	$|\estRank(x) - \rank(x)| \le \eps \cdot n$,
	and for any quantile query $q\in (0, 1]$, the algorithm should return a $q'$-quantile for $q' = q \pm \eps$.

\paragraph{Our contributions}
Our main contribution is a new technique for designing quantile summaries with bounded worst-case error and high accuracy in practice, combining dynamically maintained subdivision of the input range with a suitable interpolation
over the buckets, outlined in \Cref{sec:overview} and described in \Cref{sec:SplineSketchDescription}.

We demonstrate the power of this technique on a new quantile sketch, called SplineSketch.
It  requires a single parameter $k$ that determines its space usage, which
is $O(k)$ memory words for processing the input and $16\cdot k$ bytes when serialized for storage.
The sketch deterministically guarantees a \emph{theoretically near-optimal uniform rank error} on any input, close to the best possible error of $O(n/k)$
with $k$ memory words; namely, the error is $O(n/k \cdot \log \alpha)$,
where $\alpha$ is the aspect ratio. %
At the same time, the accuracy of our sketch in practice is in fact typically 10--200 times better than $n/k$.
Furthermore, we analyze mergeability of our sketch with respect to worst-case error and show that a similar error bounds applies when merging sketches summarizing a similar number of items.
Finally, our sketch can be \emph{resized} based on the current memory requirements.

We evaluate SplineSketch in a prototype implementation and compare it with state-of-the-art sketches on a range of synthetic and real-world datasets.
We demonstrate both in the streaming and mergeability settings
that SplineSketch has error smaller by a factor of 2--20 compared to $t$-digest's error when given the same space
(see Figure~\ref{figTop:error});
in some cases, the improvement of the maximum error is by two orders of magnitude.
For an input with frequent items, a.k.a.\ heavy hitters, such accuracy is achieved by additionally using
the Misra-Gries sketch~\cite{misra1982finding} to filter them out.
The improvement over the GK~\cite{greenwald2001space} and KLL sketches~\cite{KarninLL16} is even larger, typically by 2--3 orders of magnitude.
MomentSketch~\cite{GanDTSB18} performs better on many smooth distributions as it is very compact,
but requires very large query time and its accuracy on real-world datasets is worse, without the possibility for improvement by increasing the number of moments and log-moments due to numerical issues.
We also evaluate the update, query, and merge times. Similarly to $t$-digest, we keep a buffer of size $O(k)$ when building the sketch, with bigger buffer resulting in faster update time.
Depending on the data size, we measured the average time per update of our prototype implementation to be 0.1 to 1$\mu$s, similar to $t$-digest.
Nevertheless, MomentSketch and KLL achieve 2--3 times faster update times than SplineSketch due to their simplicity.
However, the query time of SplineSketch is also 0.1 to 1$\mu$ depending on the number of queries, one of the fastest in our evaluation.
The time per merge operation in our implementation is worse than for other sketches (up to hundreds of $\mu$s) and remains to be optimized.	
Finally, we also provide ablation studies, justifying the design choices behind our sketch.

Overall, we demonstrate that besides worst-case guarantees and high accuracy in practice,
our approach offers flexibility, making it suitable for a range of applications.

\subsection{Overview of Our Approach}\label{sec:overview}

The core of our approach is to maintain a subdivision of the input range into \emph{buckets}
in a way that adapts well to the input distribution.
Namely, SplineSketch consists of thresholds $\tau_1, \cdots, \tau_k$ and then $(\tau_{i-1}, \tau_i]$ is the $i$-th bucket, with $\tau_0 = -\infty$.
For each bucket, we store a counter $b_i$ that represents an approximate number of items that lie in $(\tau_{i-1}, \tau_i]$;
we refer to $b_i$ as the \emph{size} of the bucket $i$.
Furthermore, we use cubic spline interpolation in order to answer queries accurately in practice;
see Figure~\ref{fig:sketch} for an illustration. %

\begin{figure}[t]
	\centering
	\includegraphics[width=0.5\linewidth]{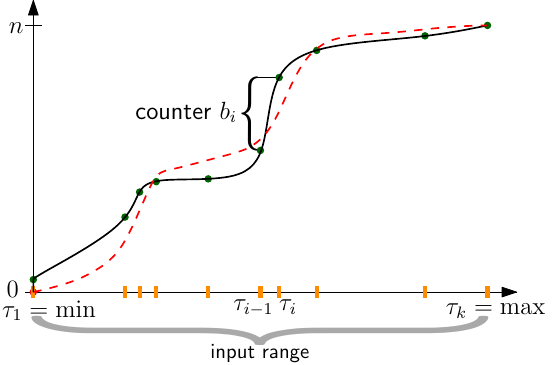}
	\caption{SplineSketch illustration, with the values and thresholds $\tau_i$ on the $x$ axis
	and the rank space $[0,n]$ on the $y$ axis.
	The true cumulative distribution function of the data is depicted as a red dashed curve.
	The green dots correspond to the prefix sums of bucket counters, i.e., the counter $b_i$ is the difference of the $y$ values of the green dots at $\tau_i$ and $\tau_{i-1}$.
	The solid black curve is an interpolation over the green dots.}
	\label{fig:sketch}
\end{figure}

\paragraph{Answering queries using cubic splines.}
The estimated rank of any threshold $\tau_i$ is $\sum_{j=1}^i b_i$. 
However, it is not clear what is the best way of answering rank queries that are not at one of the bucket boundaries $\tau_i$. Most of the previous work would just instead query the closest value $\tau_j$ and return that as the answer, or perform a linear interpolation. 
However, we seek to do better than that, using more advanced interpolation methods, namely the piecewise cubic Hermite interpolating polynomial (PCHIP)~\cite{fritsch.butland:1984:siamj.sci.andstat.comput., fritsch1980monotone}.
Specifically, if we are querying $x$, such that $\tau_i < x < \tau_{i+1}$, we estimate the ranks $\hat{r}_{i-1}, \cdots, \hat{r}_{i+2}$ of $\tau_{i-1}, \cdots, \tau_{i+2}$, then compute the PCHIP in order to interpolate between $\hat{r}_{i}$ and $\hat{r}_{i+1}$, and evaluate this polynomial at $x$, which gives the estimate that SplineSketch returns. 

\paragraph{Maintaining the buckets.}
The most important choice we need to make is how to set the thresholds of the buckets. 
We always set $\tau_1$ and $\tau_k$ to be the smallest and largest items seen so far, respectively. 
As new data arrives, we adjust the buckets by splitting some of them and joining certain pairs of adjacent ones.
When we \emph{join} two adjacent buckets, we simply remove the threshold they share and set the new counter to the sum of the two bucket counters.
In the \emph{split} procedure, we split the bucket in the middle. 
We then need to decide how to divide the count between the two new sub-buckets. 
To this end, we perform a query (in a way described above) at the midpoint of the original bucket and set the counts of the new sub-buckets to be consistent with the query result.
It remains to specify how we choose which buckets to split and join.
For simplicity, in the description below, we ignore the issue of update efficiency.
Ultimately, we address this by using a buffer and performing updates in batches, thereby achieving low amortized update time.

There are two rules for choosing a bucket to split.
Firstly, throughout the execution, we ensure that all buckets have counter $O(n/k)$, and we split a bucket when it gets bigger than the bound. %
We later use these size bounds, together with other invariants, to show the worst-case error bounds. 
Intuitively speaking, we show that no individual item participates in too many splits, and that each split contributes to the total error by at most the size of the bucket containing that item.

The second rule for bucket splitting is intended to achieve significantly better error in practice, while preserving the worst-case error bounds. 
To this end, for each bucket we introduce a value called the \emph{heuristic error}, which is our estimate of the order of error of the spline interpolation within the bucket.
The motivation behind this notion is that changes in the data distribution function are captured by the derivatives of that distribution.
To illustrate the intuition behind our heuristic error, suppose the data items are i.i.d.\ from an unknown distribution. 
If the growth of its CDF $F$ does not change too much around a bucket, specifically the absolute value of its second derivative $F''$ is small, the error of the spline approximation should be low as well.
In practice, we do not know the CDF and the data may not be i.i.d., but we nevertheless use an estimate of the CDF's second derivative around a bucket as a heuristic approximation of the true error.
We do not store this value explicitly and instead compute it only when needed.
If the total heuristic error can be decreased substantially by splitting one bucket and joining two other adjacent buckets, while preserving bucket bounds, then we perform the split and join of these buckets. 
The exact function that we use as the heuristic error is given by Equation~\ref{eqn:heurError-der2} in Section~\ref{sec:consolidate}.

To ensure we always have exactly $k$ buckets, whenever we split a bucket, we join some pair of adjacent buckets.
When we perform a join, we would like to choose two adjacent buckets which together would form a bucket with the lowest possible value of the heuristic error.
However, we need to ensure that the new bucket would not grow too large.
Therefore, we only consider pairs for which the size of the resulting bucket would be not too close to the bucket bound, specifically 
at most 75\% of the bucket bound.

Finally, we want to avoid performing too many joins and splits on a single bucket, as this can cause the error to accumulate.
To this end, we divide the input into epochs, where one epoch is defined as time during which the input size increased by a fixed constant factor.
If the bucket was created by a split, we mark its thresholds as \emph{protected} and we do not allow it to be joined during the same epoch.
We prove that despite all these constraints after a split there will always be at least one pair of adjacent buckets that can be joined.

\paragraph{Merging two sketches.}

The merge operation for our sketch is conceptually simple:
Given two SplineSketches on input,
we first create a combined set of thresholds by taking the union of both sketches’ thresholds,
which removes duplicates. %
We then assign counters to the new buckets based on the sum of their estimated ranks in both original sketches,
for which we use precomputed interpolations over the input sketches.
Taking the union of thresholds typically results in a sketch with more buckets than desired,
so we reduce the number of buckets by joining adjacent ones with the lowest heuristic error, until we reach the target bucket count.
To maintain worst-case guarantees, if the two sketches being merged differ significantly in size, the threshold protection from the larger sketch is preserved during merging; however, when their size is similar, we basically reset the epoch, together with the threshold protection.

\paragraph{Resizing the sketch.}
Using joining or splitting the buckets allows us to easily change the number $k$ of buckets.
Namely, when we want to make the sketch smaller, we just join an appropriate number of adjacent buckets without any splits, 
whereas for increasing the sketch size, we perform splits without joins.
The selection of buckets to split or join is done similarly as in maintaining the buckets described above.
This demonstrates the flexibility of our approach, and can be practically useful, e.g., when additional memory is allocated during input processing and the resulting sketch is later compressed to a much smaller size for storage, yielding overall better accuracy than processing directly with the smaller size.

\paragraph{Paper outline.}
\Cref{sec:related} outlines the most relevant work on quantile sketches and provides a comparison with our approach.
The full algorithmic description of SplineSketch appears in \Cref{sec:SplineSketchDescription},
together with its theoretical properties (the proofs are in Appendices~\ref{app:monotoneCubicPolyMiddleBounds} and~\ref{sec:analysis}).
The implementation details appear in \Cref{sec:implementation} and the experimental evaluation in \Cref{sec:experiments}.
We conclude with discussing the pros and cons of our approach and the open problems in \Cref{sec:conclusions}.

\section{Related Work}\label{sec:related}

Quantile summaries are intensively investigated since the 90s, with the pioneering work of Munro and Paterson from 1978~\cite{MunroP80} who showed that exact quantile selection in the streaming setting, i.e., one pass with sublinear space, is impossible. 
In the comparison-based model, Greenwald and Khanna~\cite{greenwald2001space}
designed a deterministic algorithm, called the GK summary, that stores $O(\eps^{-1}\cdot \log \eps n)$ stream items and guarantees $\pm\, \eps\cdot n$ uniform error;
this bound is optimal among deterministic comparison-based quantile summaries~\cite{CormodeV20}.
However, the GK summary is not known to be mergeable while retaining its space bounds~\cite{AgarwalCHPWY13},
despite the fact that the intricate analysis from~\cite{greenwald2001space} was recently simplified and generalized to weighted updates~\cite{AssadiJPS23-GK_weighted,GribelyukSWY24_GK-KG_sketch}.
Randomization allows the removal of the logarithmic dependency on the stream length $n$. 
After a sequence of improvements~\cite{manku1999random,luo16_quantiles_experimental,felber2015randomized}, Karnin, Lang, and Liberty~\cite{KarninLL16} developed an optimal randomized algorithm, denoted KLL, that stores just $O(\eps^{-1})$ items (with a constant probability of a too large error).
Beyond the comparison-based model, 
$q$-digest~\cite{shrivastava2004medians} 
uses an implicit binary tree over the items' universe $\mathcal{U}$ and stores a subset of the tree nodes with associated counters.
It requires space of $O(\eps^{-1}\cdot \log |\mathcal{U}|)$ and a foreknowledge of $\mathcal{U}$.
Very recently, Gupta, Singhal, and Wu developed a compressed variant of $q$-digest that fits into $O(\eps^{-1})$ memory words, by packing more tree nodes into one memory word~\cite{GuptaSW24optimaQS-non-comp}.
However, unlike SplineSketch, GK summary, KLL, and $q$-digest do not offer practical accuracy on real-world data that goes beyond their worst-case behavior,
as demonstrated in \Cref{sec:experiments} (we only include a variant of GK and KLL as $q$-digest was already shown to be outperformed by GK in~\cite{luo16_quantiles_experimental}).

Another line of work focused on the stronger relative (multiplicative) error guarantee, which 
requires that the error for an item of rank $r$ is at most $\pm\, \eps\cdot r$.
For deterministic comparison-based algorithms,
there is a gap of $O(\log \eps n)$ between the merge-and-prune algorithm using space $O(\eps^{-1}\cdot \log^3 \eps n)$~\cite{zhangwang}
and the lower bound of $\Omega(\eps^{-1}\cdot \log^2 \eps n)$~\cite{CormodeV20}.
Efficient randomized sketches with relative error have only appeared
recently, namely, ReqSketch with space $O(\eps^{-1}\cdot \log^{1.5} \eps n)$~\cite{CormodeKLTV23} %
and its ``elastic`` version achieving space close the information-theoretic lower bound of
$\Omega(\eps^{-1}\cdot \log \eps n)$~\cite{GribelyukSWY25elasticCompactors}; however, mergeability of the ``elastic'' version is open,
while ReqSketch is analyzed in the most general mergeability setting.
There is also an extension of $q$-digest to the relative error
that uses space $O(\eps^{-1}\cdot \log \eps n\cdot \log |\mathcal{U}|)$~\cite{cormode2006space};
it is not known to be optimal.
Compared to SplineSketch, the main benefit of relative error is better accuracy for the tails which, however, comes at the cost
of higher space complexity, both in theory (additional dependency on $\log n$) and practice (see e.g.~\cite{CormodeMRV21}).
Thus, relative-error sketches are incomparable to SplineSketch and other uniform-error algorithms.
Finally, we note that stronger error guarantees, such as q-error (consider e.g. in~\cite{moerkotte2014exploiting}),
cannot be achieved in the streaming or mergeability settings with sketches of sublinear size, due to standard space lower bounds from the one-way communication complexity of indexing.

For dynamic streams (i.e., with deletions), one can only achieve the uniform error guarantee using sublinear space,
and the best quantile summary supporting deletions is Dyadic CountSketch~\cite{luo16_quantiles_experimental}.
The support for deletions comes at an additional cost of a polylogarithmic dependency on the universe size $|\mathcal{U}|$,
while insertion-only sketches, including KLL or SplineSketch, are much more space-efficient.

In contrast to rank error, DDSketch~\cite{MassonRL19} and UDDSketch~\cite{EpicocoMCPM20}
provide the relative \emph{value} error (i.e., an item close to the desired quantile in the item space)
based on maintaining a suitable exponential histogram.
While this usually captures the distribution tails well, the rank error can be arbitrarily large as these sketches in fact fail to capture most of the distribution in sufficient detail, leading to overall worse accuracy than KLL, $t$-digest, or SplineSketch.
	Furthermore, the value error is not invariant to natural data transformations such as translations.

There exist more heuristic approaches that,
similarly as SplineSketch, work only for numerical data,
but do not aim for any worst-case rank guarantees (unlike SplineSketch).
Specifically, $t$-digest~\cite{dunning19-t-digest,dunning21} performs an online one-dimensional $k$-means clustering
by averaging data items into a given number of centroids, and interpolates linearly between the centroids.
A downside of this approach is that many of the items summarized by a centroid may be smaller or larger than an adjacent centroid mean, implying under- or over-estimation.
That is, the centroids are only ordered by their means,
while the represented items may be far from ordered.
In fact, one can impose high overlaps of represented items, leading to almost arbitrarily large error on adversarial instances~\cite{CormodeMRV21}.

A different application of interpolation techniques in the context of quantile estimation was presented in~\cite{SchieferCINSW23-KLLinterpolation}.
In this work the linear interpolation were combined with the KLL sketch to improve its accuracy in some settings.

Finally, MomentSketch~\cite{GanDTSB18,MitchellFH21} is possibly the most compact summary as it consists of $k$ moments and log-moments of data items,
for a small $k$ such as $k=15$. 
Upon a query, the sketch constructs a distribution with the same moments and log-moments, using the maximum entropy principle
and Chebyshev polynomials, which is a costly operation.
The consequence is that MomentSketch only works well for certain smooth distributions but, as our experiments also demonstrate,
suffers a large error in many real-world cases,
including the uniform distribution.

The aforementioned quantile sketches form a basis for per-key quantile estimation (a simultaneous estimation of quantiles for multiple substreams in the main stream), as investigated in a recent line of work~\cite{dong.etal:2024:2024ieee40thint.conf.dataeng.icdea,guo.etal:2023:proc.29thacmsigkddconf.knowl.discov.datamin.,he.etal:2023:2023ieee39thint.conf.dataeng.icdea, shahout.etal:2023:proc.acmmanag.data}.

	Additionally, we remark that the idea of approximating the data distribution using splines has been widely used in computer science before, e.g., for query optimization~\cite{neumann.michel:2008:shar.datainf.knowl.} and indexing~\cite{kipf.etal:2020:} in database systems.

\section{Description of SplineSketch}\label{sec:SplineSketchDescription}

We provide a description of our algorithm, omitting implementation details that are explained in Section~\ref{sec:implementation}.
In fact, there are many possible implementations of the sketch that still satisfy the worst-case guarantee.
Our sketch works for any numerical input, i.e., items are integers or floating-point numbers.
We first present the algorithm assuming there are no items of high frequency, namely appearing more than $n/k$ times in the input
after inserting any number $n$ of items into SplineSketch of size $k$.
We show how to remove this assumption in Section~\ref{sec:theory}. %

\paragraph{Buckets.}
Fix an integer parameter $k \ge 6$ (a technical assumption required in the analysis).
The sketch primarily consists of $k$ distinct thresholds $\tau_1 < \tau_2 < \cdots < \tau_{k}$ and $k$ counters $b_1, \ldots, b_k$, where $b_i$ is our estimate of the number of items that lie in the interval $(\tau_{i-1}, \tau_i]$; we define $\tau_{0} = -\infty$.
We call these intervals \emph{buckets} and refer to a bucket by the index of its right threshold, so $(\tau_{i-1}, \tau_i]$ is called "bucket~$i$" or the $i$-th bucket.
We also frequently refer to $b_i$ as the \emph{size of bucket} $i$ and we define the \emph{length of bucket} $i$ as $\ell_i := \tau_i - \tau_{i-1}$.

Throughout the execution of the algorithm, we ensure that no bucket is empty
and that $\tau_1$ and $\tau_{k}$ are always set to be the minimum and the maximum item in the stream, respectively.
The latter makes the first bucket special as its counter only stores the frequency of $\tau_1$.

We also maintain a buffer of size $O(k)$ that we use for faster processing of incoming items, i.e.,
low average update time,
and one auxiliary bit-array of length $k$ for ``protected thresholds'' (Section~\ref{sec:consolidate}).
After processing the whole input, the buffer is merged into the buckets and the auxiliary bit-array can be discarded.
We ensure that all bucket counters plus the number of items in the buffer sum up to $n$.

\paragraph{Insertion and initial buckets.} When a new element arrives, we first put it in the buffer. 
When the buffer gets full, we execute the consolidate method (Section~\ref{sec:consolidate})
that may change bucket thresholds and merges the buffer into buckets. 
The only exception is the first time when the buffer gets full.
In that case, we simply sort it, pick the first and then every $\approx n / k$-th item, selecting $k$ items as the initial thresholds, and set counters $b_i$ to exactly the number of items in the corresponding buckets. 
Since no item has frequency greater than $n/k$, these selected items are distinct and every bucket is nonempty.
For the pseudocode of the bucket initialization, see Algorithm~\ref{alg:initBuckets}.

\begin{algorithm}[t] 
	\SetKwInOut{Input}{Input}
	\SetKwInOut{Output}{Output}
	\caption{Buckets Initialization}
	\label{alg:initBuckets}
	\Input{Buffer of size $n$, parameter $k$ (number of buckets)}
	\Output{Initialized $k$ buckets with thresholds $\tau_1, \dots, \tau_k$ and counters $b_1, \dots, b_k$}
	Sort all $n$ items in the buffer in ascending order\;
	For $i = 0, \dots, k -1$, pick the item at index $\min\{\lceil i\cdot (n-1) / (k-1) \rceil, n-1\}$ from the sorted buffer to form $k$ thresholds $\tau_1, \dots, \tau_k$\;
	For each bucket $(\tau_{i-1}, \tau_i]$, set the counter $b_i$ to the exact number of items in the buffer that fall into this interval\;
	Return the initialized buckets with thresholds $\tau_1, \dots, \tau_k$ and counters $b_1, \dots, b_k$\;
\end{algorithm}

\paragraph{Rank query.}
For a rank query at a threshold $\tau_i$, we return the sum of the counters up to the $i$-th, i.e., $\sum_{j=1}^i b_j$.
For $x$ strictly inside the $i$-th bucket, i.e. $\tau_{i-1} < x < \tau_i$, we use an interpolation method to estimate the rank of $x$; note that this works no matter whether $x$ appeared on input or not.
The purpose of using a suitable interpolation is to get much better error in practice,
while not affecting the worst-case bounds (see Section~\ref{sec:theory}).

Specifically, we use a monotone variant of the Piecewise Cubic Hermite Interpolating Polynomial (PCHIP)~\cite{fritsch.butland:1984:siamj.sci.andstat.comput., fritsch1980monotone}, which is an interpolation method that constructs a continuous function by fitting cubic polynomials between data points. 
Unlike standard cubic spline interpolation, PCHIP adjusts the first derivatives at each data point to preserve the shape and monotonicity of the original data.
For an introduction in greater detail, including the formulas for the interpolation, see \cite{Wong2020splines}.
Since the ranks of the input items are non-decreasing, preservation of monotonicity makes PCHIP suitable for our application.

We remark that computing this interpolation requires only the knowledge of $\tau_{i-2}, \cdots, \tau_{i+1}$ together with corresponding prefix sums of bucket counters.
Therefore, provided that we maintain the prefix sums, the computation of the interpolation takes $O(1)$ time.
The interpolation does not extrapolate, i.e., for $x < \tau_1$, it returns 0, and for $x > \tau_k$, it returns $\sum_{j=1}^k b_j$.
Since one has to find the right bucket by binary search, one query takes time $O(\log k)$  provided that the prefix sums
are precomputed.

In principle, our approach can be used with any interpolation method (ideally one that preserves monotonicity).
One popular alternative would be to use kernel smoothers.
However, this particular method would introduce additional error on the thresholds, which would be harder to control.
We leave the exploration of alternative interpolation methods for future work.
For an extensive treatment of different interpolation methods, see \cite{fan.gijbels:1996:}.

\subsection{Consolidating Buckets}\label{sec:consolidate}

The main procedure of the sketch is \consolidate. %
This method merges the buffer into the current buckets, by counting the new number of items in each bucket, and updates the counters accordingly.
It also changes the buckets based on the new values of the counters, possibly removing some of the thresholds and creating new ones.
Finally, the buffer gets emptied.
The crux of the consolidate method is how to modify the thresholds, which we describe in the remainder of this subsection.
For the pseudocode of \consolidate, see \Cref{alg:consolidate}.

\begin{figure}[t]
	\centering
	\includegraphics[width=0.5\linewidth]{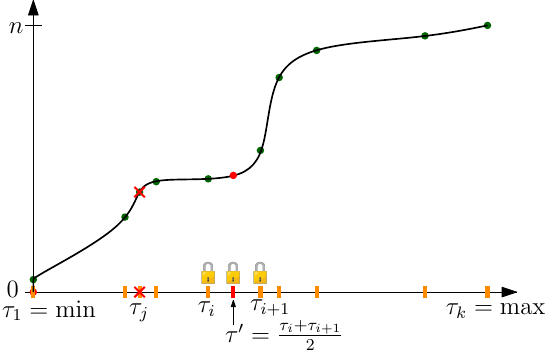}
	\caption{Illustration of joining a pair of buckets and splitting one bucket (depicted similarly as in Figure~\ref{fig:sketch}).
		The join removes threshold $\tau_j$ and merges (sums up) the two counters of the joined buckets.
		The split takes the $i$-th bucket $(\tau_i, \tau_{i+1}]$ and inserts a new threshold $\tau'_{i+1} = (\tau_i + \tau_{i+1})/2$ in the middle,
		setting the counters of the two resulting buckets $(\tau_i, \tau'_{i+1}]$ and $(\tau'_{i+1}, \tau_{i+1}]$ using the interpolation (solid black curve).
		The thresholds $\tau_i, \tau'_{i+1}$, and $\tau_{i+1}$ get protected from joining, as depicted by the locks,
		and this protection is removed at the end of the epoch.\vspace{-0.3cm}}
	\label{fig:splitJoin}
\end{figure}

\paragraph{Joining and splitting buckets.}
We have two basic operations that change bucket thresholds, join and split (Figure~\ref{fig:splitJoin}). 
A join takes two neighboring buckets $(\tau_{i-1}, \tau_i]$ and $(\tau_{i}, \tau_{i+1}]$ for $i > 1$ and turns them into
a single bucket $(\tau_{i-1}, \tau_{i+1}]$ with its counter set to $b_i + b_{i+1}$.
Note that after joining, buckets are (implicitly) renumbered.
That is, we actually set new thresholds $\tau'$ and counters $b'$ as $\tau'_i = \tau_{i+1}$, $b'_i = b_i + b_{i+1}$, and for any $j > i$, $\tau'_j = \tau_{j+1}$ and $b'_j = b_{j+1}$.
The buckets up to $\tau_{i-1}$ do not change.

When splitting bucket $(\tau_{i-1}, \tau_i]$ for $i > 1$, we replace it by two buckets $(\tau_{i-1}, \frac{\tau_{i-1}+\tau_i}{2}]$ and $(\frac{\tau_{i-1}+\tau_i}{2}, \tau_{i}]$. The bucket counters are set as follows: We use the query method
to get estimated ranks $\hat{r}$ at $\tau_{i-1}, \frac{\tau_{i-1}+\tau_i}{2},\tau_i$;
for more precise estimates, we use the query method on the buckets and buffer in their state just before performing \consolidate. We then set the counters for the two new buckets to $\hat{r}(\frac{\tau_{i-1}+\tau_i}{2}) - \hat{r}(\tau_{i-1})$ and $\hat{r}(\tau_i) - \hat{r}(\frac{\tau_{i-1}+\tau_i}{2})$, respectively.
Other buckets are again implicitly renumbered.

Observe that every split increases and every join decreases the number of buckets by one.
Thus, after every split we join some other pair of buckets, which ensures there are always $k$ buckets
(unless a change of $k$ is desired).
The crux of the algorithm lies in selecting which buckets to split and join, which we describe below.

\paragraph{Protected thresholds and epochs.}
After splitting a bucket, the estimated rank, and therefore the estimation error, at its right threshold stays the same.
However, the error at the newly created threshold may be bigger, in the worst case by the size of the original bucket counter.
If we repeatedly join and split the newly created buckets with their neighbors, we might let the error increase too much.
With the small overhead of $k$ bits when processing the input, we show how to avoid too frequent splits and joins at a single location and prevent the associated error accumulation.
To this end, we introduce the notion of \emph{protected thresholds}.

Specifically, we divide the input into \emph{epochs}. The first epoch has length $\Theta(k)$ and each successive epoch starting at time $t$ has length $0.25 \cdot t$ (we use in the analysis that the constant factor is at most $0.25$). When we split a bucket $(\tau_{i-1}, \tau_i]$, we mark all three thresholds $\tau_{i-1}, \tau'_i, \tau_i$ as protected, where $\tau'_i = (\tau_{i-1} + \tau_i)/2$ is the new threshold.
We store a bit vector of length $k$, with the $i$-th bit set to 1 when the $i$-th threshold is protected.
When a new epoch starts, we reset all thresholds to the default unprotected state
by zeroing the bit vector.
As specified later in Definition~\ref{def:joinableBuckets}, we never join buckets that have a protected threshold between them.

Note that we may still split a single bucket multiple times during one epoch,
or even during one call of \consolidate, namely, if the bucket receives many new items. Similarly, a single bucket may be joined multiple times during one consolidation.
However, a combination of splitting and then joining is not allowed -- a bucket which was created by a split cannot be joined with another bucket in the same epoch.

\paragraph{Handling new minima and maxima.}

If the buffer contains a new minimum $\tau'_1$ (an item smaller than the previous minimum $\tau_1$), we add a new bucket $(-\infty, \tau'_1]$ before all of the existing
buckets, set its counter to the frequency of $\tau'_1$ in the buffer. We also increase the counter
of the bucket $(\tau'_1, \tau_1]$ by the number of buffer items in this interval.
A new maximum $\tau'_{k}$ is handled similarly, by adding a new bucket $(\tau_{k}, \tau'_{k}]$ and setting
its counter to the number of buffer items in this interval.
For each of these at most two new buckets, we perform a join of another pair, selected according to the rules described below.
Note that these new buckets may be later split during \consolidate in which they were created, apart from the new first bucket $(-\infty, \tau'_1]$.

\paragraph{Maintaining bounds on bucket counters and small heuristic error.}

We maintain two invariants, the first of which is used to prove worst-case bounds while the second ensures high accuracy in practice.

In order to obtain a low worst-case uniform error, we need to ensure that no bucket has a too large counter.
To this end, we require that all bucket counters are bounded by $O(n/k)$, specifically
\begin{equation}\label{eqn:bucketBound}
	\forall i = 1, \dots, k:\, b_i \leq C_b \cdot n/k	
\end{equation}
for some constant $C_b > 1$
that is determined in the analysis in Appendix~\ref{sec:analysis}.

Secondly, we define a \emph{heuristic error} for each bucket, which is intended to capture ``hard'' parts of the input,
where the input distribution changes a lot, namely, where the second derivative of the empirical CDF is large.
For such buckets, the PCHIP interpolation may have larger error than for buckets with a low second derivative.
Since the empirical PDF (i.e., first derivative of the CDF) in the bucket $i$ is $b_i / \ell_i$, we use the following heuristic error for bucket $i > 1$, which approximates the second derivative of the empirical CDF, normalized by length squared:
\begin{equation}\label{eqn:heurError-der2}
    \max\left(
        \frac{\left|\frac{b_i}{\ell_i} - \frac{b_{i-1}}{\ell_{i-1}}\right|}{\ell_i + \ell_{i-1}}
        \,,\,
        \frac{\left|\frac{b_{i+1}}{\ell_{i+1}} - \frac{b_i}{\ell_i}\right|}{\ell_{i+1} + \ell_i}
    \right)
    \cdot \ell_i^2.
\end{equation}
For the edge cases $i=2$ and $i=k$, we make the following assumptions: $b_{k+1} = 0$, 
$\ell_{1} = \ell_2$, and $\ell_{k+1} = \ell_{k}$; that is, we add a ``virtual'' bucket $k+1$
with zero counter and a length equal to the length of the adjacent (real) bucket, and we set the length of bucket 1 to the length of bucket 2.

We now describe how we use the bucket bound~\eqref{eqn:bucketBound} and heuristic error~\eqref{eqn:heurError-der2} when performing the splits and joins. 
Each time we split a bucket $i$, we choose a pair of adjacent buckets that do not  contain buckets created from splitting $i$, and we join the chosen pair; thus, the total number of buckets stays the same.
The pair is selected so that the bucket resulting from joining has the lowest heuristic error of all \emph{joinable pairs}, defined as follows.
In a nutshell, the bucket resulting from the join should not violate the bucket size bound~\eqref{eqn:bucketBound} or be relatively close to it.
Moreover, none of the two buckets should be subject to a split recently,
which we implement using the protected thresholds.
\begin{definition}\label{def:joinableBuckets}
    We say that a pair $i, i+1$ of adjacent buckets is \emph{joinable} if both of the following conditions hold:
    \begin{enumerate}[label=(\roman*)]
        \item the counter of the bucket resulting from the join satisfies $b_i + b_{i+1} \le 0.75\cdot C_b \cdot n/k$, and
        \item threshold $\tau_i$ is not \emph{protected}.
    \end{enumerate}
\end{definition}

We perform splits as follows. First, we take all buckets that exceed the size bound~\eqref{eqn:bucketBound}. We split them one by one, joining a joinable pair according to \Cref{def:joinableBuckets} for each split; this joinable pair is selected so that it has the lowest possible heuristic error~\eqref{eqn:heurError-der2} after temporarily performing the join (i.e., the heuristic error of a pair is computed by joining this pair without changing other buckets, then using~\eqref{eqn:heurError-der2} for the resulting bucket, and reverting the join).
In Lemma~\ref{lem:existsBucketToMerge}, we show that at least one pair of adjacent buckets is joinable.
We also join at most two further pairs of joinable buckets if the buffer contains a new minimum or maximum, each creating a new bucket.
Note that a bucket may need to be split repeatedly during one call of \consolidate, namely, if it exceeds the bucket bound~\eqref{eqn:bucketBound} substantially.

After all these necessary splits and joins, we consider a bucket $i$ with a high heuristic error~\eqref{eqn:heurError-der2}, and we split it if its error is more than $\gamma$ times greater than the lowest heuristic error of a bucket resulting from joining a joinable pair, for a parameter $\gamma > 1$ (say, $\gamma = 1.5$). Furthermore, we do not perform a split of a bucket $i$ due to the heuristic error if there are less than $k/3 + 2$ pairs of adjacent buckets that can be joined, which
is needed to prove that there is a joinable pair of buckets; see Lemma~\ref{lem:existsBucketToMerge}.
Finally, if the counter of bucket $i$ is substantially below the bound, say, $b_i \le 2n/k$,
we also do not split bucket $i$, avoiding splitting almost empty buckets. %
We summarize the conditions for splitting:

\begin{definition}\label{def:splittableBuckets}
	We say that a bucket $i$ is \emph{splittable} if one of the following conditions holds:
	\begin{enumerate}[label=(\roman*)]
		\item bucket $i$ violates the bucket bound~\eqref{eqn:bucketBound}, i.e., $b_i > C_b\cdot n/k$, or %
		\item there are at least $k/3 + 2$ joinable pairs of buckets,
		$b_i > 2n/k$, and the heuristic error~\eqref{eqn:heurError-der2} of bucket $i$ is $\gamma > 1$ times larger than the heuristic error of a joinable pair of buckets (after joining) not overlapping bucket $i$.
	\end{enumerate}
\end{definition}

We repeat finding a splittable bucket $i$ and a joinable pair of buckets not overlapping bucket $i$,
performing the split and join as long as there are splittable buckets.
Since every split increases the number of protected thresholds, we know that this process terminates after at most $O(k)$ iterations.
The next lemma shows the algorithm is well-defined\footnote{The proof requires that if the buffer size $\Theta(k)$ is larger than $(C_b/2)\cdot n/k$, then the buffer is processed in batches of size $(C_b/2)\cdot n/k$.
See Appendix~\ref{sec:analysis} for more details.}.
	\begin{lemma}[See \Cref{lem:existsBucketToMergeStreaming,lem:existsBucketToMerge} in Appendix~\ref{sec:analysis} for proofs in the streaming and mergeability settings, respectively.]\label{lem:existsBucketToMergeMainBody}
	After initialization of buckets by Algorithm~\ref{alg:initBuckets},
	SplineSketch always contains a joinable pair of adjacent buckets, according to Definition~\ref{def:joinableBuckets}.
	\end{lemma}

This completes the description of \consolidate; see Algorithm~\ref{alg:consolidate} for a pseudocode.

\begin{algorithm}[h]
	\SetKwInOut{Input}{Input}
	\SetKwInOut{Output}{Output}
	\caption{\consolidate}
	\label{alg:consolidate}
	\Input{$k$ buckets and buffer of new items
	} %
	\Output{Updated buckets $\tau_1, \dots, \tau_k$ with counters $b_1, \dots, b_k$, protection bit vector, and epoch end $T$}

	\If{$n > T$}{ 
		Reset protection bit vector, and set $T := 1.25 \cdot T$ \tcp*{set the next epoch end $T$}
	}
	Compute the PCHIP interpolation over the buckets\;
	\If{new minimum or maximum found in the buffer}{
		Add a new boundary bucket for minimum/maximum\;
		For each new bucket, join a pair of joinable buckets $j,j+1$ that has the lowest heuristic error after temporarily performing the join\;
	}
	Incorporate buffer items into each bucket's counter (keep the buffer for executing splits)\;
	\While{exists a splittable bucket $(\tau_{i-1}, \tau_i]$, first processing those exceeding the bound \eqref{eqn:bucketBound}}{
	   Split bucket $i$ into $(\tau_{i-1}, \tau']$ and $(\tau', \tau_i]$ for $\tau' = (\tau_i + \tau_{i-1})/2$, using the interpolation and buffer to determine the values of new counters\;
	   Mark thresholds $\tau_{i-1}, \tau', \tau_i$ as protected\;
	   Join a pair of joinable buckets $j,j+1$ that does not overlap with $(\tau_{i-1}, \tau']$ and $(\tau', \tau_i]$
	   and has the lowest heuristic error after temporarily performing the join\;
	}
	Empty the buffer\;
	\end{algorithm}

\paragraph{Resizing the sketch during \consolidate.}
Bucket consolidation allows to change the number $k$ of buckets,
making the sketch larger or smaller as desired.
Changing from $k$ to $k'$ buckets is straightforward: 
If we increase the sketch size, i.e., $k' > k$, we perform $k' - k$ splits
without any join, first dealing with buckets exceeding bound~\eqref{eqn:bucketBound} and then in the order of the heuristic error.
On the other hand, if $k' < k$, after performing the necessary splits due to exceeding the bound (and the corresponding joins),
we execute $k - k'$ joins
in the order of their heuristic errors after joining, on bucket pairs satisfying Definition~\ref{def:joinableBuckets} and without performing any splits.
These separate splits or joins may be followed by performing \consolidate in a normal way, where we always join after splitting. %
If $k'$ is substantially different from $k$, say $|k' - k| > 0.25\cdot k$, we also %
reset the protection bit vector (otherwise, it may happen that there are no joinable pairs of buckets after resizing).

\subsection{Merging Two SplineSketches}\label{sec:mergeOp}

We describe how to merge two SplineSketches, $S_1$ and $S_2$, with $k_1$ and $k_2$ buckets, respectively, into a single sketch $S$ for the combined datasets.
In the analysis, we assume that  $k_1 = k_2$ but the merge operation works even for $k_1 \neq k_2$.

First, we create a buffer that contains all the items from the buffers of the two sketches we are merging.
Second, we take the union of all the thresholds $\tau$ from both sketches. After removing duplicates, this gives up to $k_1+k_2$ thresholds. For each of these thresholds $\tau$, we query for its rank in both input sketches and add these two ranks to get the rank of $\tau$ in the merged sketch $S$. We set the counter of each bucket to be the difference between the ranks of the endpoints.
Third, the new sketch inherits the bit vector for threshold protection from the source sketch $S_i$ that summarizes more items; that is, any threshold taken from $S_i$ remains protected if it was protected in $S_i$. However, if the total input size of the new sketch is larger than the epoch end for $S_i$, we set all thresholds to the unprotected state, starting a new epoch (in particular, if $S_1$ and $S_2$ summarize a similar number of items, no threshold is protected after merging).

Finally, we run \consolidate that reduces the number of buckets to $k_i$ and also merges the buffer to buckets; if $k_1\neq k_2$, then the index $i$ is chosen so that the source sketch $S_i$ summarizes more items, similarly as for threshold protection (alternatively, one can also choose the larger of $k_1$ and $k_2$). The reduction to $k_i$ buckets is simply done by repeatedly taking a joinable pair of adjacent buckets whose joining would result in a bucket with the lowest heuristic error and joining it without any split until we have $k_i$ buckets.
Then we continue with \consolidate\ similarly as in the streaming setting,
also merging the buffer into buckets if it is full or overflows.

\subsection{Theoretical Properties}\label{sec:theory}

	We now state the error bounds, focusing on the streaming setting first.
The analysis assumes that no element appears more than $O(t/k)$ times after processing any $t$ items.
In order to lift this assumption, we can run a suitable heavy-hitter streaming algorithm, such as the Misra-Gries (MG) sketch~\cite{misra1982finding} of size $O(k)$ in parallel.
We describe how to effectively combine MF with SplineSketch in Section~\ref{sec:highFrequencyItems-Impl}.

		\begin{restatable}[Error bound in the streaming setting.]{theorem}{mainthmStreaming}\label{thm:mainStreaming}
			Suppose that SplineSketch, with $k$ buckets and with heavy-hitter filtering by MG, processes an input of $n$ numerical items in a single pass.
			Then the resulting SplineSketch has worst-case rank error of $O(\log (\alpha) \cdot n / k)$.
		\end{restatable}

	Next, we provide an analysis of the merge operation when the two sketches merged summarize
		a similar number of items, up to a constant factor; we call this the \emph{balanced mergeability setting}.
		Here, we get an additional $O(\log (n/k)\cdot n/k)$ error due to taking the union of thresholds when merging; it holds that $\log (n/k) \in O(\log \alpha)$ provided that $\alpha > (n/k)^c$ for a constant $c > 0$, which is typically the case but may fail when the number of distinct items is very small.
		For simplicity, we assume that during \consolidate\ after a merge operation, there are $O(1)$ splits of a single bucket due to the heuristic error;
		more precisely, for a bucket $(\tau_i, \tau_{i+1}]$ resulting from the union of thresholds of the two source sketches and the joins to get back to $k$ thresholds,
		we bound the number of such splits of buckets inside the interval $(\tau_i, \tau_{i+1}]$.

		\begin{restatable}[Error bound in the balanced mergeability setting.]{theorem}{mainthmBalancedMergeability}\label{thm:balancedMergeability}
			Suppose that we build a SplineSketch, with heavy-hitter filtering by MG, using a sequence of pairwise merge operations executed on $n$ data items such that every SplineSketch resulting from merging consists of $k$ buckets
			and in every merge operation, the two sketches merged summarize the same number of items up to a factor of $c \in \Theta(1)$.
			Furthermore, suppose that any bucket is split $O(1)$ times due to the heuristic error during \consolidate\ performed after each merge operation.
			Then the resulting SplineSketch has a worst-case rank error of $O((\log (\alpha) + \log (n/k)) \cdot n / k)$.
		\end{restatable}

We note that state-of-the-art quantile sketches achieve better worst-case accuracy-space trade-offs.
Namely, KLL using space of $O(k)$ memory words has worst-case rank error $O(n/k)$ with constant probability~\cite{KarninLL16}.
The deterministic Greenwald-Khanna (GK) sketch guarantees rank error $O(n/k)$ but with worst-case space $O(k\cdot \log (n/k))$~\cite{greenwald2001space}; the $ \log (n/k)$ factor is in general incomparable to the $\log (\alpha)$ in the error for SplineSketch, albeit the latter is typically larger.
However, both KLL and GK achieve significantly worse error than SplineSketch in practice, as we show in Section~\ref{sec:experiments}.
As for mergeability, KLL is fully mergeable while retaining the guarantees~\cite{KarninLL16} but GK is known to be only one-way mergeable~\cite{AgarwalCHPWY13}.
Other quantile sketches have worse rank error or do not provide a bounded rank error; see Section~\ref{sec:related}.

Our analysis is very flexible with respect to the individual components and parameters of the sketch,
and allows for certain changes in constant factors (such as the value of $C_b$ in Eqn.~\ref{eqn:bucketBound}).
It is also completely independent of how the heuristic error is computed and, in particular, of the value of $\gamma$ in Definition~\ref{def:splittableBuckets}; an exception is the assumption in Theorem~\ref{thm:balancedMergeability}.
The epochs can also be adjusted; however, our analysis requires that their lengths increase geometrically (the increase factor affects other constant factors).
As for the interpolation, it can be changed to linear or another monotone interpolation that is exact at the thresholds.
The only other property that the interpolation method must satisfy is that, in the middle of a non-empty bucket, the interpolation is bounded away from the minimum and maximum estimated values inside the bucket;
that is, a bucket counter is divided somewhat evenly during a split (a property similar to Lemma~\ref{lem:cubicPolyMonoMiddleBound} in Appendix~\ref{app:monotoneCubicPolyMiddleBounds}).

Regarding the effect of resizing the sketch from $k'$ buckets to $k''$ buckets on the guarantees, 
the error bound of $O(\log (\alpha) \cdot n / k)$ holds with $k=\min\{k', k''\}$ after performing the resizing,
provided that no resizing was done before.
To the best of our knowledge, resizability is not addressed in the literature on quantile summaries yet.

As for the time complexities, the time for one rank or quantile query is $O(\log k)$, as derived at the beginning of Section~\ref{sec:SplineSketchDescription}.
The amortized update time can be made $O(\log k)$ by a heap-based implementation of \consolidate, as explained in Section~\ref{sec:consolidateImpl}.
Finally, the merge time is $O(k\cdot \log k)$, again for the heap-based \consolidate.
These theoretical time complexities are asymptotically the same as for KLL, GK, and $t$-digest,
albeit particular implementations of the sketches may not satisfy those theoretical bounds.

\section{Implementation Details}
\label{sec:implementation}

Here, we explain the details of our SplineSketch implementation (in Python and Java) and elaborate
on the update time. %
First, we naturally use the \texttt{double} data type (64-bit floating point numbers) for thresholds. This, however, 
implies some limitations, such as worse accuracy for some inputs with many significant digits (e.g., the SOSD benchmark~\cite{sosd,MarcusKRSMK0K20sosd}), which is due to the fact that many different items get rounded to the 
same \texttt{double} value. This can be addressed by changing the underlying data type for thresholds, e.g., to 128-bit floating-point numbers or to unsigned 64-bit integers.
Note that the internal PCHIP interpolation would still need to use floating-point numbers.

\paragraph{Quantile queries.} %
So far, we focused on answering rank queries.
In order to implement a quantile query, i.e., returning an estimated $r$-th smallest number for a given $r$,
we invert the interpolation, and query the inverse function at $r$.
This inverse exists since buckets are non-empty and the PCHIP interpolation is thus strictly increasing from $\tau_1$ to $\tau_k$.
The inverse function can be evaluated by first finding the bucket in which the inverse of $r$ lies (which can be done using binary search, assuming we maintain the prefix sums of the bucket counts). Then, inside the corresponding bucket, the value can be found numerically, for example, using Newton's method or even just bisection.
Note that our algorithm, similarly as $q$-digest~\cite{shrivastava2004medians}, may return a number that did not appear in the stream, unlike
comparison-based methods such as the KLL sketch~\cite{KarninLL16} or the Greenwald-Khanna sketch~\cite{greenwald2001space}.

\subsection{Implementations of Bucket Consolidation}
\label{sec:consolidateImpl}

We describe two particular implementations of \consolidate. The first achieves a low amortized update time,
while the second is easier to implement and works well in practice.

\vspace{-0.2cm}

\paragraph{Heap-based \consolidate.} 
To obtain an asymptotically fast \consolidate, one can use a combination of doubly-linked lists for maintaining buckets and two heaps, sorted by the bucket size and heuristic error, respectively.
This approach is formalized in the following lemma. %
Since the buffer size is $\Theta(k)$, we get $O(\log k)$ amortized complexity of an update with heap-based \consolidate\ as an immediate corollary.

\begin{restatable}{lemma}{heapBasedConsolidate}
	\label{lem:heapBasedConsolidate}
	\consolidate can be implemented in $O(k \log k)$ time.
\end{restatable}

\begin{proof}
	We maintain buckets in a doubly linked list $L$. 
	We use a queue $Q_S$ to keep track of buckets that need to be split due to violating bucket bound.
	We also use two heaps: a max-heap $H_S$ to keep track of buckets that might get split due to the heuristic error~\eqref{eqn:heurError-der2}, ordered by their heuristic errors, and a min-heap $H_J$ for joinable pairs of buckets, ordered by their heuristic errors after temporarily joining the pair.
	In both heaps, every node contains a pointer to the corresponding bucket (or pair of buckets) in $L$.
	Conversely, every bucket in $L$ contains a (constant size) list of pointers to nodes of $H_S$ and $H_J$ containing that bucket. 
	Both heaps can be initialized in $O(k)$ time.
	
	After every split and join performed by \consolidate, we update $Q_S, H_S$, and $H_J$ according to the new bucket sizes.
	Since every split and join affects only the sizes and heuristic errors of a constant number of adjacent buckets, the update can be done in $O(1)$ time using the pointers that we maintain.
	The modifications of the keys and deletions of arbitrary nodes in the heaps can be done in $O(\log k)$ time using standard decrease/increase key operation.
	
	Every split introduces a constant number of new protected thresholds and every join is accompanied with a corresponding split.
	Since the number of possible protected thresholds is bounded by $k$, the total number of splits and joins performed by \consolidate is also $O(k)$.

	Finally, note that the buffer items can be merged into buckets in $O(k\log k)$ time.
	Indeed, we sort the buffer and start adding items from the smallest to the largest,
	maintaining the index of the bucket containing the current buffer item;
	if the processing of the buffer is done in batches (as described in Appendix~\ref{sec:analysis}), we just perform the necessary splits after each batch.

	Therefore, the total time complexity of the heap-based implementation of \consolidate is $O(k \log k)$.
\end{proof}

\paragraph{Iteration-based \consolidate.}
In our implementation, we use a different \consolidate that does not require linked lists or heaps.
We process the buckets in iterations such that, in each iteration, any bucket is split or joined at most once (but not both).
Since a particular bucket may need to be split or joined multiple times during \consolidate, we continue the process for multiple iterations,
until the last iteration does not change any threshold.
Throughout this process, we keep the buffer and the old buckets
for more accurate setting of counters for the new buckets.

In each iteration, we first collect the buckets with size exceeding the bound~\eqref{eqn:bucketBound} into a set $E$.
We also keep a list of splittable buckets (Definition~\ref{def:splittableBuckets}), sorted in non-increasing order by their heuristic error,
and a list of joinable pairs of buckets (Definition~\ref{def:joinableBuckets}), sorted in non-decreasing order by the resulting heuristic error after temporarily joining them.
We select $|E|$ joinable pairs according to their order that do not overlap buckets in $E$ or a previously selected joinable pairs.
Next, from the remaining buckets (not yet split or joined), we consider the splittable bucket with the largest heuristic error $h_{\mathrm{split}}$ (if any)
and the joinable pair with the lowest heuristic error $h_{\mathrm{join}}$ and if $h_{\mathrm{split}} > \gamma\cdot h_{\mathrm{join}}$ for $\gamma > 1$,
then we perform the split and join of these buckets;
in our implementation, we use $\gamma = 1.5$.
After collecting non-overlapping sets of buckets to split and join, we perform these splits and joins
in one pass over the buckets, using the buffer and interpolation over the buckets in their state before executing \consolidate
to estimate ranks in the middle of split buckets;
the usage of former buckets is important to decrease the error accumulation.
We refer to the prototypes referenced in Section~\ref{sec:experiments} for details.

Finally,
we remark that the multiplicative constant $C_b$ in the bucket bound~\eqref{eqn:bucketBound} resulting from the analysis in~Appendix~\ref{sec:analysis} is too large for practical usage,
and we just use $C_b = 3$, while increasing $C_b$ during an epoch if needed and resetting it to $C_b = 3$ at each epoch end. That is, if there are no joinable pairs when a bucket must be split due to exceeding the bound, we double $C_b$ and reset the iteration.
We note that such increases of $C_b$ happen rarely in practice, and provably stop after several iterations (when exceeding the constant $C_b$ from the analysis).
Therefore, this trick preserves the theoretical properties.

\subsection{Handling Frequent Items}
\label{sec:highFrequencyItems-Impl}

The theoretical error guarantees of SplineSketch, outlined in Section~\ref{sec:theory}, require 
a method for filtering items with frequency $\Omega(n/k)$.
In Section~\ref{sec:expSkewed} we show that this filtering is also important in practice, namely we provide inputs that include frequent elements and that are hard for vanilla SplineSketch.
This issue can be resolved by using a heavy hitter sketch such as the Misra-Gries (MG) summary~\cite{misra1982finding}.
However, using the MG comes at a cost of increased update, merge, and query times,
and therefore, we provide two versions of SplineSketch, with and without MG.

\vspace{-0.2cm}

\paragraph{SplineSketch with MG}

The MG sketch stores up to $k-1$ items, each with a counter $c_x$. 
When a new element $x$ arrives,
if $x$ is already in the sketch, its counter is incremented.
Otherwise, if $x$ is not in the sketch, we add it with a counter of $1$.
If adding a new item causes the sketch to reach its capacity of $k$ items, all counters are decremented.
Any item with a counter that hits zero is then removed.
It can be shown that the counters in the sketch underestimate the true frequencies by at most $n/k$.
In particular, any element with a frequency greater than $n/k$ is present in the sketch~\cite{misra1982finding}.
The MG sketch is also fully mergeable~\cite{AgarwalCHPWY13}.

In our implementation, each item $x$ in the MG sketch has an additional counter, $\hat{c}_x$, which is never decreased.
It represents the number of times $x$ has been inserted since the beginning of the stream or its last removal from the MG sketch.
During \consolidate, the MG sketch processes a buffer of new items in two stages.
First, we filter out buffer items that are present in MG and increase their sketch counters by the buffer frequencies.
For the remaining distinct items in the buffer, we process them one by one, and add them to the MG sketch with both their counters ($c_x$ and $\hat{c}_x$) set to their buffer frequency.
If at any point the MG sketch becomes full, we decrease the $c_x$ counters of all items in the MG sketch by the smallest counter.
We then remove all items with $c_x = 0$ and add them back into the buffer with frequency $\hat{c}_x$.
Finally, the items that remain in the buffer are processed by standard \consolidate.

Last, after processing the whole input, the extra space used by the MG sketch is compressed by moving all the items $x$ with frequency $\hat{c}_x < n/(2\cdot k)$
from the MG sketch to the buckets (with frequency $\hat{c}_x$). 
Additionally, if $\ell$ is the number of remaining items with $\hat{c}_x \ge n/(2\cdot k)$,
we resize SplineSketch so that there are $\max\{k - \ell, k / 2, 6\}$ buckets.
Thus, if $\ell \le k/2$, there will be $k$ buckets and MG items in total.
Based on the standard analysis of the Misra-Gries algorithm~\cite{misra1982finding}, each input element is inserted into the SplineSketch at most $n/k$ times, and more strongly,
at any time $t$ (i.e., after processing $t$ items), the frequency of any items inserted into SplineSketch is at most $t/k$.
To process a rank query $y$ in the MG sketch, we sum the frequencies $\hat{c}_x$ for all items $x \leq y$.
Since the values $\hat{c}_x$ are counted exactly, this sum does not increase the error of our rank estimate.
We then add the result of the SplineSketch query for $y$, i.e., the $y$'s rank estimate based on values stored in the buckets and the buffer.

\vspace{-0.2cm}

\paragraph{SplineSketch without MG}
For applications in which item frequencies are relatively small, i.e., below $n/k$,
we also implement a faster and simpler version without heavy-hitter filtering.
Still, this version tries to apply a heuristic detection of buckets with high-frequency items.
Namely, we set a lower limit on the relative length of a bucket, depending on numerical precision $\delta$ of data.
That is, a bucket $(\tau_i, \tau_{i+1}]$ must have length $\tau_{i+1} - \tau_i \ge \delta'\cdot \max\{|\tau_i|, |\tau_{i+1}|, \zeta\}$,
where $\delta' > \delta$ is a parameter and
$\zeta >0$ is the smallest non-zero absolute value of an input item ($\zeta$ is intended to avoid too small buckets around 0).
In our prototype implementations with the \texttt{double} data type, we set $\delta' = 10^{-8}$.
We do not split a bucket that would violate this relative length bound, even if it does not satisfy the size bound~\eqref{eqn:bucketBound}; this way, we avoid too many splits of a bucket with a frequent item.

We also deal with repeated thresholds when creating initial buckets in Algorithm~\ref{alg:initBuckets},
by removing thresholds that would violate the bucket length bound.
This, however, results in a sketch with less than $k$ thresholds.
We thus replace the removed thresholds by new thresholds relatively close to the previous threshold so that the resulting buckets would have length close to the relative length bound and that we end up with $k$ thresholds.
This in fact heuristically guarantees that a frequent item gets its own bucket
if it is frequent from the beginning, i.e., when creating the initial thresholds.

\section{Experimental Evaluation}\label{sec:experiments}

We have implemented SplineSketch as a prototype in Python and Java and 
compared the Java version with state-of-the-art quantile sketches for the uniform rank error
with an available implementation,
namely, $t$-digest~\cite{dunning19-t-digest}, MomentSketch~\cite{GanDTSB18}, and KLL~\cite{KarninLL16}.
Our implementation of SplineSketch and the code to run the experiments or generate the synthetic datasets,
is available at supplementary repository (\url{https://github.com/PavelVesely/SplineSketch-experiments}).
We measure the average and maximum absolute rank errors for $10^5$ evenly spaced queries,
together with average update time (or merge time if applicable) and query times.
Experiments were performed on an AMD EPYC 7302 (3 GHz) server with 251 GB RAM and SSD disk,
and run in memory, with the datasets, queries, and sketches' output stored on disk.

\vspace{-0.2cm}

\paragraph{Quantile sketches.}
We have used the Java prototype of SplineSketch, in two versions: with and without the Misra-Gries sketch.
We have also evaluated $t$-digest~\cite{dunning19-t-digest} (v3.3, \url{https://github.com/tdunning/t-digest}), MomentSketch~\cite{GanDTSB18} (\url{https://github.com/stanford-futuredata/msketch}), and KLL~\cite{KarninLL16} (Java implementation by Apache DataSketches, v6.0.0, \url{https://github.com/apache/datasketches-java}), and GKAdaptive, a practically well-performing variant of the Greenwald-Khanna sketch~\cite{greenwald2001space},
	implemented in Java according to~\cite{luo16_quantiles_experimental}.
For $t$-digest, we use the default merging variant with $k_0$ scale function that aims at the uniform error~\cite{dunning19-t-digest}.
KLL was used in the variant with the \texttt{double} data type, i.e., an instance of \texttt{KllDoublesSketch}.
We note that the MomentSketch implementation only answers quantile queries, while
our experiments required rank estimates; we have simulated rank queries using binary search
that was done for all $m$ queries in parallel in order to decrease the number of
calls to the quantile query method from $\approx m \log_2 n$ (querying just one rank) to $\approx \log_2 n$ calls (querying $m$ ranks at once), to avoid repeated computation of the distribution fitting the moments and log-moments, which is a costly operation.
We have not included other quantile sketches because they have worse error than KLL or \textit{t}\mbox{-}digest, lack available implementations (e.g., algorithms from~\cite{SchieferCINSW23-KLLinterpolation,GribelyukSWY25elasticCompactors}), or aim for different error guarantees;
namely, ReqSketch~\cite{CormodeKLTV23} provides relative rank error and its space usage is much worse than KLL's, while
the algorithmic technique is similar.
DDSketch~\cite{MassonRL19} and UDDSketch~\cite{EpicocoMCPM20} aim for the relative \emph{value} error, which captures the distribution tails
but unlike the rank error, fails to provide a reasonable approximation for most of the distribution;
indeed, we have verified that the uniform rank error of DDSketch is 10-100 times worse than KLL and even more compared to SplineSketch;
see the supplementary repository for details.

The sketch size is measured in bytes when serialized on disk for storage without supporting
data structures such as buffers. Namely, for SplineSketch and $t$-digest we count 16 bytes per bucket/centroid,
for MomentSketch we measured $16k + 16$ bytes ($k$ moments and log-moments, minimum, and maximum), for KLL we count 8 bytes per stored item,
and for GKAdaptive, we count 24 bytes per a stored tuple of one item and its rank lower and upper bounds.
For SplineSketch with MG, we also account for the MG sketch size, after removing non-frequent items from the MG and reducing the number of buckets based on the number of heavy hitters as described in Section~\ref{sec:highFrequencyItems-Impl}.

\vspace{-0.2cm}

\paragraph{Datasets.}
We use three real-world datasets, the first two from UCI Machine Learning Repository~\cite{UCI_MLrepo} : the HEPMASS dataset~\cite{hepmass} ($n = 10\,500\,000$ items from the 2nd column),
the Power dataset~\cite{individual_household_electric_power_consumption_235} ($n = 2\,075\,259$ items from the 3rd column), and the Books dataset from SOSD~\cite{sosd,MarcusKRSMK0K20sosd}.
We have synthetically generated 7 datasets by drawing i.i.d.\ (independent and identically distributed) samples from
a range of distributions: uniform, normal, Pareto, Gumbel, lognormal, loguniform, and randomly signed loguniform.
We have also generated three datasets with a distribution change, all starting with $n/2$ items from a normal distribution
and then either changing the parameters of the normal distribution (two options for the parameter change) or adding samples from a set of 42 distinct items,
which all have high frequency.
Finally, we generated a sorted input with frequent items, each with random frequency between 1 and $n/50$.
All synthetic datasets for accuracy experiments consist of $n = 10^8$ items.

Rank queries are generated by equally spaced selection from sorted data and are executed in one batch whenever possible.
While all of the sketches allow any number to be queried, using only the data points for queries leads to the same results.

\begin{figure*}[t]
	\vspace{-0.2cm}
	\centering
	\includegraphics[width=0.9\textwidth]{legend.png}
	\subcaptionbox{Normal distribution\label{fig:iid-avgErr-subfig1normal}}{%
		\includegraphics[width=0.32\textwidth]{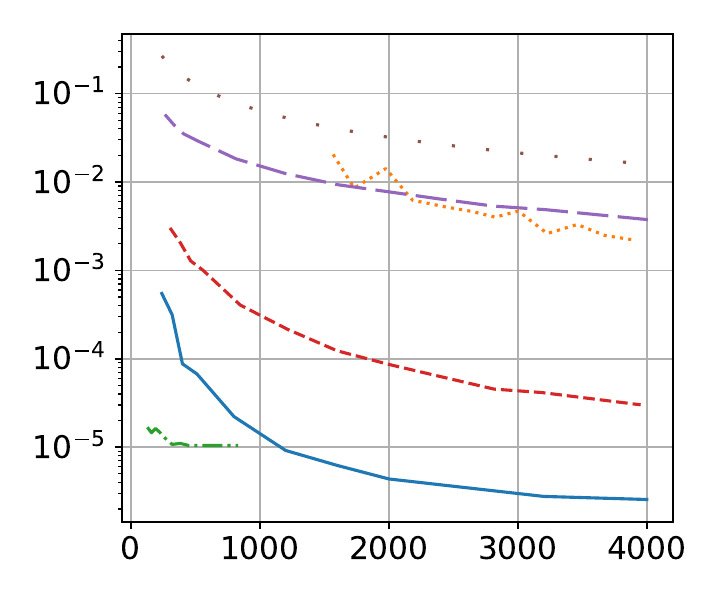}%
	}
	\subcaptionbox{Uniform distribution\label{fig:iid-avgErr-subfig2uniform}}{%
		\includegraphics[width=0.32\textwidth]{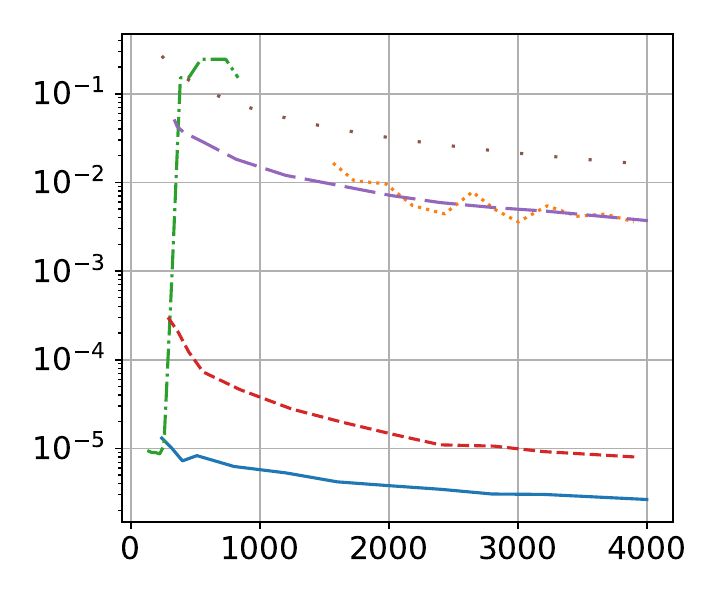}%
	}
	\subcaptionbox{Pareto distribution\label{fig:iid-avgErr-subfigPareto}}{%
		\includegraphics[width=0.32\textwidth]{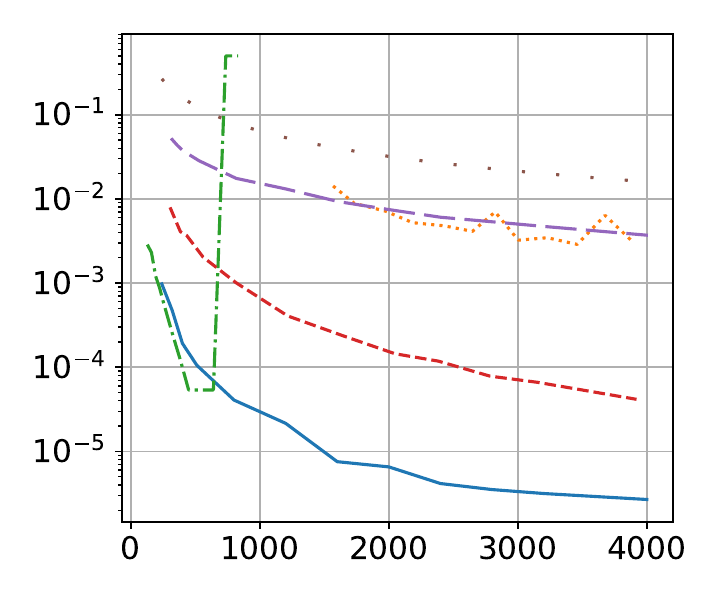}%
	}
	\subcaptionbox{Lognormal distribution\label{fig:iid-avgErr-subfigLognormal}}{%
		\includegraphics[width=0.32\textwidth]{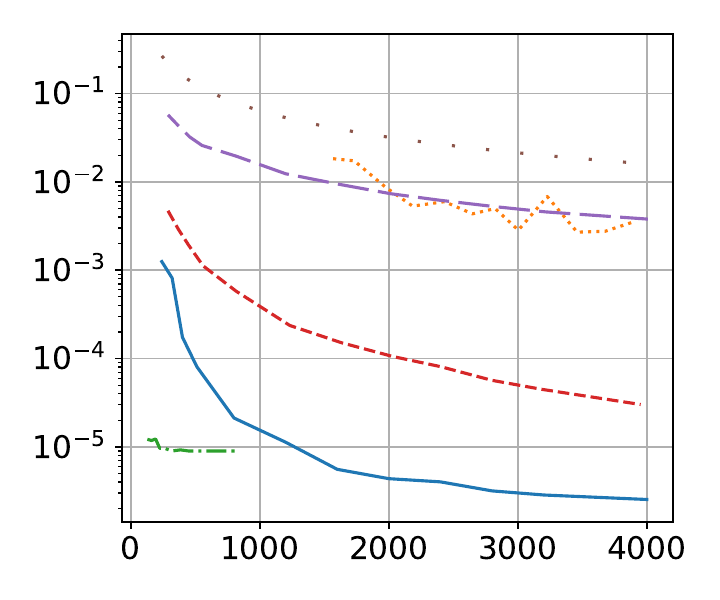}%
	}
	\subcaptionbox{Loguniform distribution\label{fig:iid-avgErr-subfigLoguniform}}{%
		\includegraphics[width=0.32\textwidth]{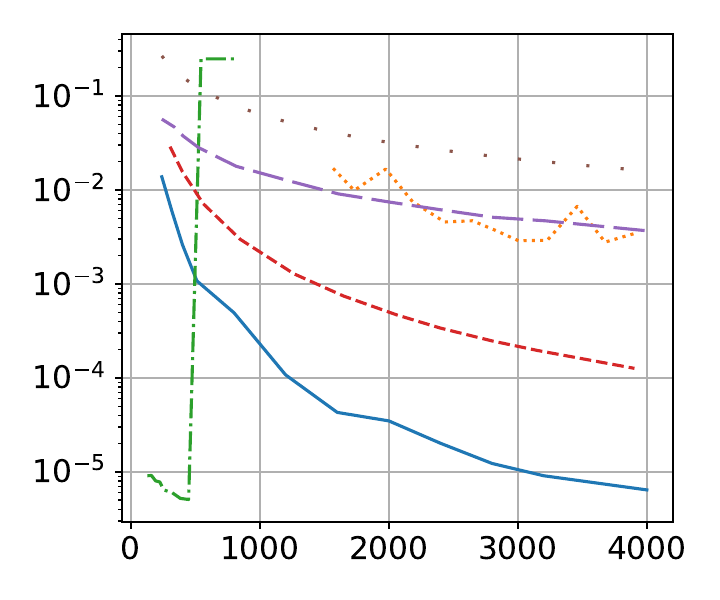}%
	}
		\subcaptionbox{Signed loguniform distribution\label{fig:iid-avgErr-subfigSignedLogunif}}{%
				\includegraphics[width=0.32\textwidth]{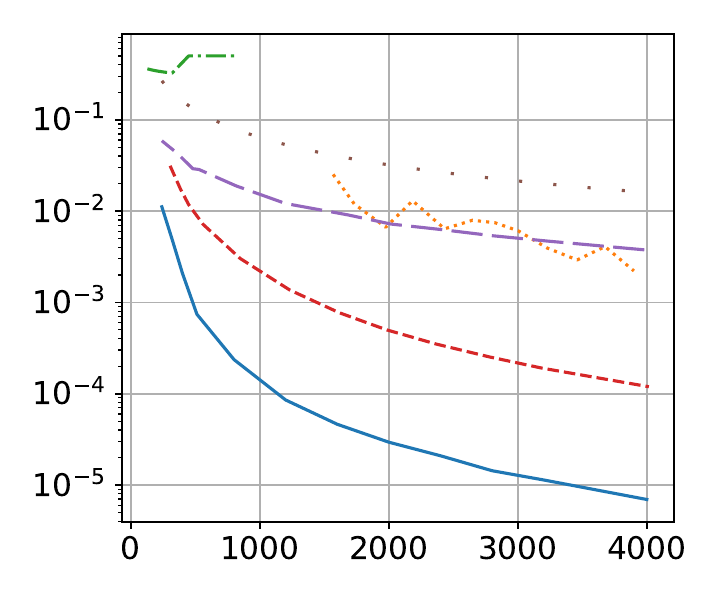}%
			}
	\subcaptionbox{Signed loguniform distribution with large exponents\label{fig:iid-avgErr-subfigSignedLogunifExtreme}}{%
		\includegraphics[width=0.32\textwidth]{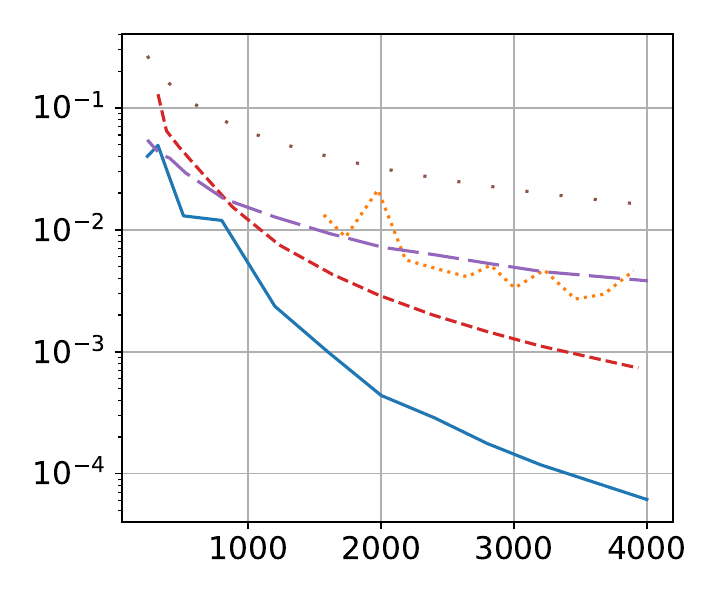}%
	}
	\subcaptionbox{Normal distrib.\ with a small parameter change after $n/2$ items\label{fig:iid-avgErr-subfigNormalWSmallChange}}{%
		\includegraphics[width=0.32\textwidth]{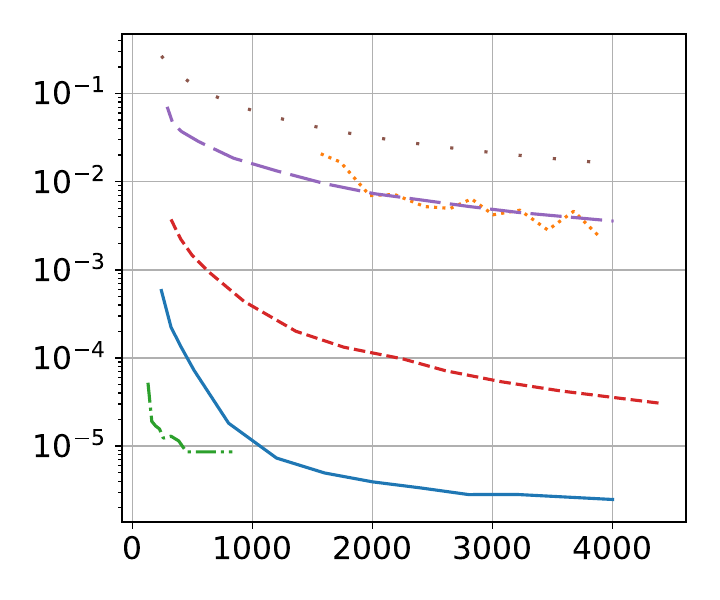}%
	}
	\subcaptionbox{Normal distribution with a large parameter change after $n/2$ items\label{fig:iid-avgErr-subfigNormalWLargeChange}}{%
		\includegraphics[width=0.32\textwidth]{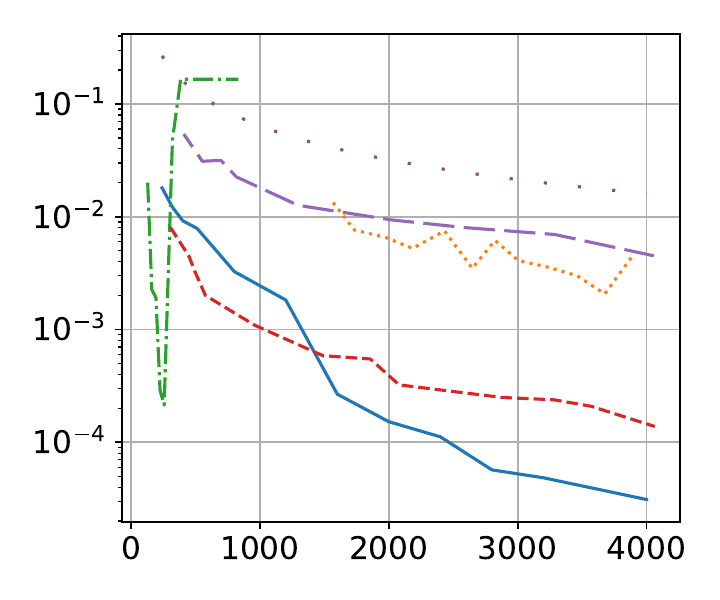}%
	}
	\vspace{-0.2cm} %
	\caption{Average rank error (log-scale) depending on the sketch size in bytes.
	}
	\label{fig:iid-avgErr}
	\vspace{-0.3cm}
\end{figure*}

\begin{figure*}[t]
	\vspace{-0.2cm}
	\centering
	\includegraphics[width=0.9\textwidth]{legend.png}
	\subcaptionbox{HEPMASS dataset~\cite{hepmass}\label{fig:real-maxErr-subfig1hepmass}}{%
		\includegraphics[width=0.32\textwidth]{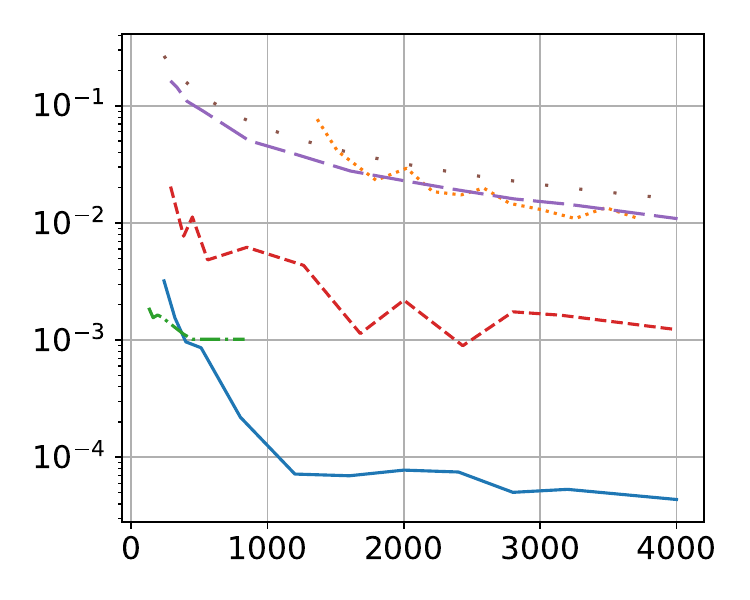}%
	}
	\subcaptionbox{Power dataset~\cite{individual_household_electric_power_consumption_235}\label{fig:real-maxErr-subfig2power}}{%
		\includegraphics[width=0.32\textwidth]{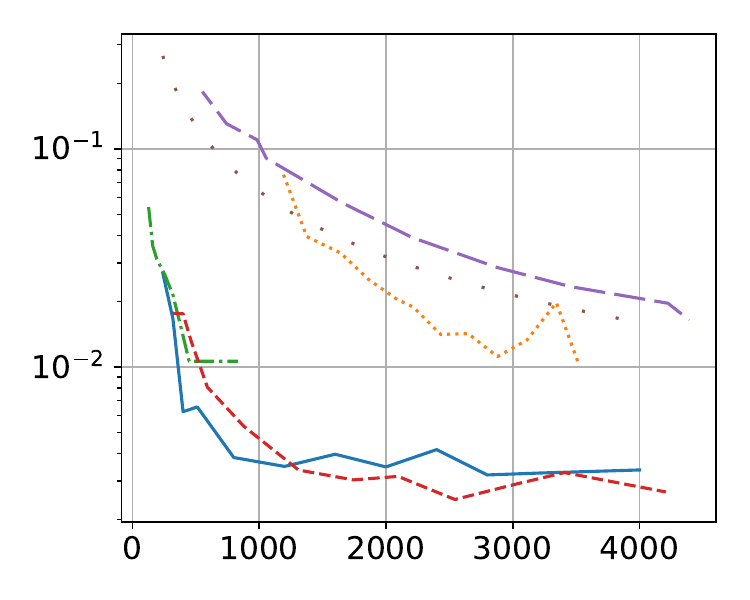}%
	}
		\subcaptionbox{Books dataset~\cite{sosd,MarcusKRSMK0K20sosd}\label{fig:real-maxErr-subfig2books}}{%
				\includegraphics[width=0.32\textwidth]{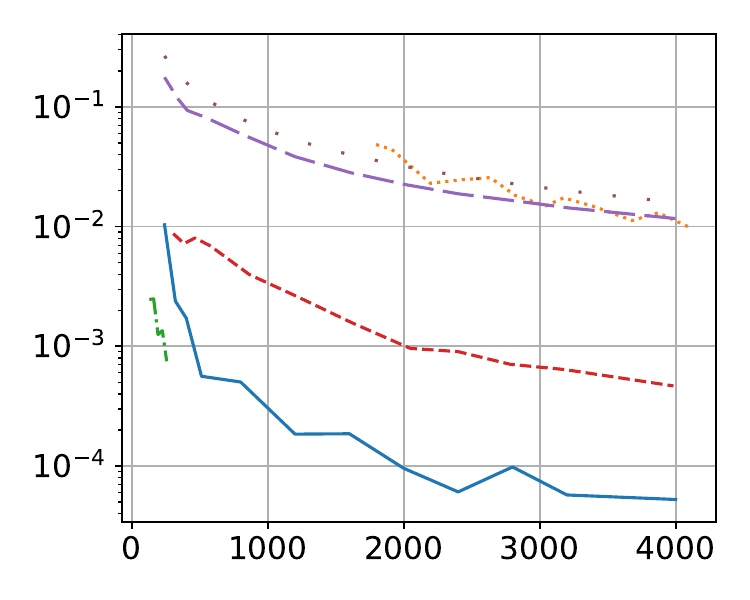}%
			}
	\subcaptionbox{Normal distribution\label{fig:iid-maxErr-subfig1normal}}{%
		\includegraphics[width=0.32\textwidth]{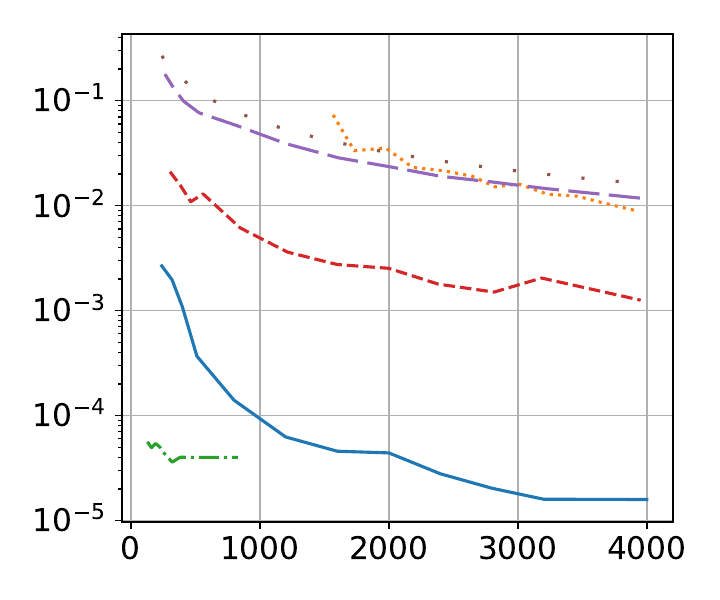}%
	}
		\subcaptionbox{Uniform distribution\label{fig:iid-maxErr-subfig2uniform}}{%
				\includegraphics[width=0.32\textwidth]{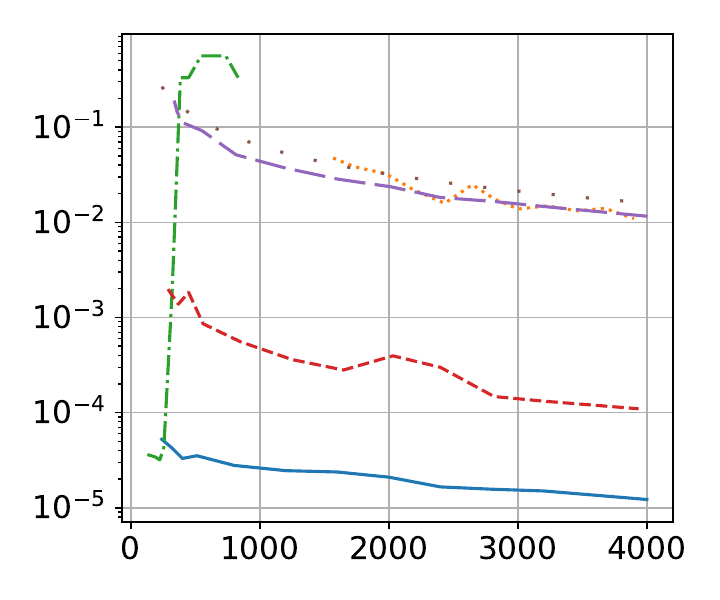}%
			}
	\subcaptionbox{Pareto distribution\label{fig:iid-maxErr-subfigPareto}}{%
		\includegraphics[width=0.32\textwidth]{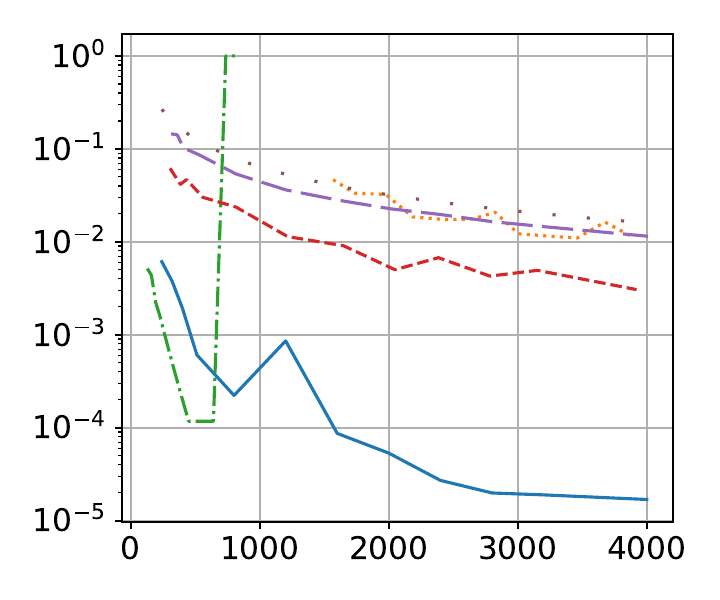}%
	}
	\subcaptionbox{Lognormal distribution\label{fig:iid-maxErr-subfigLognormal}}{%
		\includegraphics[width=0.32\textwidth]{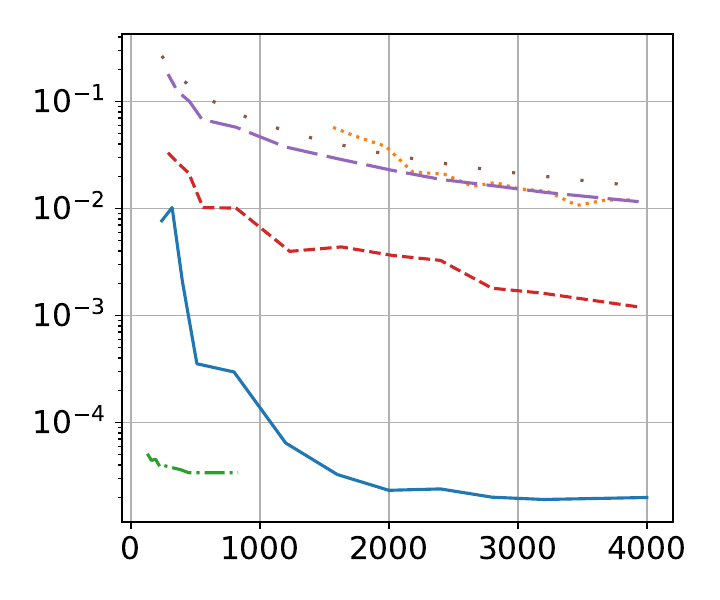}%
	}
		\subcaptionbox{Loguniform distribution\label{fig:iid-maxErr-subfigLoguniform}}{%
				\includegraphics[width=0.32\textwidth]{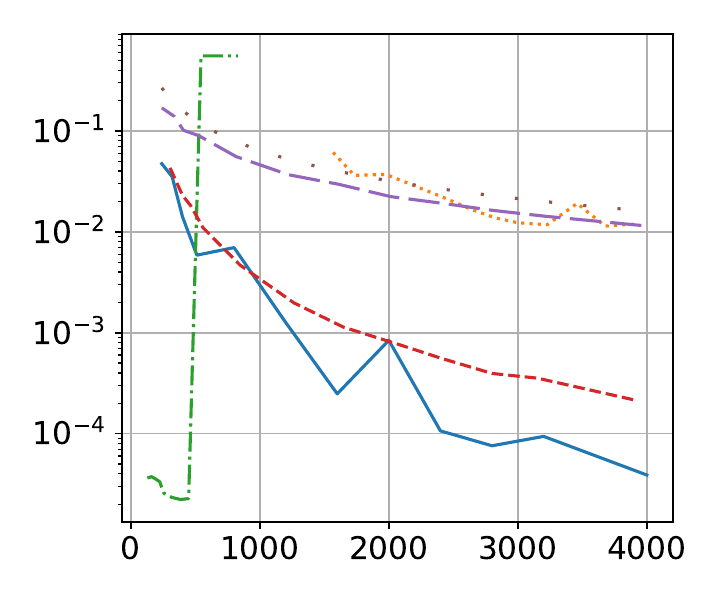}%
			}
		\subcaptionbox{Signed loguniform distribution\label{fig:iid-maxErr-subfigSignedLogunif}}{%
				\includegraphics[width=0.32\textwidth]{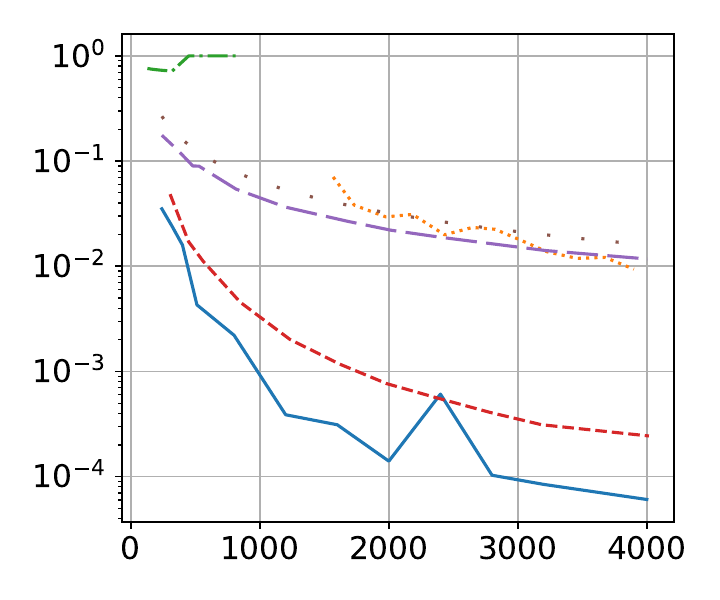}%
			}
	\subcaptionbox{Signed loguniform distribution with large exponents\label{fig:iid-maxErr-subfigSignedLogunifExtreme}}{%
		\includegraphics[width=0.32\textwidth]{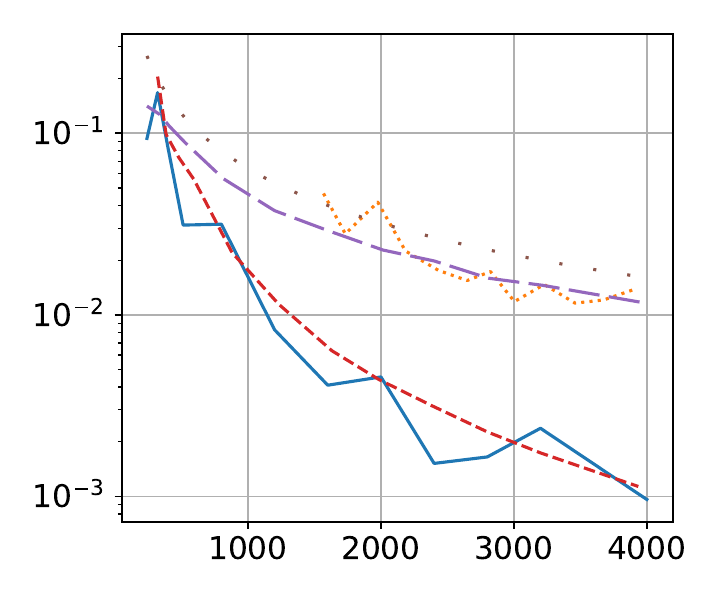}%
	}
	\subcaptionbox{Normal distribution with a small parameter change after $n/2$ items\label{fig:iid-maxErr-subfigNormalWSmallChange}}{%
		\includegraphics[width=0.32\textwidth]{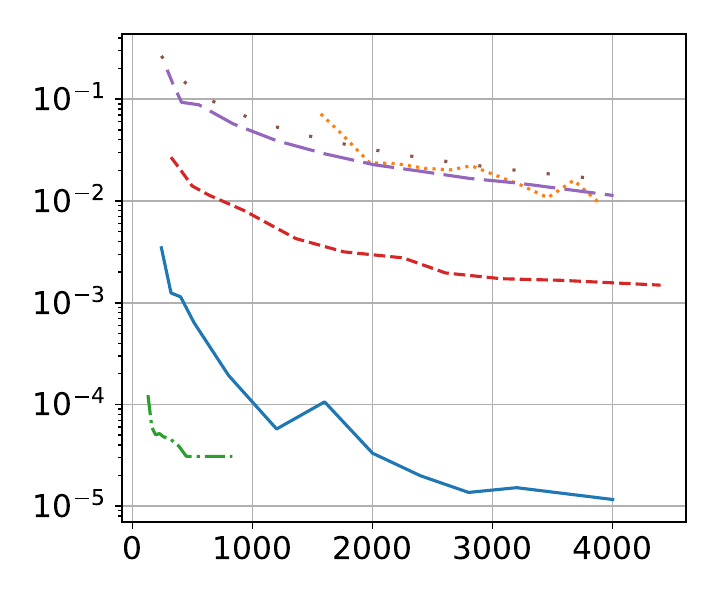}%
	}
	\subcaptionbox{Normal distribution with a large parameter change after $n/2$ items\label{fig:iid-maxErr-subfigNormalWLargeChange}}{%
		\includegraphics[width=0.32\textwidth]{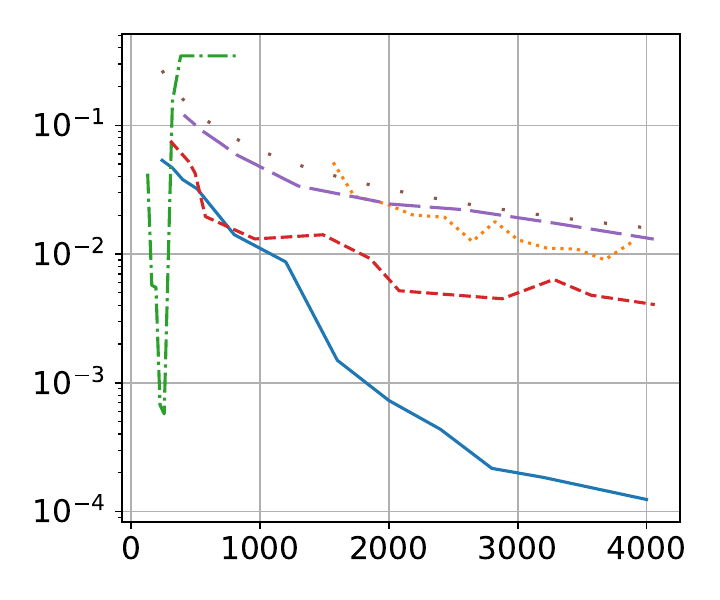}%
	}

	\caption{Maximum rank error (log-scale) depending on the sketch size in bytes, on a selection of datasets.
		The maximum is taken over $10^5$ evenly spaced queries.
		\vspace{-0.1cm}
	}
	\label{fig:iid-maxErr}
\end{figure*}

\clearpage

\subsection{Evaluation in the Streaming Setting}

\paragraph{Accuracy-space trade-off.}
We first present the results in the streaming setting, focusing on synthetic and real-world inputs without frequent items.
Figures~\ref{figTop:error} and~\ref{fig:iid-avgErr} show the average errors of the sketches depending on their size in bytes.
Figure~\ref{fig:iid-maxErr} shows the maximum error over all queries, again depending on the sketch size.
In all these plots, the error is normalized by the input size, i.e., it is the difference 
between the true rank and the estimated rank of a query item divided by the input size $n$, averaged or maximized over $10^5$ queries.
We also depict an upper bound on the theoretically optimal error bound of $O(1/k)$ in purple (loosely dotted), where $k$ is the sketch size; the hidden constant factor is based on the maximum error of KLL and GKAdaptive.

Overall, SplineSketch consistently provides the best accuracy if given sufficient size, namely $k\ge 100$,
even though in rare cases or for small sketch sizes, $t$-digest or MomentSketch are better.
In more detail,
compared to $t$-digest, SplineSketch typically achieves 2--20 times smaller error,
and in some cases, such as for normally distributed data, up to 100 times.
In a few cases, their error was similar, e.g., for the Power dataset or the maximum error for signed loguniform distribution with large exponents.
We note that the final $t$-digest size is somewhat unpredictable as
the final number of centroids does not match the compression parameter given to it, sometimes
by 66\%.
This demonstrates that $t$-digest does not use the space budget efficiently, unlike SplineSketch
which always has $k$ buckets.

MomentSketch is mainly intended for use with very small size, typically up to $k=15$ moments and log-moments,
when it often gives the best accuracy by 1--2 orders of magnitude.
It works the best on many smooth distributions (uniform, normal, loguniform, etc.).
However, its performance is significantly affected by non-smoothness as witnessed on the real-world datasets. Moreover, one cannot increase accuracy by increasing $k$ because of numeric issues.
Furthermore, MomentSketch query procedure failed with an exception on many datasets,
specifically on signed loguniform distribution with large exponents for any $k$ and 
on other datasets for large $k$.

Finally, the KLL sketch and GKAdaptive are worse than SplineSketch typically by 1--3 orders of magnitude 
as they use no interpolation (cf.~\cite{SchieferCINSW23-KLLinterpolation} for KLL with interpolation; however, no implementation is available).
Nevertheless, as KLL and GKAdaptive are comparison-based, their errors depend just on the data size and the arrival order,
not on their particular distribution, and their error decreases linearly with increasing sketch size.

We note that both versions of SplineSketch (without and with the MG sketch) provide the same accuracy on the inputs in Figures~\ref{fig:iid-avgErr} and~\ref{fig:iid-maxErr}, i.e., those without frequent items.

\paragraph{Update and query times.} %
In Figure~\ref{fig:updateTime}, we show a log-log plot with average times per update in $\mu$s
in dependence on $n$ samples from the normal distribution. The sketches' size is fixed to about 1.6 kB,
i.e., with $100$ buckets for SplineSketch and about 100 centroids for $t$-digest.
SplineSketch is evaluated without the MG sketch; see Section~\ref{sec:expSkewed} for the other version.
For MomentSketch, we use $k=15$ (i.e., size about 256 bytes) as such $k$ generally provides the best accuracy.
For KLL, we use $k=8$, i.e., the smallest meaningful size parameter.
For all sketches, the update time is decreasing fast from small data sizes
as initializing and processing the first part of the input takes relatively the most time;
namely, for SplineSketch and $t$-digest, it takes some time before the buckets/centroids stabilize.
For sufficiently long input streams, the average update time becomes constant.
As expected, MomentSketch and KLL achieve the best update times of about 0.03$\mu$s and 0.06$\mu$s, respectively.
SplineSketch achieved less than 0.1$\mu$s update time for $n\ge 10^8$,
while $t$-digest required about 0.12$\mu$.
GKAdaptive is by far the slowest for large stream lengths, with update time increasing with its size from 0.1 to 1$\mu$; see Figure~\ref{fig:updateTime2};
	however, this GKAdaptive implementation lacks running time optimization.
Figure~\ref{fig:updateTime2} shows the dependency of the update time on the sketch size.
Both SplineSketch and $t$-digest have mildly (sublinearly) increasing update time with the sketch size.
For MomentSketch, the update time is increasing linearly with its size,
while for KLL it remains stable.

\begin{figure*}[t]
	\centering
	\vspace{-0.2cm}
	\includegraphics[width=0.8\textwidth]{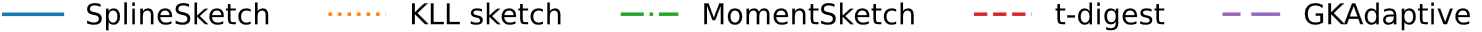}
	\newline
	\subcaptionbox{Log-log plot with update time depending on data size with fixed sketch size of $\approx 1.6$ kB.\label{fig:updateTime1}}{%
		\includegraphics[width=0.49\textwidth]{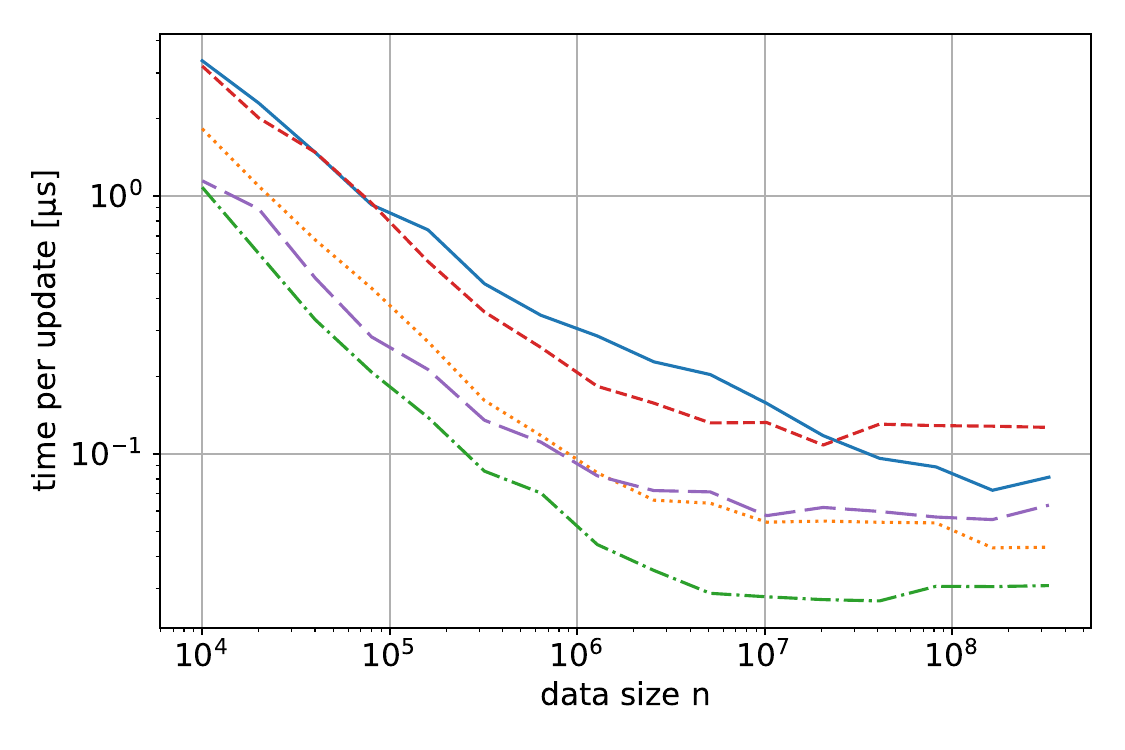}%
	}
	\hfill
	\subcaptionbox{Update time depending on the sketch size, for $n = 10^8$.\label{fig:updateTime2}}{
		\includegraphics[width=0.49\textwidth]{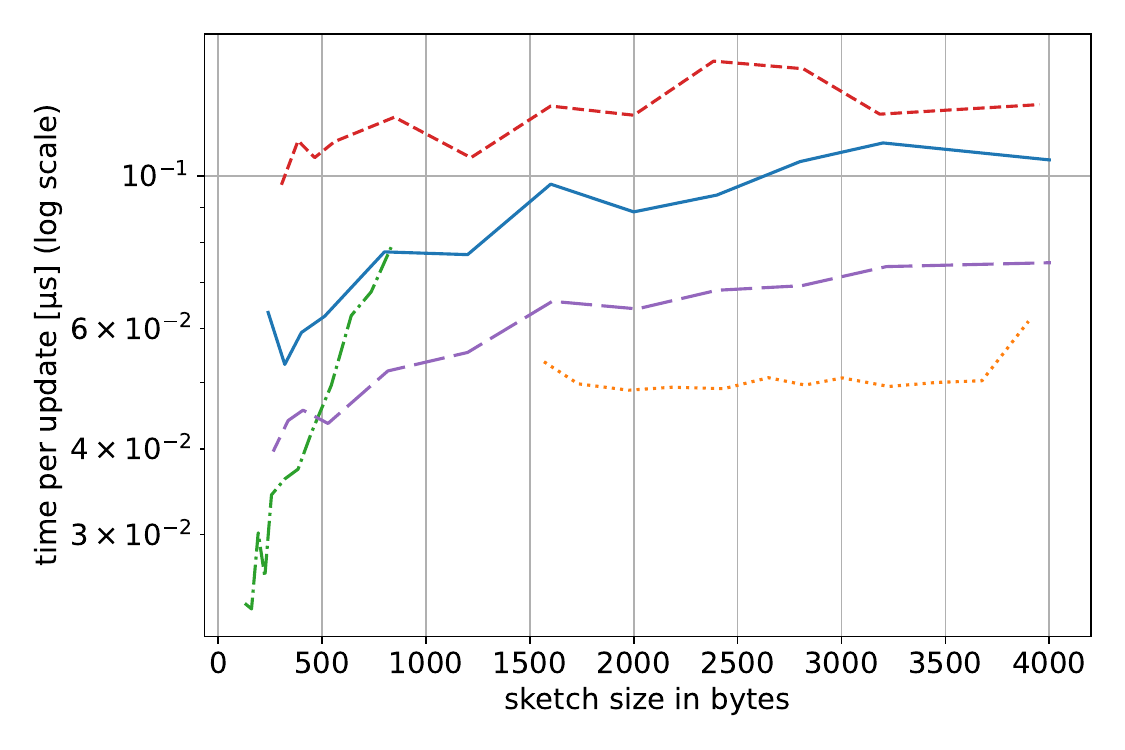}
	}
	\vspace{-0.3cm}
	\caption{Time per update in $\mu$s on normally distributed data;.
	}\label{fig:updateTime}
\end{figure*}
\begin{figure*}[t]
	\centering
	\subcaptionbox{Log-log plot with query time depending on the number of queries, with fixed sketch size of $\approx 1.6$ kB.\label{fig:queryTime}}{%
		\includegraphics[width=0.48\textwidth]{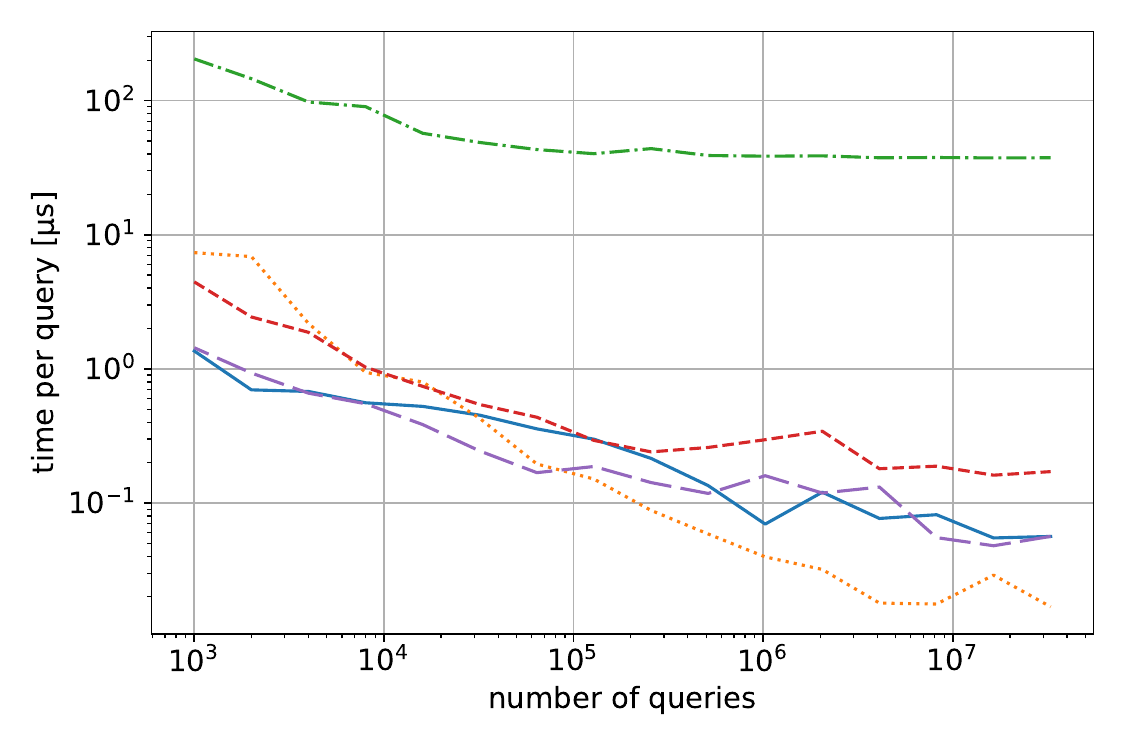}%
	}
	\hfill
	\subcaptionbox{Query time depending on the sketch size for $10^5$ queries.\label{fig:queryTime2}}{%
		\includegraphics[width=0.48\textwidth]{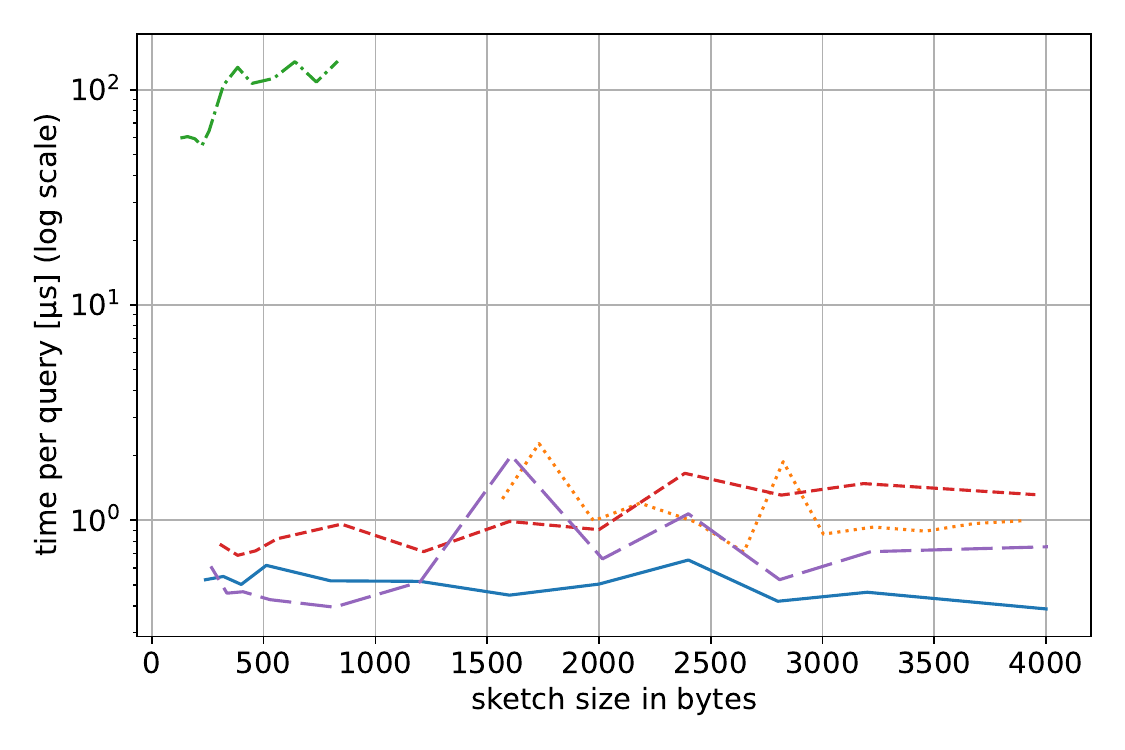}%
	}
	\vspace{-0.4cm}
	\caption{Time per query in $\mu$s on $n=10^8$ samples from the normal distribution.
		\vspace{-0.2cm}}
\end{figure*}

In Figure~\ref{fig:queryTime}, we show average 
query times in $\mu$s depending on the number of queries in the same setup as for update time, again with a fixed sketch sizes of about 1.6 kB,
for $10^8$ normally distributed items.
KLL is the fastest to query for a large number of queries due to its simplicity, requiring only about 0.03$\mu$s per query,
while SplineSketch and GKAdaptive are nearly as fast as KLL
and outperform KLL for a small number of queries (less than 40\,000);
namely, SplineSketch and GKAdaptive use at most 1$\mu$ per query for more than 1\,000 queries.
$t$-digest is typically 2--3 times slower than SplineSketch.
MomentSketch is especially slow, by over two orders of magnitude due to the need to compute a distribution fitting the moments and also, since the particular
implementation used only provides quantile queries, we obtain rank estimates using binary search.
Figure~\ref{fig:queryTime2} shows that, with the exception of MomentSketch, the query time does not increase with the sketch size.

\subsection{Evaluation of Mergeability}\label{sec:expMerges}

	We evaluated the performance sketches when they are created by pairwise merge operations. 
	To this end, we split the input into 10\,000 chucks of the same size, feed each chunk into one sketch, and then merge the resulting
	sketches into one in a complete binary tree fashion, i.e., merging sketches summarizing a similar number of items.	

	As illustrated in Figures~\ref{fig:mergingAccuracy} and~\ref{fig:mergingAccuracyHepmass}, the overall accuracy picture is similar as in the streaming setting,
	with SplineSketch having 2--20 times smaller error than $t$-digest and even more compared to other sketches.
	However, the error of SplineSketch and $t$-digest got worse compared to the streaming setting, about 5--10 times,
	even though both of these sketches still largely outperform KLL and the comparison to MomentSketch is similar as in streaming.
	As for the time efficiency, we measured the time to perform a merge operation (averaged over the 9\,999 merges performed to obtain the final sketch); see Figure~\ref{fig:mergingSpeed}.
	Moment Sketch was the fastest to merge, due to taking simple sums of the moments and log-moments,
	followed by KLL, which still fit below 10$\mu$s per merge operation.
	$t$-digest and GKAdaptive required up to 100--200$\mu$s per merge, due to the need to merge the centroids,
	and SplineSketch up to 500$\mu$s as merging of buckets, while conceptually simple, requires a lot of bucket joins. 
	The query time is the same as in streaming, i.e., does not depend on how the sketch is created.
	The results for other datasets are similar; see the supplementary repository.

\begin{figure*}[t]
	\centering
		\vspace{-0.2cm}
		\includegraphics[width=0.8\textwidth]{legendNoErrorBound.png}
		\newline
	\subcaptionbox{Maximum rank error (log-scale) on the normal distribution\label{fig:mergingAccuracy}}{%
		\includegraphics[width=0.32\textwidth]{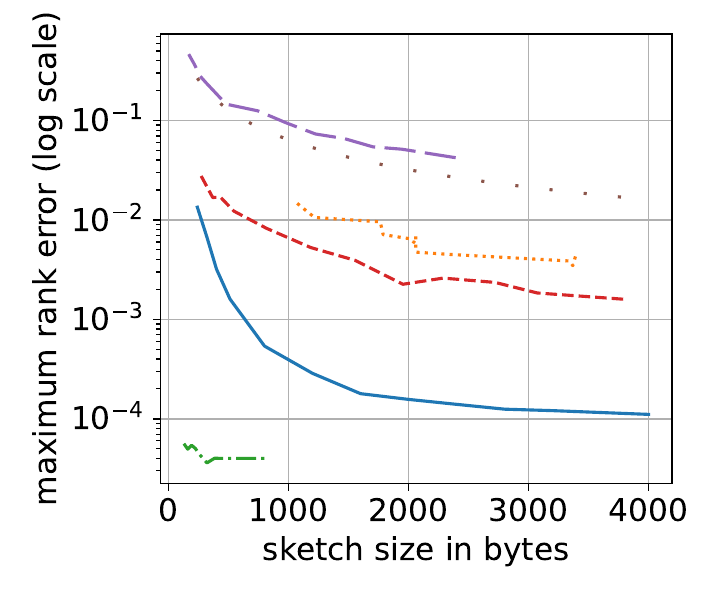}%
	}%
	\hfill
	\subcaptionbox{Maximum rank error (log-scale) on the HEPMASS dataset~\cite{hepmass}\label{fig:mergingAccuracyHepmass}}{%
		\includegraphics[width=0.32\textwidth]{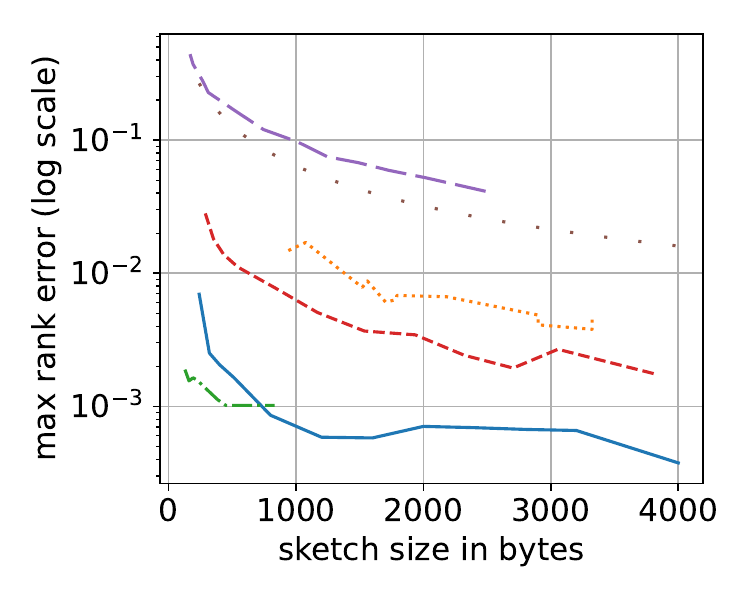}%
	}%
	\hfill
	\subcaptionbox{Time per merge operation in $\mu$s (log-scale) on the normal distribution\label{fig:mergingSpeed}}{%
		\includegraphics[width=0.32\textwidth]{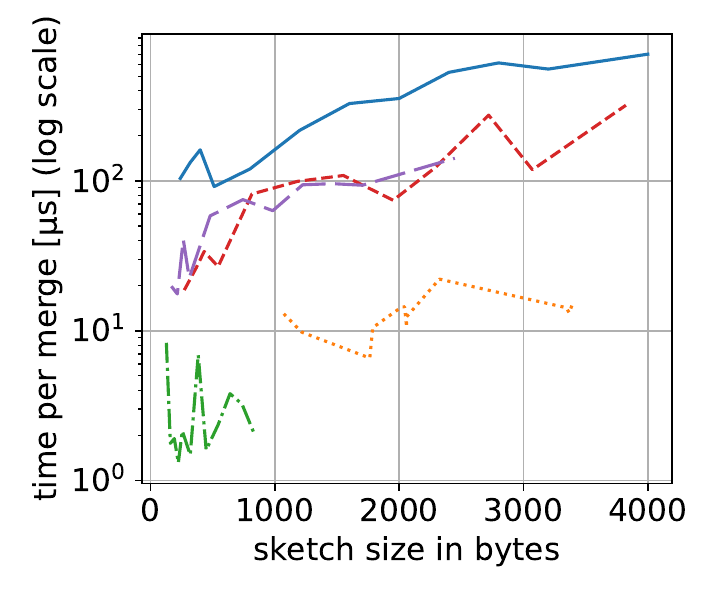}%
	}
	\caption{Error and merge time in the mergeability setting  depending on the sketch size in bytes,
		namely, for sketches built by merge operations from sketches for 10\,000 input chunks.
		\vspace{-0.2cm}}	
\end{figure*}

\subsection{Inputs with Heavy Hitters}\label{sec:expSkewed}

Here, we focus on inputs with high-frequency items, where SplineSketch without the MG sketch
does not perform so well.
Our results in Figures~\ref{fig:skewedAccuracy1} and~\ref{fig:skewedAccuracy2} confirm that 
using the Misra-Gries (MG) sketch alongside SplineSketch does indeed restore the very high accuracy
when the MG sketch is large enough to detect the heavy hitters (note that we account for the size of the MG sketch after compressing it as described in Section~\ref{sec:highFrequencyItems-Impl}).
In more detail, while the error of $t$-digest and SplineSketch without MG gets worse and does not improve with the sketch size,
SplineSketch with MG achieves by up to three orders of magnitude better error, and zero error
on inputs consisting of a small number of distinct items only as in Figure~\ref{fig:skewedAccuracy2}.
We note that the error of KLL and GKAdaptive remains the same as on inputs with distinct items as they are comparison-based.
As shown in Figure~\ref{fig:skewedUpdateTime}, the update time of SplineSketch with MG is about three times worse than without MG 
(though it can be optimized), while merge and query times are only slightly larger.

\begin{figure*}[htbp]
	\includegraphics[width=0.99\textwidth]{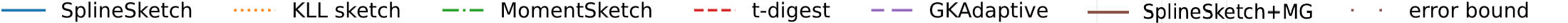}
	\newline
	\centering
	\subcaptionbox{Average rank error (log-scale),
		on normal distribution with $n/2$ items, followed by $n/2$ high-frequency items.\label{fig:skewedAccuracy1}}{%
			\includegraphics[width=0.32\textwidth]{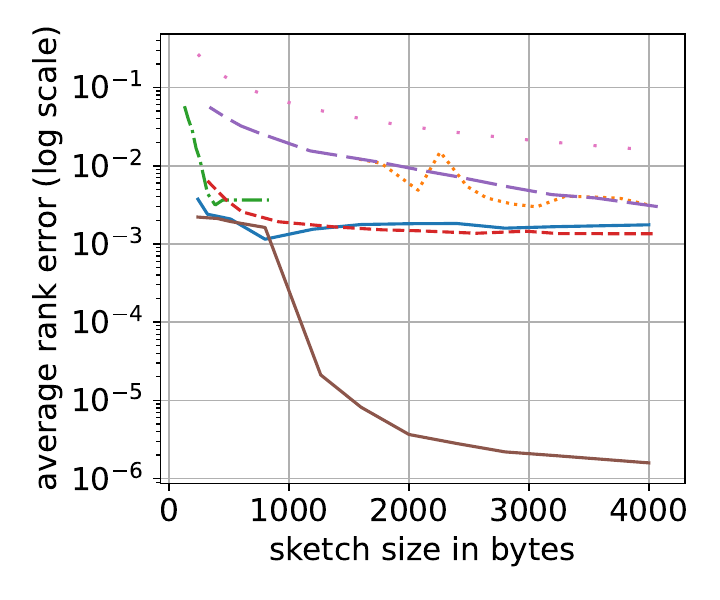}%
		}
	\hfill
	\subcaptionbox{Average rank error (log-scale),
		on a sorted input with frequent items.\label{fig:skewedAccuracy2}}{%
			\includegraphics[width=0.32\textwidth]{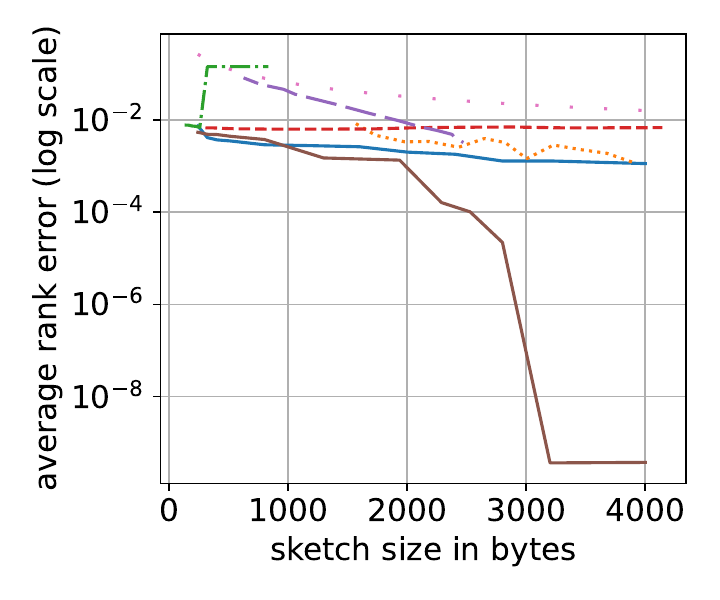}%
		}
	\hfill
	\subcaptionbox{Time per update in $\mu$s for the same input as in Figure~\ref{fig:skewedAccuracy1}.\label{fig:skewedUpdateTime}}{%
			\includegraphics[width=0.32\textwidth]{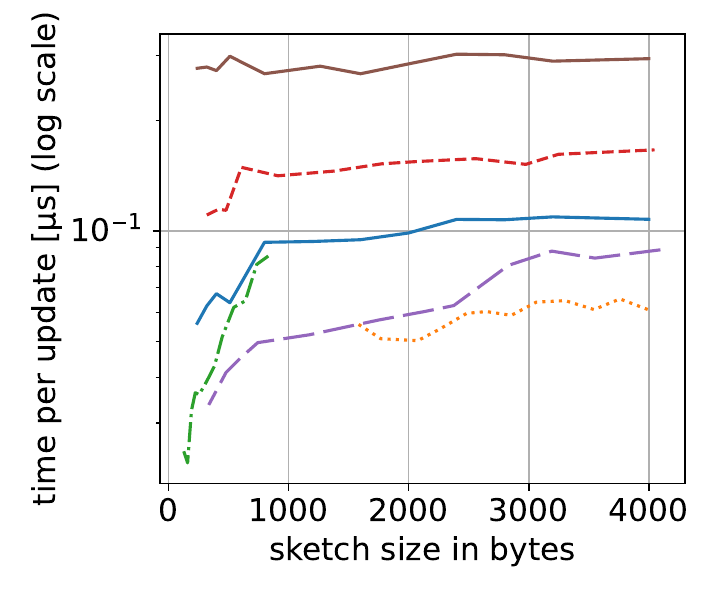}%
		}
	\vspace{-0.4cm}
\caption{Accuracy and update time for inputs with frequent items, depending on the sketch size in bytes.
\vspace{-0.2cm}}	
\end{figure*}

\subsection{Ablation Studies}\label{sec:ablation}

	Here, we provide insights into the specific design choices in SplineSketch.
	First, we evaluated the benefit of using PCHIP interpolation and
	found that PCHIP interpolation achieves 3--10 times better
	error than just linear interpolation over the buckets; see Figure~\ref{fig:ablationInterpolations} for an example.

	Second, we focused on the heuristic error and compared several alternatives to the second derivative of the empirical CDF, normalized by length squared, as defined in~\eqref{eqn:heurError-der2}. The evaluated alternatives are: using no heuristic error, bucket length, bucket counter (which is the first derivative of the empirical CDF), and the third derivative of the empirical CDF, normalized by length cubed. 
	Figure~\ref{fig:ablationHeurError} shows that~\eqref{eqn:heurError-der2} is indeed the best, while using the bucket length, i.e., making the buckets similar in length on the input range, is by only about 20\% worse.
	The third derivative is about 1.5--2 times worse, and using the bucket counter or no heuristic error is over 10 times worse.

	Finally, we briefly comment on parameter tuning.
	We did not observe significant differences in the accuracy of SplineSketch with respect to parameter $\gamma > 1$ in \Cref{def:splittableBuckets} (for values between 1.1 and 2). 
	For the initial multiplicative constant $C_b$ in the bucket bound~\eqref{eqn:bucketBound},
	we tested values in $[1, 7]$ and for $C_b \ge 3$, the accuracy was similar ($C_b < 3$ was only better on the uniform distribution).
	Other parameters, such as the epoch increase factor (default=1.25) or the minimum fraction of the bucket bound in \Cref{def:splittableBuckets} of splittable buckets (default=0.01)
	also do not influence the accuracy much if adjusted. %
	See the supplementary repository for details.

\begin{figure}[htbp]
	\vspace{-0.2cm}
	\centering
	\subcaptionbox{Interpolation\label{fig:ablationInterpolations}}{
			\includegraphics[width=0.48\columnwidth]{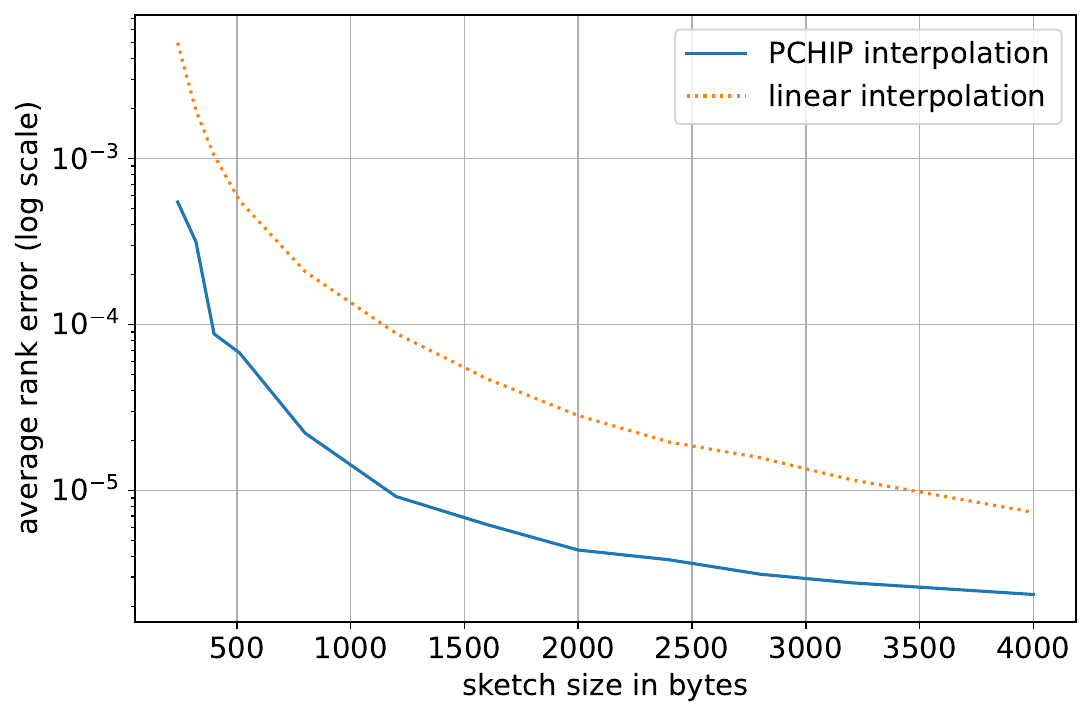}
		}
	\subcaptionbox{Heuristic error\label{fig:ablationHeurError}}{
			\includegraphics[width=0.48\columnwidth]{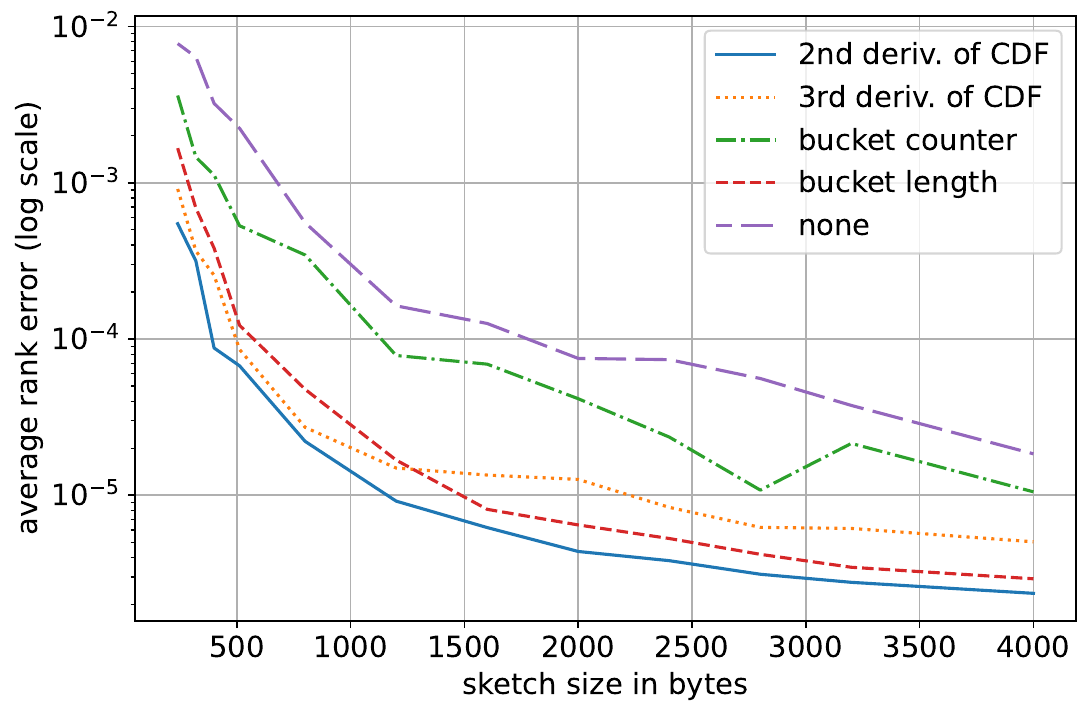}
		}	
	\vspace{-0.2cm}
	\caption{Ablation studies on the normal distribution.
		\vspace{-0.5cm}}
\end{figure}

\section{Conclusions and Discussion}\label{sec:conclusions}

In this work, we designed a new deterministic quantile sketch which has fixed memory consumption of $O(k)$ memory words.
The main technical contribution is in careful maintenance of bucket thresholds and counters.
This leads to extremely accurate rank and quantile estimates using the PCHIP interpolation as well as uniformly bounded worst-case error,
achieving a near-optimal accuracy-space trade-off.
We have proven that our sketch is mergeable and retains its guarantees in a balanced mergeability setting, up to an additive $O(n \log(n/k) / k)$ error.
Despite that we claimed full mergeability (i.e., the same error guarantees under arbitrary merging),
it does not hold for our sketch. One standard way to deal with it is to maintain $\approx \log_2(n/k)$ instances of the sketch with $k$ buckets,
one for each power of two between $\Theta(k)$ and $n$, with the $i$-th instance summarizing $(2^{i-1} k, 2^i k]$ items. 
This reduces arbitrary merging into balanced merging of the individual instances but introduces an extra $O(\log (n/k))$ factor in the space complexity.
In fact, compared to $q$-digest, our sketch may be less suitable for arbitrary merging (at least from a theoretical point of view) due to bucket thresholds not being
aligned between the sketches merged.
We leave open how to modify SplineSketch for full mergeability without losing a logarithmic factor.

Besides very high accuracy on various datasets, the key advantage of our approach is flexibility.
For instance, it is possible to resize the sketch, making it smaller or larger depending on how the memory resources change
(e.g., one can use a larger sketch for input processing and then make it smaller for storage).
More importantly, if the user has a prior knowledge about the data distribution, 
it is possible to preset the initial thresholds of SplineSketch based on this prior distribution
and aim for even lower error.
Furthermore,  if the prior knowledge turns out to be wrong, our sketch adapts to the new distribution,
similarly as we have shown in Figure~\ref{fig:iid-avgErr-subfigNormalWLargeChange}.

One of the challenges when implementing SplineSketch is dealing with high-frequency items.
One possible approach is using the Misra-Gries sketch~\cite{misra1982finding}, which
	restores theoretical guarantees (Section~\ref{sec:theory})
	and very high accuracy for inputs with heavy hitters, as demonstrated in \Cref{sec:expSkewed}.
Nevertheless, using a heavy-hitter sketch comes at a cost of increased update, merge, and query times,
so we also provide a faster version without a heavy hitter sketch.
Another limitation of our implementation arises from using the \texttt{double} data type in Java,
which means that for some inputs with many significant digits (e.g., the SOSD benchmark~\cite{sosd,MarcusKRSMK0K20sosd}), accuracy degrades as the sketch rounds many items to the same \texttt{double} value;
this can be addressed by changing the underlying data type, e.g., to 64-bit integers.
While we optimized the Java version to some extent,
our implementation is so far a prototype %
intended to demonstrate the main advantages of our approach and outline its implementation challenges.
Due to the intricacies of maintaining the buckets, our sketch is slower in processing the input stream than KLL or MomentSketch,
especially when many merge operations are performed.
However, SplineSketch is already faster in streaming than $t$-digest.
Moreover, further optimizations are possible, such as the ones presented in \cite{pfeil.etal:2025:proc.21stint.workshopdatamanag.newhardw.} (though not applicable to Python or Java).
Thus, a well-optimized implementation may be significantly faster than $t$-digest.

From the theoretical point of view, we ask whether one can modify our algorithm so that the error bounds
have better dependency on $\alpha$. We note that there is some evidence from the
comparison-based lower bounds in~\cite{CormodeV20} that a factor of $\log \log \alpha$ may be necessary
without using bit-packing tricks developed for $q$-digest in~\cite{GuptaSW24optimaQS-non-comp}.
The accuracy in practice, despite being already very high, may also be improved,
and it is not clear where is the accuracy limit of quantile summaries on real-world or synthetic datasets.

Our work opens up a new direction for future work in streaming quantile estimation.
In particular, we believe that our techniques
have potential to yield much better quantile sketches with relative-error guarantees,
capturing the distribution tails more accurately than the median.
However, achieving the relative error is substantially harder as witnessed by several 
major open problems in the theory of relative-error quantile sketches (cf.~\cite{CormodeV20,GribelyukSWY25elasticCompactors}).

\paragraph{Acknowledgments.}
We thank the SIGMOD'26 reviewers for their constructive feedback and suggestions for improvement.
We also thank Tomáš Domes and Jakub Komárek for their feedback; in particular, Tomáš discovered issued in the analysis in the previous version.
ChatGPT was used to generate parts of this work, including polishing language in the text, pseudocodes prototyping, translation of SplineSketch from Python to Java, and the proof of Lemma~\ref{lem:cubicPolyMonoMiddleBound}.
We have checked all the generated content, corrected it where necessary.
In particular, the Java implementation was largely rewritten and optimized manually.

A. Łukasiewicz and P.~Veselý were partially supported by the ERC CZ project LL2406 of the Ministry of Education of Czech Republic.
J. Tětek was supported by the VILLUM Foundation grant 16582. This research was partially funded from the Ministry of Education and Science of Bulgaria (support for INSAIT, part of the Bulgarian National Roadmap for Research Infrastructure).
P. Veselý was partially supported by Czech Science Foundation project 24-10306S
and by Center for Foundations of Modern Computer Science (Charles Univ.\ project UNCE 24/SCI/008).
Part of this work was done while A. Łukasiewicz was a PhD student at the University of Wrocław.
Part of this work was done when J. Tětek was visiting at the University of Wrocław and at Charles University. Part of this work was done while J. Tětek was employed at BARC, University of Copenhagen.

\bibliographystyle{plainnat}
\bibliography{literature}

\begin{thebibliography}{55}
\providecommand{\natexlab}[1]{#1}
\providecommand{\url}[1]{\texttt{#1}}
\expandafter\ifx\csname urlstyle\endcsname\relax
  \providecommand{\doi}[1]{doi: #1}\else
  \providecommand{\doi}{doi: \begingroup \urlstyle{rm}\Url}\fi

\bibitem[::a()]{::apachedubbodocumentation}
{Apache Dubbo {Documentation}} -- metrics.
\newblock URL \url{https://dubbo.apache.org/en/overview/mannual/java-sdk/reference-manual/merics/meter/}.
\newblock Accessed: 2025-10-13.

\bibitem[::e()]{::elasticdocs}
{Elastic Docs} -- {Percentiles}.
\newblock URL \url{https://www.elastic.co/docs/reference/aggregations/search-aggregations-metrics-percentile-aggregation}.
\newblock Accessed: 2025-10-13.

\bibitem[::t()]{::tigerdatadocumentation}
{TigerData Documentation} -- {Tdigest()}.
\newblock URL \url{https://docs.tigerdata.com/api/latest/hyperfunctions/percentile-approximation/tdigest/}.
\newblock Accessed: 2025-10-13.

\bibitem[Agarwal et~al.(2013)Agarwal, Cormode, Huang, Phillips, Wei, and Yi]{AgarwalCHPWY13}
Pankaj~K. Agarwal, Graham Cormode, Zengfeng Huang, Jeff~M. Phillips, Zhewei Wei, and Ke~Yi.
\newblock Mergeable summaries.
\newblock \emph{{ACM} Trans. Database Syst.}, 38\penalty0 (4):\penalty0 26, 2013.
\newblock \doi{10.1145/2500128}.

\bibitem[Assadi et~al.(2023)Assadi, Joshi, Prabhu, and Shah]{AssadiJPS23-GK_weighted}
Sepehr Assadi, Nirmit Joshi, Milind Prabhu, and Vihan Shah.
\newblock Generalizing greenwald-khanna streaming quantile summaries for weighted inputs.
\newblock In \emph{26th International Conference on Database Theory, {ICDT} 2023, March 28-31, 2023, Ioannina, Greece}, volume 255 of \emph{LIPIcs}, pages 19:1--19:19. Schloss Dagstuhl - Leibniz-Zentrum f{\"{u}}r Informatik, 2023.
\newblock \doi{10.4230/LIPICS.ICDT.2023.19}.

\bibitem[Baker and Langmead(2023)]{baker2023genomic}
Daniel~N Baker and Ben Langmead.
\newblock Genomic sketching with multiplicities and locality-sensitive hashing using dashing 2.
\newblock \emph{Genome Research}, 33\penalty0 (7), 2023.
\newblock \doi{10.1101/gr.277655.123}.

\bibitem[Bonnie et~al.(2024)Bonnie, Ahmed, and Langmead]{bonnie2024dandd}
Jessica~K Bonnie, Omar~Y Ahmed, and Ben Langmead.
\newblock {DandD}: efficient measurement of sequence growth and similarity.
\newblock \emph{iScience}, 27\penalty0 (3), 2024.
\newblock \doi{10.1016/j.isci.2024.109054}.

\bibitem[Clemens et~al.(2023)Clemens, Schulz, Gartner, and Hausheer]{ClemensSGH23}
Vera Clemens, Lars{-}Christian Schulz, Marten Gartner, and David Hausheer.
\newblock {DDoS} detection in {P4} using {HYPERLOGLOG} and {COUNTMIN} sketches.
\newblock In \emph{{NOMS} 2023, {IEEE/IFIP} Network Operations and Management Symposium, Miami, FL, USA, May 8-12, 2023}, 2023.
\newblock \doi{10.1109/NOMS56928.2023.10154315}.

\bibitem[Cormode(2021)]{Cormode21a}
Graham Cormode.
\newblock Current trends in data summaries.
\newblock \emph{{SIGMOD} Rec.}, 50\penalty0 (4), 2021.
\newblock \doi{10.1145/3516431.3516433}.

\bibitem[Cormode and Vesel{\'{y}}(2020)]{CormodeV20}
Graham Cormode and Pavel Vesel{\'{y}}.
\newblock A tight lower bound for comparison-based quantile summaries.
\newblock In \emph{Proceedings of the 39th {ACM} {SIGMOD-SIGACT-SIGAI} Symposium on Principles of Database Systems, {PODS} 2020, Portland, OR, USA, June 14-19, 2020}, pages 81--93. {ACM}, 2020.
\newblock \doi{10.1145/3375395.3387650}.

\bibitem[Cormode and Yi(2020)]{cormodeY20}
Graham Cormode and Ke~Yi.
\newblock \emph{Small summaries for big data}.
\newblock Cambridge University Press, 2020.
\newblock \doi{10.1017/9781108769938}.

\bibitem[Cormode et~al.(2006)Cormode, Korn, Muthukrishnan, and Srivastava]{cormode2006space}
Graham Cormode, Flip Korn, S~Muthukrishnan, and Divesh Srivastava.
\newblock Space- and time-efficient deterministic algorithms for biased quantiles over data streams.
\newblock In \emph{Proceedings of the 25th ACM SIGMOD-SIGACT-SIGART symposium on Principles of database systems}, PODS '06, pages 263--272. ACM, 2006.
\newblock \doi{10.1145/1142351.1142389}.

\bibitem[Cormode et~al.(2021)Cormode, Mishra, Ross, and Vesel{\'{y}}]{CormodeMRV21}
Graham Cormode, Abhinav Mishra, Joseph Ross, and Pavel Vesel{\'{y}}.
\newblock Theory meets practice at the median: {A} worst case comparison of relative error quantile algorithms.
\newblock In \emph{{KDD} '21: The 27th {ACM} {SIGKDD} Conference on Knowledge Discovery and Data Mining, Virtual Event, Singapore, August 14-18, 2021}, pages 2722--2731. {ACM}, 2021.
\newblock \doi{10.1145/3447548.3467152}.

\bibitem[Cormode et~al.(2023)Cormode, Karnin, Liberty, Thaler, and Vesel{\'{y}}]{CormodeKLTV23}
Graham Cormode, Zohar~S. Karnin, Edo Liberty, Justin Thaler, and Pavel Vesel{\'{y}}.
\newblock Relative error streaming quantiles.
\newblock \emph{J. {ACM}}, 70\penalty0 (5):\penalty0 30:1--30:48, 2023.
\newblock \doi{10.1145/3617891}.

\bibitem[Dong et~al.(2024)Dong, Fan, Bai, Yang, Xue, Chen, and Wu]{dong.etal:2024:2024ieee40thint.conf.dataeng.icdea}
Siyuan Dong, Zhuochen Fan, Tianyu Bai, Tong Yang, Hanyu Xue, Peiqing Chen, and Yuhan Wu.
\newblock M4: {{A Framework}} for {{Per-Flow Quantile Estimation}}.
\newblock In \emph{2024 {{IEEE}} 40th {{International Conference}} on {{Data Engineering}} ({{ICDE}})}, pages 4787--4800. IEEE, 2024.
\newblock ISBN 9798350317152.
\newblock \doi{10.1109/ICDE60146.2024.00364}.

\bibitem[Dunning(2021)]{dunning21}
Ted Dunning.
\newblock The t-digest: Efficient estimates of distributions.
\newblock \emph{Software Impacts}, 7:\penalty0 100049, 2021.
\newblock ISSN 2665-9638.
\newblock \doi{https://doi.org/10.1016/j.simpa.2020.100049}.

\bibitem[Dunning and Ertl(2019)]{dunning19-t-digest}
Ted Dunning and Otmar Ertl.
\newblock Computing extremely accurate quantiles using t-digests.
\newblock \emph{CoRR}, abs/1902.04023, 2019.
\newblock URL \url{http://arxiv.org/abs/1902.04023}.

\bibitem[Epicoco et~al.(2020)Epicoco, Melle, Cafaro, Pulimeno, and Morleo]{EpicocoMCPM20}
Italo Epicoco, Catiuscia Melle, Massimo Cafaro, Marco Pulimeno, and Giuseppe Morleo.
\newblock Uddsketch: Accurate tracking of quantiles in data streams.
\newblock \emph{{IEEE} Access}, 8:\penalty0 147604--147617, 2020.
\newblock \doi{10.1109/ACCESS.2020.3015599}.

\bibitem[Fan and Gijbels(1996)]{fan.gijbels:1996:}
Jianqing Fan and Irene Gijbels.
\newblock \emph{Local {{Polynomial Modelling}} and {{Its Applications}}: {{Monographs}} on {{Statistics}} and {{Applied Probability}} 66}.
\newblock Chapman Hall, 1996.
\newblock ISBN 978-0-412-98321-4.
\newblock \doi{10.1201/9780203748725}.

\bibitem[Felber and Ostrovsky(2015)]{felber2015randomized}
David Felber and Rafail Ostrovsky.
\newblock A randomized online quantile summary in {O}(1/epsilon * log(1/epsilon)) words.
\newblock In \emph{Approximation, Randomization, and Combinatorial Optimization. Algorithms and Techniques (APPROX/RANDOM 2015)}, volume~40 of \emph{Leibniz International Proceedings in Informatics (LIPIcs)}, pages 775--785, Dagstuhl, Germany, 2015. Schloss Dagstuhl--Leibniz-Zentrum fuer Informatik.
\newblock ISBN 978-3-939897-89-7.
\newblock \doi{10.4230/LIPIcs.APPROX-RANDOM.2015.775}.
\newblock URL \url{http://drops.dagstuhl.de/opus/volltexte/2015/5335}.

\bibitem[Fritsch and Butland(1984)]{fritsch.butland:1984:siamj.sci.andstat.comput.}
F.~N. Fritsch and J.~Butland.
\newblock A {{Method}} for {{Constructing Local Monotone Piecewise Cubic Interpolants}}.
\newblock \emph{SIAM J. Sci. and Stat. Comput.}, 5\penalty0 (2):\penalty0 300--304, 1984.
\newblock ISSN 0196-5204.
\newblock \doi{10.1137/0905021}.

\bibitem[Fritsch and Carlson(1980)]{fritsch1980monotone}
Frederick~N Fritsch and Ralph~E Carlson.
\newblock Monotone piecewise cubic interpolation.
\newblock \emph{SIAM Journal on Numerical Analysis}, 17\penalty0 (2):\penalty0 238--246, 1980.

\bibitem[Gan et~al.(2018)Gan, Ding, Tai, Sharan, and Bailis]{GanDTSB18}
Edward Gan, Jialin Ding, Kai~Sheng Tai, Vatsal Sharan, and Peter Bailis.
\newblock Moment-based quantile sketches for efficient high cardinality aggregation queries.
\newblock \emph{Proc. {VLDB} Endow.}, 11\penalty0 (11):\penalty0 1647--1660, 2018.
\newblock \doi{10.14778/3236187.3236212}.

\bibitem[Greenwald and Khanna(2001)]{greenwald2001space}
Michael Greenwald and Sanjeev Khanna.
\newblock Space-efficient online computation of quantile summaries.
\newblock In \emph{ACM SIGMOD Record}, volume~30, pages 58--66. ACM, 2001.
\newblock \doi{10.1145/375663.375670}.

\bibitem[Gribelyuk et~al.(2024)Gribelyuk, Sawettamalya, Wu, and Yu]{GribelyukSWY24_GK-KG_sketch}
Elena Gribelyuk, Pachara Sawettamalya, Hongxun Wu, and Huacheng Yu.
\newblock Simple {\&} optimal quantile sketch: Combining {Greenwald-Khanna with Khanna-Greenwald}.
\newblock \emph{Proc. {ACM} Manag. Data}, 2\penalty0 (2):\penalty0 109, 2024.
\newblock \doi{10.1145/3651610}.

\bibitem[Gribelyuk et~al.(2025)Gribelyuk, Sawettamalya, Wu, and Yu]{GribelyukSWY25elasticCompactors}
Elena Gribelyuk, Pachara Sawettamalya, Hongxun Wu, and Huacheng Yu.
\newblock Near-optimal relative error streaming quantile estimation via elastic compactors.
\newblock In \emph{Proceedings of the 2025 Annual {ACM-SIAM} Symposium on Discrete Algorithms, {SODA} 2025, New Orleans, LA, USA, January 12-15, 2025}, pages 3486--3529. {SIAM}, 2025.
\newblock \doi{10.1137/1.9781611978322.115}.

\bibitem[Guo et~al.(2023)Guo, Hong, Wu, Liu, Yang, and Cui]{guo.etal:2023:proc.29thacmsigkddconf.knowl.discov.datamin.}
Jiarui Guo, Yisen Hong, Yuhan Wu, Yunfei Liu, Tong Yang, and Bin Cui.
\newblock {{SketchPolymer}}: {{Estimate Per-item Tail Quantile Using One Sketch}}.
\newblock In \emph{Proceedings of the 29th {{ACM SIGKDD Conference}} on {{Knowledge Discovery}} and {{Data Mining}}}, {{KDD}} '23, pages 590--601. Association for Computing Machinery, 2023.
\newblock ISBN 9798400701030.
\newblock \doi{10.1145/3580305.3599505}.

\bibitem[Gupta et~al.(2024)Gupta, Singhal, and Wu]{GuptaSW24optimaQS-non-comp}
Meghal Gupta, Mihir Singhal, and Hongxun Wu.
\newblock Optimal quantile estimation: Beyond the comparison model.
\newblock \emph{2024 IEEE 65th Annual Symposium on Foundations of Computer Science (FOCS)}, pages 1137--1158, 2024.
\newblock \doi{10.1109/FOCS61266.2024.00075}.

\bibitem[He et~al.(2023)He, Zhu, and Huang]{he.etal:2023:2023ieee39thint.conf.dataeng.icdea}
Jintao He, Jiaqi Zhu, and Qun Huang.
\newblock {{HistSketch}}: {{A Compact Data Structure}} for {{Accurate Per-Key Distribution Monitoring}}.
\newblock In \emph{2023 {{IEEE}} 39th {{International Conference}} on {{Data Engineering}} ({{ICDE}})}, pages 2008--2021, 2023.
\newblock \doi{10.1109/ICDE55515.2023.00156}.

\bibitem[Hebrail and Berard(2006)]{individual_household_electric_power_consumption_235}
Georges Hebrail and Alice Berard.
\newblock {Individual Household Electric Power Consumption}.
\newblock UCI Machine Learning Repository, 2006.
\newblock {DOI}: https://doi.org/10.24432/C58K54.

\bibitem[Karnin et~al.(2016)Karnin, Lang, and Liberty]{KarninLL16}
Zohar~S. Karnin, Kevin~J. Lang, and Edo Liberty.
\newblock Optimal quantile approximation in streams.
\newblock In \emph{{IEEE} 57th Annual Symposium on Foundations of Computer Science, {FOCS} 2016, 9-11 October 2016, Hyatt Regency, New Brunswick, New Jersey, {USA}}, pages 71--78. {IEEE} Computer Society, 2016.
\newblock \doi{10.1109/FOCS.2016.17}.

\bibitem[Kelly et~al.(2023)Kelly, Longjohn, and Nottingham]{UCI_MLrepo}
Markelle Kelly, Rachel Longjohn, and Kolby Nottingham.
\newblock {The UCI Machine Learning Repository}.
\newblock https://archive.ics.uci.edu, 2023.

\bibitem[Kipf et~al.(2019)Kipf, Marcus, van Renen, Stoian, Kemper, Kraska, and Neumann]{sosd}
Andreas Kipf, Ryan Marcus, Alexander van Renen, Mihail Stoian, Alfons Kemper, Tim Kraska, and Thomas Neumann.
\newblock {SOSD:} {A} benchmark for learned indexes.
\newblock In \emph{NeurIPS Workshop on Machine Learning for Systems}, 2019.

\bibitem[Kipf et~al.(2020)Kipf, Marcus, van Renen, Stoian, Kemper, Kraska, and Neumann]{kipf.etal:2020:}
Andreas Kipf, Ryan Marcus, Alexander van Renen, Mihail Stoian, Alfons Kemper, Tim Kraska, and Thomas Neumann.
\newblock {{RadixSpline}}: A single-pass learned index.
\newblock In \emph{Proceedings of the {{Third International Workshop}} on {{Exploiting Artificial Intelligence Techniques}} for {{Data Management}}}, {{aiDM}} '20, pages 1--5. Association for Computing Machinery, 2020.
\newblock ISBN 978-1-4503-8029-4.
\newblock \doi{10.1145/3401071.3401659}.

\bibitem[Luo et~al.(2016)Luo, Wang, Yi, and Cormode]{luo16_quantiles_experimental}
Ge~Luo, Lu~Wang, Ke~Yi, and Graham Cormode.
\newblock Quantiles over data streams: Experimental comparisons, new analyses, and further improvements.
\newblock \emph{The VLDB Journal}, 25\penalty0 (4):\penalty0 449--472, August 2016.
\newblock ISSN 1066-8888.
\newblock \doi{10.1007/s00778-016-0424-7}.

\bibitem[Manku et~al.(1999)Manku, Rajagopalan, and Lindsay]{manku1999random}
Gurmeet~Singh Manku, Sridhar Rajagopalan, and Bruce~G Lindsay.
\newblock Random sampling techniques for space efficient online computation of order statistics of large datasets.
\newblock In \emph{ACM SIGMOD Record}, volume~28, pages 251--262. ACM, 1999.

\bibitem[Marcus et~al.(2020)Marcus, Kipf, van Renen, Stoian, Misra, Kemper, Neumann, and Kraska]{MarcusKRSMK0K20sosd}
Ryan Marcus, Andreas Kipf, Alexander van Renen, Mihail Stoian, Sanchit Misra, Alfons Kemper, Thomas Neumann, and Tim Kraska.
\newblock Benchmarking learned indexes.
\newblock \emph{Proc. {VLDB} Endow.}, 14\penalty0 (1):\penalty0 1--13, 2020.
\newblock \doi{10.14778/3421424.3421425}.

\bibitem[Masson et~al.(2019)Masson, Rim, and Lee]{MassonRL19}
Charles Masson, Jee~E. Rim, and Homin~K. Lee.
\newblock Ddsketch: {A} fast and fully-mergeable quantile sketch with relative-error guarantees.
\newblock \emph{Proc. {VLDB} Endow.}, 12\penalty0 (12):\penalty0 2195--2205, 2019.
\newblock \doi{10.14778/3352063.3352135}.

\bibitem[Misra and Gries(1982)]{misra1982finding}
Jayadev Misra and David Gries.
\newblock Finding repeated elements.
\newblock \emph{Science of computer programming}, 2\penalty0 (2):\penalty0 143--152, 1982.
\newblock \doi{10.1016/0167-6423(82)90012-0}.

\bibitem[Mitchell et~al.(2021)Mitchell, Frank, and Holmes]{MitchellFH21}
Rory Mitchell, Eibe Frank, and Geoffrey Holmes.
\newblock An empirical study of moment estimators for quantile approximation.
\newblock \emph{{ACM} Trans. Database Syst.}, 46\penalty0 (1):\penalty0 3:1--3:21, 2021.
\newblock \doi{10.1145/3442337}.

\bibitem[Moerkotte et~al.(2014)Moerkotte, DeHaan, May, Nica, and B{\"o}hm]{moerkotte2014exploiting}
Guido Moerkotte, David DeHaan, Norman May, Anisoara Nica, and Alexander B{\"o}hm.
\newblock Exploiting ordered dictionaries to efficiently construct histograms with q-error guarantees in {SAP HANA}.
\newblock In \emph{Proceedings of the 2014 ACM SIGMOD international conference on Management of data}, pages 361--372, 2014.

\bibitem[Munro and Paterson(1980)]{MunroP80}
J.~Ian Munro and Mike Paterson.
\newblock Selection and sorting with limited storage.
\newblock \emph{Theor. Comput. Sci.}, 12:\penalty0 315--323, 1980.
\newblock \doi{10.1016/0304-3975(80)90061-4}.

\bibitem[Neumann and Michel(2008)]{neumann.michel:2008:shar.datainf.knowl.}
Thomas Neumann and Sebastian Michel.
\newblock Smooth {{Interpolating Histograms}} with {{Error Guarantees}}.
\newblock In \emph{Sharing {{Data}}, {{Information}} and {{Knowledge}}}, pages 126--138. Springer, 2008.
\newblock ISBN 978-3-540-70504-8.
\newblock \doi{10.1007/978-3-540-70504-8_12}.

\bibitem[Pagh and Stausholm(2021)]{PaghS21}
Rasmus Pagh and Nina~Mesing Stausholm.
\newblock Efficient differentially private $f_0$ linear sketching.
\newblock In \emph{ICDT'21}, 2021.
\newblock \doi{10.4230/LIPICS.ICDT.2021.18}.

\bibitem[Pfeil et~al.(2025)Pfeil, Horn, Polychroniou, Erickson, Eng, Cai, and Kraska]{pfeil.etal:2025:proc.21stint.workshopdatamanag.newhardw.}
Pascal Pfeil, Dominik Horn, Orestis Polychroniou, George Erickson, Zhe~Heng Eng, Mengchu Cai, and Tim Kraska.
\newblock Insert-{{Optimized Implementation}} of {{Streaming Data Sketches}}.
\newblock In \emph{Proceedings of the 21st {{International Workshop}} on {{Data Management}} on {{New Hardware}}}, {{DaMoN}} '25, pages 1--8. Association for Computing Machinery, 2025.
\newblock ISBN 979-8-4007-1940-0.
\newblock \doi{10.1145/3736227.3736238}.

\bibitem[Rothchild et~al.(2020)Rothchild, Panda, Ullah, Ivkin, Stoica, Braverman, Gonzalez, and Arora]{RothchildPUISB020}
Daniel Rothchild, Ashwinee Panda, Enayat Ullah, Nikita Ivkin, Ion Stoica, Vladimir Braverman, Joseph Gonzalez, and Raman Arora.
\newblock Fetchsgd: Communication-efficient federated learning with sketching.
\newblock In \emph{ICML'20}, 2020.
\newblock \href{http://proceedings.mlr.press/v119/rothchild20a.html}{Paper link.}

\bibitem[Schiefer et~al.(2023)Schiefer, Chen, Indyk, Narayanan, Silwal, and Wagner]{SchieferCINSW23-KLLinterpolation}
Nicholas Schiefer, Justin~Y. Chen, Piotr Indyk, Shyam Narayanan, Sandeep Silwal, and Tal Wagner.
\newblock Learned interpolation for better streaming quantile approximation with worst-case guarantees.
\newblock In \emph{{SIAM} Conference on Applied and Computational Discrete Algorithms, {ACDA} 2023, Seattle, WA, USA, May 31 - June 2, 2023}, pages 87--97. {SIAM}, 2023.
\newblock \doi{10.1137/1.9781611977714.8}.

\bibitem[Shahout et~al.(2023)Shahout, Friedman, and Ben~Basat]{shahout.etal:2023:proc.acmmanag.data}
Rana Shahout, Roy Friedman, and Ran Ben~Basat.
\newblock Together is {{Better}}: {{Heavy Hitters Quantile Estimation}}.
\newblock \emph{Proceedings of the ACM on Management of Data}, 1\penalty0 (1):\penalty0 83:1--83:25, 2023.
\newblock \doi{10.1145/3588937}.

\bibitem[Shrivastava et~al.(2004)Shrivastava, Buragohain, Agrawal, and Suri]{shrivastava2004medians}
Nisheeth Shrivastava, Chiranjeeb Buragohain, Divyakant Agrawal, and Subhash Suri.
\newblock Medians and beyond: new aggregation techniques for sensor networks.
\newblock In \emph{Proceedings of the 2nd international conference on Embedded networked sensor systems}, pages 239--249. ACM, 2004.

\bibitem[Siffer et~al.(2017)Siffer, Fouque, Termier, and Largou{\"{e}}t]{SifferFTL17}
Alban Siffer, Pierre{-}Alain Fouque, Alexandre Termier, and Christine Largou{\"{e}}t.
\newblock Anomaly detection in streams with extreme value theory.
\newblock In \emph{Proceedings of the 23rd {ACM} {SIGKDD} International Conference on Knowledge Discovery and Data Mining, Halifax, NS, Canada, August 13 - 17, 2017}, pages 1067--1075. {ACM}, 2017.
\newblock \doi{10.1145/3097983.3098144}.

\bibitem[Smith et~al.(2020)Smith, Song, and Thakurta]{Smith0T20}
Adam~D. Smith, Shuang Song, and Abhradeep Thakurta.
\newblock The {Flajolet-Martin} sketch itself preserves differential privacy: {Private} counting with minimal space.
\newblock In \emph{Advances in Neural Information Processing Systems 33: Annual Conference on Neural Information Processing Systems 2020, NeurIPS 2020, December 6-12, 2020, virtual}, 2020.
\newblock URL \url{https://proceedings.neurips.cc/paper/2020/hash/e3019767b1b23f82883c9850356b71d6-Abstract.html}.

\bibitem[Tene(2015)]{tene2015_latency_talk}
Gil Tene.
\newblock How {NOT} to measure latency.
\newblock Talk available at \url{https://www.youtube.com/watch?v=lJ8ydIuPFeU}, 2015.

\bibitem[Whiteson(2016)]{hepmass}
Daniel Whiteson.
\newblock {HEPMASS}.
\newblock UCI Machine Learning Repository, 2016.

\bibitem[Wong(2020)]{Wong2020splines}
Jeffrey Wong.
\newblock Lecture notes on splines.
\newblock \url{https://services.math.duke.edu/~jtwong/math563-2020/lectures/Lec1b-splines.pdf}, 2020.
\newblock Accessed: 2024-10-09.

\bibitem[Zhang and Wang(2007)]{zhangwang}
Qi~Zhang and Wei Wang.
\newblock An efficient algorithm for approximate biased quantile computation in data streams.
\newblock In \emph{Proceedings of the 16th ACM conference on Conference on information and knowledge management}, pages 1023--1026, 2007.

\end{thebibliography}

\appendix

\section{Bounds on Middle Point of Monotonic Cubic Polynomials}\label{app:monotoneCubicPolyMiddleBounds}

The following lemma provides a key technical property of PCHIP interpolation
that we use subsequently in the analysis of our algorithm in Appendix~\ref{sec:analysis}.

\begin{lemma}
	\label{lem:cubicPolyMonoMiddleBound}
	Let \( P(x) \) be a cubic polynomial such that \( P(0) = 0 \), \( P(1) = 1 \), and \( P(x) \) is non-decreasing on the interval \( [0,1] \). Then,
	$P\left( \dfrac{1}{2} \right) \in [\beta, 1 - \beta] \approx [0.067, 0.933],$ %
	where $\beta = \frac{1}{2} - \frac{\sqrt{3}}{4} \approx 0.067$.
\end{lemma}

\begin{proof}
	We prove that $P( \frac{1}{2} ) \leq \frac{1}{2} + \frac{\sqrt{3}}{4}$. The other inequality then holds by symmetry, that is, considering the polynomial $Q(x) = 1-P(1-x)$ which is also non-decreasing and $Q( \frac{1}{2} ) = 1-P( \frac{1}{2} )$. From now on, we focus on proving that $P( \frac{1}{2} ) \leq \frac{1}{2} + \frac{\sqrt{3}}{4}$.
	
	Let us denote the coefficients of $P$ by $a, b, c$ and $d$, that is,
	\[
	P(\lambda) = a \lambda^3 + b \lambda^2 + c \lambda + d \,.
	\]
	Because $P(0) = 0$, we have that $d=0$. At the same time, because $P(1) = 1$, we have that $a+b+c = 1$.
	From the assumption that the polynomial is non-decreasing, we also have that its derivative is non-negative, meaning that
	for any $\lambda\in [0,1]$,
	\[
	3 a \lambda^2 + 2b \lambda + c \geq 0\,.
	\]
	
	At $1/2$, the polynomial evaluates to
	\[
	f(a,b,c) = a/8 + b/4 + c/2 \,.
	\]
	
	We need to prove that under the constraints we mentioned, it holds that 
	\[
	f(a,b,c) \leq \frac{1}{2} + \frac{\sqrt{3}}{4}
	\]
	
	First, using the constraint \( a + b + c = 1 \), we express \( c \) as \( c = 1 - a - b \). Substituting this into the inequality \( 3a\lambda^2 + 2b\lambda + c \geq 0 \), we obtain
	$3a\lambda^2 + 2b\lambda + (1 - a - b) \geq 0$.
	Simplifying this expression gives:
	\begin{equation}\label{eqn:cubicPolyMonoMiddleBound-1}
		(3\lambda^2 - 1)a + (2\lambda - 1)b + 1 \geq 0\,.
	\end{equation}
	
	Next, we choose \( \lambda = \dfrac{3 + \sqrt{3}}{6} \), which lies within \( [0, 1] \). We compute the coefficients:
	\[
	3\lambda^2 - 1 = \dfrac{\sqrt{3}}{2}, \quad 2\lambda - 1 = \dfrac{\sqrt{3}}{3}\,.
	\]
	Substituting these values into inequality~\eqref{eqn:cubicPolyMonoMiddleBound-1} gives:
	\[
	\dfrac{\sqrt{3}}{2} a + \dfrac{\sqrt{3}}{3} b + 1 \geq 0\,,
	\]
	which simplifies to:
	\begin{equation}\label{eqn:cubicPolyMonoMiddleBound-2}
		\sqrt{3} \left( \dfrac{a}{2} + \dfrac{b}{3} \right) + 1 \geq 0\,.
	\end{equation}
	We now express \( f(a, b, c) \) in terms of \( a \) and \( b \). Since \( c = 1 - a - b \), we have:
	\begin{equation}\label{eqn:cubicPolyMonoMiddleBound-3}
		f(a, b, c) = \dfrac{a}{8} + \dfrac{b}{4} + \dfrac{1 - a - b}{2} = \dfrac{1}{2} - \dfrac{3a}{8} - \dfrac{b}{4}\,.
	\end{equation}
	To relate inequality~\eqref{eqn:cubicPolyMonoMiddleBound-2} to \( f(a, b, c) \), observe that inequality~\eqref{eqn:cubicPolyMonoMiddleBound-2} implies:
	\[
	\dfrac{a}{2} + \dfrac{b}{3} \geq -\dfrac{1}{\sqrt{3}}.
	\]
	Multiplying both sides by \( -\frac{3}{4} \) gives:
	\[
	-\left( \dfrac{3a}{8} + \dfrac{b}{4} \right) \leq \dfrac{3}{4\sqrt{3}}.
	\]
	From Equation~\eqref{eqn:cubicPolyMonoMiddleBound-3}, we have:
	\[
	f(a, b, c) = \dfrac{1}{2} - \left( \dfrac{3a}{8} + \dfrac{b}{4} \right)\,.
	\]
	Using the inequality derived above, we obtain:
	\[
	f(a, b, c) \leq \dfrac{1}{2} + \dfrac{3}{4\sqrt{3}} = \dfrac{1}{2} + \dfrac{\sqrt{3}}{4}\,.
	\]
\end{proof}

By scaling on both x- and y-axes and using the monotonicity of PCHIP~\cite{fritsch.butland:1984:siamj.sci.andstat.comput., fritsch1980monotone} and exactness on the thresholds, Lemma~\ref{lem:cubicPolyMonoMiddleBound} implies that using PCHIP interpolation
for splitting buckets in the middle of the gap between two thresholds gives a pair of two buckets with counters smaller than the original bucket counter by at least a factor $1 - \beta$ for constant $\beta > 0$ (here, we do not take the buffer items into account which cause certain technical complications in the subsequent proofs). This property is further used in our analysis in Appendix~\ref{sec:analysis}.
In fact, we have verified that $\beta = 1/8$ for the PCHIP interpolation from the SciPy library that we use in the implementation.
We note that better bounds on $\beta$ can be obtained for specific interpolation methods.
For the linear interpolation we have $\beta = 0.5$.

\section{Analysis of Uniform Error Guarantees} %
\label{sec:analysis}

We provide a formal analysis of SplineSketch with respect to the uniform error guarantees.
We first focus on the streaming setting and then analyze mergeability.
We also prove that the algorithm can always find a pair of buckets to join when we need to split a bucket (\Cref{lem:existsBucketToMerge}), showing that the algorithm is well-defined;
the latter is proven in the most general mergeability setting, i.e., when the sketch is created
by an arbitrary sequence of pairwise merge operations executed on single data items (\Cref{sec:mergeTree}), 
only assuming that all sketches in the process have the same number of buckets $k$.

In the subsequent analysis, we assume that $k\ge 6$.
The precise constant factors in the analysis depend on the interpolation used and in particular,
on the value of $\beta$ as defined in Appendix~\ref{app:monotoneCubicPolyMiddleBounds}.
Next, we will require that the constant factor $C_b$ in the bucket bound~\eqref{eqn:bucketBound}
is large enough, namely $C_b \ge 18\cdot 2.5 / \beta$, which implies $C_b \ge 16$ as $\beta < 1$.

\subsection{Analysis in the Streaming Setting}\label{sec:streamingAnalysis}

Here, we analyze SplineSketch with $k$ buckets created by $n$ update operations, each adding one data item.
By the sketch at time $t \in [0, n]$, we mean the sketch after adding $t$ items.
In the analysis, we consider intermediate states of the sketch when executing the consolidate method and
performing splits and joins, i.e., the state of buckets after each join or split (the number of buckets may thus be $k-1$ or $k+1$). 
Recall that time is split into epochs.
Epoch 0 ends at $n_0 := \Theta(k)$, which is the initial buffer size, and for $j > 0$, the $j$-th epoch
ends at $n_j := \lceil 1.25\cdot n_{j-1}\rceil $.
Note that during epoch 0, there are no buckets and thus the buffer contains all items seen so far.

\subsubsection{Existence of joinable pair of buckets}
In describing the consolidate subroutine (\Cref{sec:consolidate}), we assumed that there is always a pair of buckets that we can join if a bucket counter exceeds the bound~\eqref{eqn:bucketBound} or if there is a new minimum or maximum in the buffer. Now we show that it holds in the streaming setting.

An important technical detail is that when the number of items in sketch is $t \ll k^2$ and the buffer has size $\Theta(k) \gg t/k$,
	the following lemma may not hold for a carefully designed input.
	The reason is that, as we keep the buffer of size $\Theta(k)$ during \consolidate\ and use it for more precise splits, it may not hold that after 
	a split, at least a constant fraction of items ends up in both of the resulting buckets; this is because almost all of the items in a bucket may be from the buffer and close to one of the thresholds.

	To resolve this issue, in the following proof we assume that during \consolidate\ at time $t$ we process the buffer in batches of size at most $(C_b / 2)\cdot t/k$,
	performing splits due to exceeding the bucket bound for each batch.
	This modification is not done in the actual implementation as it would negatively affect the update time for small datasets and possibly also the accuracy;
	furthermore, it is only needed when many more than $t/k$ items are inserted from the buffer to a single bucket in a way
	that they end up in one of the two new buckets created by the split.\footnote{
		Tomáš Domes suggested a correct way to resolve such an issue by performing a ``double split'' when needed:
	That is, split the bucket in the middle by length, as in the current algorithm, and if this does not split the buffer items inside the bucket somewhat evenly, perform another split at the median of the buffer items inside the bucket. An implementation of double splits is left to future work.}

\begin{lemma}\label{lem:existsBucketToMergeStreaming}
For any time $t$ during processing the stream 
after initialization of buckets by Algorithm~\ref{alg:initBuckets},
SplineSketch contains a joinable pair of adjacent buckets, according to Definition~\ref{def:joinableBuckets}.
\end{lemma}

\begin{proof}
	Suppose for a contradiction that there is no joinable pair at time $t$.
	Note that a pair of adjacent buckets is not joinable for one of two reasons: (1) due to bucket bounds, namely that the resulting bucket would have counter exceeding $0.75\cdot C_b \cdot t/k$ (as the current value of $n$ is equal to $t$) or (2) due to protected thresholds.
	Observe that there are at most $2t / (0.75\cdot C_b \cdot t/k) = \frac83 k / C_b \le k/6$ pairs of buckets satisfying case (1), where we use $C_b \ge 16$.
	Thus, the main part of the proof is bounding the number of protected thresholds.
	
	Let $j$ be the epoch at time $t$, i.e., the smallest $j > 0$ such that $t\le n_j$. 
	(For $j=0$, there is nothing to prove as there are no buckets in the sketch.)
	Since there is no joinable pair at time $t$, we take the last time $t' < t$ during the $j$-th epoch
	such that the sketch at $t'$ has at least $k/3 + 2$ joinable pairs of buckets before performing \consolidate (the number of joinable pairs does gets below $k/3 + 2$ during the execution of \consolidate).
	Time $t'$ is well-defined, since at the beginning of the $j$-th epoch before performing any split,
	the sketch has at least $k - 1 - k/6$ joinable pairs
	as there are no protected thresholds and since $k - 1 - k/6 \ge k/3 + 2$ by the assumption that $k\ge 6$.
	After time $t'$, the sketch has always less than $k/3 + 2$ joinable pairs of buckets, and
	thus the algorithm does not perform splits due to the heuristic error, that is,
	it only splits due to exceeding the bucket bound.

	To bound the number of splits at times $(t', t]$, we consider a suitable potential of each sketch state $S$ at time $t''\in (t', t]$ with $k(S)$ buckets, including the intermediate states during the consolidate methods (after performing some splits or joins of buckets; thus, its number of buckets $k(S)$ may also be $k-1$ or $k+1$). The potential is defined as follows:
	\begin{equation}\label{eqn:existsBucketToMerge-potentialStreaming}
		\Psi(S) := \sum_{i = 1}^{k(S)} \max\left\{0, b_i - 0.75\cdot C_b\cdot \frac{t''}{k}\right\}\,,
	\end{equation}
	where $b_i$ is the counter of bucket $(\tau_{i-1}, \tau_i]$ of the sketch $S$.
	Clearly, $\Psi(S) \in [0, t'']$ and moreover, the total increase of the potential during $(t', t]$ is at most $t - t'$.
	Note that the potential does not increase by joining buckets by Definition~\ref{def:joinableBuckets}.
	
	We now analyze the decrease of the potential by a split at time $t''\in (t', t]$; such a split is thus due to exceeding the bucket bound.
	In particular, we now show that the potential decreases by at least $\beta\cdot C_b\cdot t / (2.5\cdot k)$.
	First, suppose that both buckets resulting from the split have counters at least $0.75\cdot C_b\cdot t''/k$; then, the potential decreases by at least $0.75\cdot C_b\cdot t''/k \ge \beta\cdot C_b\cdot t / (2.5\cdot k)$, where we use that $t' \le t''$ and $t \le 1.25\cdot t'$ as $t$ and $t'$ are in the same epoch.
	Otherwise, at least one of the resulting buckets has counter smaller than $0.75\cdot C_b\cdot t''/k$.
	Recall that when splitting, the counters of the new buckets are determined from the interpolation of the original buckets (before executing \consolidate\ that performs the split)
	and the buffer.
	Let $b'_i \le b_i$ be the number of buffer items in the split bucket (including items from the current batch only).
	Since the batch of items added from the buffer has size at most $(C_b / 2)\cdot t''/k$, we have that $b'_i \le (C_b / 2)\cdot t''/k$.
	For the remaining $b_i - b'_i$ items in the split bucket, the interpolation ensures that both buckets resulting from the split
	receive at most a $1 - \beta$ fraction of them.
	As all of the $b'_i$ items from the buffer may end up in one of the resulting buckets,
	it holds that the counters of the resulting buckets are both at most
	$(1 - \beta)\cdot (b_i - b'_i) + b'_i
	= (1 - \beta)\cdot b_i + \beta\cdot b'_i$, and therefore by at least $\beta\cdot  b_i - \beta\cdot b'_i$ smaller than the counter $b_i$ before the split.
	The potential of the bucket before the split is $b_i - 0.75\cdot C_b\cdot t'' / k$.
	Using the definition of the potential and $\beta < 0.25$, a split at time $t''$ decreases the potential by at least\footnote{For $\beta \ge 0.25$, the constant factors in this proof need to be adjusted. Nevertheless, one can use $\beta < 0.25$ even for linear interpolation.}
	$$
	\beta\cdot \left(b_i - b'_i\right)
	> \beta\cdot \left(C_b\cdot \frac{t''}{k} - \frac{C_b}{2}\cdot \frac{t''}{k}\right)
	= \beta\cdot C_b\cdot \frac{t''}{2k} 
	\ge \beta\cdot C_b\cdot \frac{t'}{2k}
	\ge \beta\cdot C_b\cdot \frac{t}{2.5\cdot k}
	$$
	using that $t \le 1.25\cdot t'$ as $t'$ and $t$ are in the same epoch.
	
	Combining the lower bound on the decrease due to a split with the upper bound on the total potential,
	it follows that the number of splits during $(t', t]$ is at most 
	$$\frac{t}{\beta\cdot C_b\cdot \frac{t}{2.5\cdot k}} = \frac{2.5\cdot k}{\beta\cdot C_b}\,.$$
	We choose $C_b$ based on $\beta$ so that $2.5\cdot k / (\beta\cdot C_b) \le k/18$.
	As every split results in three protected thresholds, the number of thresholds that get protected during $(t', t'']$ is at most $k/6$.
	Since there are at least $k/3 + 2$ joinable pairs of buckets at time $t'$,
	at most $k - k/3 - 2 = 2k/3 - 2$ thresholds are protected at time $t'$.
	By time $t$, at most $k/6$ more thresholds get protected.
	As at most $k/6$ of bucket pairs at time $t$ are non-joinable due to falling in case (1),
	there will be a joinable pair of buckets at time $t$, which is a contradiction.
\end{proof}

\subsubsection{Proof of the uniform error guarantee in the streaming setting}
We now formally prove the worst-case guarantee for SplineSketch in the streaming setting.

\mainthmStreaming*

\begin{proof}
First, we assume w.l.o.g.\ that there are no items of frequency above $n/k$ after processing $n$ items
by running the Misra-Gries sketch~\cite{misra1982finding} of size $O(k)$ in parallel, as described in Section~\ref{sec:highFrequencyItems-Impl}.
Using such MG filtering, we maintain the invariant that at any time $t$, every item inserted into the buckets of SplineSketch has frequency at most $t/k$.

Second, we show an upper bound on the error only for the thresholds of the final sketch.
The error bound then extends to any rank or quantile query by the fact that the counter of each bucket has $O(n/k)$ items and that the interpolation is monotone.
Thus, in the following we consider a threshold $\tau_i$ in the final sketch after processing $n$ items and analyze the rank error at $\tau_i$.

To this end, we define the error for any threshold $\tau_a$ in the sketch at time $t$ as follows:
Let $\mathcal{I}_B^t$ be the multiset of items that were added to the buckets by time $t$, 
let $\rank^t(\tau_a) = |\{y \in \mathcal{I}_B^t : y \le \tau_a\}|$ be the $\tau_a$'s rank with respect to items $\mathcal{I}_B^t$,
and let $\estRank^t(\tau_a) := \sum_{j = 1}^i b_j$ be the estimated rank of $\tau_a$ in $\mathcal{I}_B^t$ solely based on the buckets.
Then we define the rank error of $\tau_a$ in the sketch at $t$ as $\Err^t(\tau_a) := \rank^t(\tau_a) - \estRank^t(\tau_a)$.

We show that for any threshold $\tau_i$ of the final sketch, $\Err^n(\tau_i) \le O(\log \alpha\cdot n/k)$;
equivalently, we aim to show that only $O(\log \alpha\cdot n/k)$ input items $y \le \tau_i$
are accounted for in buckets with items greater than $\tau_i$.
The lower bound, i.e., $\Err^n(\tau_i) \ge -O(\log \alpha\cdot n/k)$, can be obtained in a symmetric way.

First note that the error at a threshold $\tau_a$ does not change when items are added to the buckets or when 
performing a join of buckets (as the removed threshold is no longer relevant).
Also, thresholds created by the equally spaced selection of initial buckets have zero error as they are exact at the bucket thresholds, and similarly, thresholds created by new minima or maxima in the buffer have zero error as they are exact at the buffer items and overlap no items in the buckets.
So the error may increase only due to splitting a bucket.

For the threshold $\tau_i$ of the final sketch,
we consider the following sequence of times $t_0, t_1, t_2, \dots$ and
thresholds $\sigma_0, \sigma_1, \sigma_2, \dots$:
Let $\sigma_0 := \tau_i$ and let $t_0$ be the last time of the split that creates $\tau_i$;
if $\tau_i$ is created by the equally spaced selection of initial buckets  (Algorithm~\ref{alg:initBuckets}) or by a new minimum or maximum in the buffer, then $\Err^n(\tau_i) = 0$ and we are done.
Note that the rank error at $\tau_i$ will not increase after the split at $t_0$, by the description of the algorithm.
Then for $a = 1, 2, \dots$, we iteratively define $\sigma_a$ and $t_a$ as follows:
If $\sigma_{a-1}$ is created by a split of a bucket $(\tau', \tau'']$ during the consolidate method at $t_{a-1}$, implying $\sigma_{a-1} = (\tau' + \tau'')/2$,
we set $\sigma_a := \tau'$ and let $t_a$ be the last time up to $t_{a-1}$ such that $\sigma_a$ is created by the split at $t_a$ or by the equally spaced selection of initial buckets at $t_a$ or by a new minimum or maximum in the buffer at $t_a$.
Note that we may have $t_a = t_{a-1}$ as one bucket may be split repeatedly during one call of consolidate method.
Otherwise, $\sigma_{a-1}$ is  created by the equally spaced selection of initial buckets or a new minimum or maximum in the buffer at time $t_{a-1}$, 
and we stop this process.
Let $m$ be the index when this process stops; that is, at $t_m$, threshold $\sigma_m$ was created
by the equally spaced selection.
We have that $\Err^{t_m}(\sigma_m) = 0$, as the equally spaced selection is exact at the bucket thresholds,
and that $\Err^{t_0}(\sigma_0) = \Err^{t_0}(\tau_i) = \Err^{n}(\tau_i)$.

Let $j_a$ be the epoch of the sketch represented by $t_a$.
We claim that for any epoch $j$, there are at most $O(\log \alpha)$ indexes $a$ with $j_a = j$.
To this end, recall that any threshold created by a split during epoch $j$ is not removed by a join during the same epoch due to the threshold protection.
Let $b$ be the smallest $a$ such that $j_a = j$; note that as the indexes are in the reverse order of time, $\sigma_b$ is the last threshold from the sequence $[\sigma_0, \sigma_1, \dots, \sigma_m]$ created during epoch $j$.
If for $a = b+1$, time $t_a \le t_b$ is still in epoch $j$ (i.e., $j_a = j$), then 
after the split at time $t_a$, the new bucket $(\sigma_a, \sigma'_a]$ contains $\sigma_b$,
and both thresholds $\sigma_a, \sigma'_a$ get protected.
Using the induction over $a > b$ such that $j_a = j$, 
it holds that after the split at time $t_a$, the new bucket $(\sigma_a, \sigma'_a]$ contains $\sigma_b$, i.e., $\sigma_b\in (\sigma_a, \sigma'_a]$.
Observe that by the definition of $\alpha$ and since we always split a bucket in the middle, after $O(\log \alpha)$ splits that involve buckets containing $\sigma_b$, the length of the bucket containing $\sigma_b$ will be below the smallest distance between distinct items, i.e., the bucket will cover at most one distinct item $x$ of the input.
The algorithm may still split the bucket containing $\sigma_b$ after these $O(\log \alpha)$ splits but the number of these additional splits is only $O(1)$, which follows from the following facts: First,
note that there are no items of frequency above $t_b/k$ by the assumption,
which implies that the number of buffer items within a bucket containing a single distinct input item is at most $t_b/k$.
Therefore, for a large-enough constant $C_b$, every split due to the bucket bound decreases the bucket counter by a factor smaller than 1, close to $1 - \beta \approx 0.934$.
Finally, we do not split a bucket with counter below $2t_b/k$ by Definition~\ref{def:splittableBuckets}, even due to the heuristic error.
This proves the claim that there are at most $O(\log \alpha)$ indexes $a$ with $j_a = j$, i.e., splits in the sequence during epoch $j$.

Moreover, each split in the sequence at time $t_a$ with $j_a = j$ increases the error by at most $O(n_j / k)$, due to bucket size bounds;
specifically for any $a$ with $j_a = j$, we have that $\Err^{t_a}(\sigma_a) \le \Err^{t_{a+1}}(\sigma_{a+1}) + O(n_j / k)$, because at most $O(n_j / k)$ items no larger than $\sigma_a$ accounted for in the bucket counter of the split bucket $(\sigma_{a+1}, \sigma'_{a+1}]$ may be accounted for in the bucket counter of the new bucket $(\sigma_a, \sigma'_{a+1}]$.
Therefore, the total error from splits during epoch $j$ is at most $O(\log (\alpha)\cdot n_j/k)$.
Since the values of $n_j$ increase geometrically, the total error at $\tau_i$ resulting from splits sums to $O(\log (\alpha)\cdot n/k)$.
This shows $\Err^n(\tau_i) \le O(\log (\alpha)\cdot n/k)$.

A symmetrical proof bounds the error from below, i.e., $\Err^n(\tau_i) \ge -O(\log (\alpha)\cdot n/k)$;
namely, we define a similar sequence of times $t'_0, t'_1, t'_2, \dots$ and
thresholds $\sigma'_0, \sigma'_1, \sigma'_2, \dots$,
with the difference that when $\sigma_{a-1}$ is created by a split of $(\tau', \tau'']$ at $t_{a-1}$, then the next threshold in the sequence is $\sigma_a := \tau''$, instead of $\tau'$.
Therefore, the absolute rank error satisfies $|\Err^n(\tau_i)| \le O(\log (\alpha)\cdot n/k)$, which concludes the proof.
\end{proof}

\subsection{Analysis of Mergeability}\label{sec:mergeabilityAnalysis}

We start by defining the mergeability setting using a binary tree where internal nodes correspond to merge operations or operations inside the sketch, namely splitting or joining buckets. 

\subsubsection{Mergeability setting}\label{sec:mergeTree}

For analyzing sketches created by a sequence of pairwise merge operations, 
it is convenient to define a binary tree $T_0$ with $n$ leaves, each representing one data item and inner nodes with two children corresponding to the merge operations.
Additionally, we will use inner nodes representing the operations with buckets, namely splitting a bucket or joining a pair of adjacent buckets.
That is, we define $T_0$ with root representing the final sketch, leaves corresponding to the single data items,
and inner nodes of three types:
(1) merge-nodes with two children, representing merge operations such that the sketch
represented by a merge-node is obtained by the union of buffers and buckets from the two source sketches represented by the children,
as described in Section~3.2, %
(2) split-nodes representing one split of a bucket, and (3) join-nodes representing one join of two adjacent buckets;
split-nodes and join-nodes have a single child.
Each node represents the sketch after the operation is performed.
(A technicality is that the number of buckets may possibly be $k-1$ or $k+1$ after a join or split
or even up to $2k$ after taking the union of thresholds during a merge operation.)
Other operations inside the sketch, namely creating initial buckets by the equally spaced selection and merging buffer into buckets, need not be represented in $T_0$ for our analysis as the equally spaced selection is exact on thresholds.

Let $n_t$ be the number of leaves in the subtree of $t$, i.e., the number of data items summarized by the sketch in $t$.

We further label the inner nodes by epochs. Namely, let $n_0 := \Theta(k)$ be the initial buffer size and $n_j := 1.25\cdot n_{j-1}$ for $j > 0$ be the epoch ends.
We call a node $t$ of $T_0$ a \emph{$j$-node} if $j$ is the smallest integer $j\ge 0$ such that
the number of items summarized by the sketch represented by $t$ is at most $n_j$.
Clearly, the labels are non-decreasing on any leaf-to-root path.
Moreover, for any $j>0$, merging two sketches represented by $j$-nodes results in a sketch represented by a merge-node
labeled by $j' > j$ as the two source sketches summarize a similar number of items.
We thus obtain:

\begin{observation}\label{obs:jNodesFormPath}
    For any epoch $j > 0$, the $j$-nodes of $T_0$ form a disjoint union of paths.
\end{observation}

The 0-nodes form possibly more complicated subtrees. However, sketches represented by 0-nodes have no buckets and are exact as they store all their items in the buffer.

\subsection{Existence of Joinable Pair of Buckets}\label{sec:existsBucketToMerge}

In describing the consolidate subroutine (Sec.~3.1), we assumed that there is always a pair of buckets that we can join if a bucket counter exceeds the bound~(1) or if there is a new minimum or maximum in the buffer. Now we show that it holds in the general mergeability setting,
thus generalizing Lemma~\ref{lem:existsBucketToMergeStreaming} from the streaming setting.

Note that we need the same assumption on the buffer when the number of items summarized by the sketch is $n \ll k^2$;
namely that the buffer is incorporated into buckets in batches of size at most $(C_b / 2) \cdot n / k$.

\begin{lemma}\label{lem:existsBucketToMerge}
For any $j$-node $t$ of the merge tree $T_0$ with $j > 0$ (i.e., containing buckets),
the sketch represented by $t$ contains a joinable pair of adjacent buckets, according to Definition~\ref{def:joinableBuckets}.
\end{lemma}

\begin{proof}
	The proof is similar to the streaming setting (Lemma~\ref{lem:existsBucketToMergeStreaming}),
	only using the merge tree instead of time.
	Suppose for a contradiction that there is no joinable pair in the sketch represented by a node $t$ of the merge tree $T_0$.
	Note that a pair of adjacent buckets is not joinable for one of two reasons: (1) due to bucket bounds, namely that the resulting bucket would have counter exceeding $0.75\cdot C_b \cdot n_t/k$ or (2) due to protected thresholds.
	Observe that there are at most $2n_t / (0.75\cdot C_b \cdot n_t/k) = \frac83 k / C_b \le k/6$ pairs of buckets satisfying case (1), where we use $C_b \ge 16$.
	Thus, the main part of the proof is bounding the number of protected thresholds.
	
	By \Cref{obs:jNodesFormPath}, $j$-nodes in the subtree of $t$ form a path $P_j$.
	If the sketches represented by nodes on $P_j$ all have at least $k/3 + 2$ joinable pairs of buckets,
	then the lemma holds.
	Consider the highest node $t'$ on the path $P_j$ such that the sketch represented by $t'$ has at least $k/3 + 2$ joinable pairs of buckets; we have that $t' \neq t$, i.e., $t'$ is below $t$, otherwise the lemma holds.
	Note that the sketch of the lowest node $u$ on $P_j$ has at least $k - 1 - k/6$ joinable pairs
	as there are no protected thresholds (the merge operation still resets the protection bitvector when the epoch changes and taking the union of buckets does not introduce new protected thresholds). It follows that $t'$ is well-defined since $k - 1 - k/6 \ge k/3 + 2$ by $k\ge 6$.

	Let $P'_j$ be the subpath of $P_j$ consisting of nodes above $t'$.
	By the definition of $t'$, all of them have less than $k/3 + 2$ joinable pairs of buckets.
	Thus, the algorithm does not perform splits on sketches on $P'_j$ due to the heuristic error, that is,
	it only splits due to exceeding the bucket bound.
	
	To bound the number of split-nodes on $P'_j$, we consider the following potential of each sketch represented by a node $v$ that has $k(v)\in \{k-1, k, \dots, 2k\}$ buckets:
	\begin{equation}\label{eqn:existsBucketToMerge-potential}
		\Psi(v) := \sum_{i = 1}^{k(v)} \max\left\{0, b_i - 0.75\cdot C_b\cdot \frac{n_v}{k}\right\}\,,
	\end{equation}
	where $b_i$ is the counter of bucket $(\tau_{i-1}, \tau_i]$ of the sketch of $v$.
	Clearly, $\Psi(v) \in [0, n_v]$ and moreover, the total increase of the potential on path $P'_j$ from merging with smaller sketches is at most $n_t$. 
	The potential does not increase by joining buckets by Definition~\ref{def:joinableBuckets}.
	Using the same argument as in the streaming setting, every split at node $v$ of a bucket $i$ with counter $b_i > C_b\cdot n_v/k$ 
	decreases the potential by at least $\beta\cdot C_b\cdot n_{t} / (2.5\cdot k)$.
	Hence, the number of split-nodes on path $P'_j$ is at most 
	$$\frac{n_t}{\beta\cdot C_b\cdot \frac{n_{t}}{2.5\cdot k}} = \frac{1.25\cdot k}{\beta\cdot C_b}\,.$$
	We choose $C_b$ based on $\beta$ so that the number of split-nodes on path $P'_j$ is at most $k/18$.
	As every split results in three protected thresholds, the number of thresholds that get protected on path $P'_j$ is at most $k/6$.
	Since there are at least $k/3 + 2$ joinable pairs of buckets in node $t'$, i.e.,
	at most $k - k/3 - 2 = 2k/3 - 2$ thresholds are protected at node $t'$,
	and at most $k/6$ of bucket pairs get non-joinable due to falling in case (1),
	there will be a joinable pair of buckets in node $t$.
\end{proof}

\subsubsection{Analysis of the Error from Merge Operations}

In the mergeability setting, the error at a threshold $\tau_i$ may also increase due to merge operations,
as we take the union of buckets from two sketches.
We first show that under balanced mergeability, the error is essentially maintained;
the fact that we only merge sketches summarizing a similar number of items will be crucially used to show
that the total error resulting from merging the sketches will be at most $O(n \log (n/k) / k)$.

\mainthmBalancedMergeability*

\begin{proof}[Proof sketch.]
	As in the streaming setting, we maintain the invariant that at any node $t$ of the merge tree summarizing $n_t$ items, every item inserted into the buffer and buckets of SplineSketch has frequency at most $n_tt/k$,
	by running the Misra-Gries sketch~\cite{misra1982finding} of size $O(k)$ in parallel, as described in Section~\ref{sec:highFrequencyItems-Impl}.
	We note that the Misra-Gries sketch is fully mergeable~\cite{AgarwalCHPWY13}.
		
	We again analyze the rank error for any threshold $\tau_i$ of the final sketch.
We consider the binary tree $T_0$ corresponding to merging sketches and operations with buckets, as described in Section~\ref{sec:mergeTree}.
To this end, we define the error for any threshold $\tau_a$ in the sketch represented by a node $t$ as follows:
Let $\mathcal{I}_B^t$ be the multiset of items at leaves in the subtree of $t$ (i.e., the items summarized by the sketch at $t$) that are not in the buffer of the sketch represented by $t$ (i.e., they are accounted for in the buckets), 
let $\rank^t(\tau_a) = |\{y \in \mathcal{I}_B^t : y \le \tau_a\}|$ be the $\tau_a$'s rank with respect to items $\mathcal{I}_B^t$,
and let $\estRank^t(\tau_a) := \sum_{j = 1}^i b_j$ be the estimated rank of $\tau_a$ in $\mathcal{I}_B^t$ solely based on the buckets.
Then we define the rank error of $\tau_a$ in the sketch represented by $t$ as $\Err^t(\tau_a) := \rank^t(\tau_a) - \estRank^t(\tau_a)$.

We show that for any threshold $\tau_i$ of the final sketch at the root $r$, $\Err^r(\tau_i) \le O(\log (n/k)\cdot n/k)$.
The lower bound, i.e., $\Err^r(\tau_i) \ge -O(\log (n/k)\cdot n/k)$, can be obtained in a symmetric way.
Note that the error may increase only due to splitting a bucket or due to taking the union of buckets when merging sketches.

First, observe that the total error from merge nodes due to taking the union of buckets is $O(n \log (n/k) / k)$.
Indeed, the error at a threshold of the sketch represented by a merge-node $t$, with children $t'$ and $t''$,
is at most the sum of errors at corresponding thresholds of $t'$ and $t''$ plus $O(n_t / k)$ due to the split induced by the merge operation
(i.e., rank query inside a bucket of one of the source sketches).
Summing $n_t$ over all merge-nodes $t$ that have buckets (i.e., summarize $\Omega(k)$ items) gives a bound of $O(n \log (n/k) / k)$, which follows from that the merges are balanced;
in more detail, consider a binary tree $T'_0$ where we consider only merge-nodes and note
that sketches on each level of $T'_0$ contain $n$ items in total and
the number of levels of $T'_0$ with sketches that have buckets is bounded by $O(\log (n/k))$.

Second, we consider splits due to the heuristic error. The rank error at a threshold increases by at most $O(n_t / k)$
by these splits during \consolidate\ after a merge operation represented by $t$
as there are $O(1)$ such splits by the assumption.
In total, for any threshold, the error increase due to these splits is at most $O(n \log (n/k) / k)$.

Finally, we account for the error from splits due to exceeding the bucket bound.
Observe that the bucket bound may increase only when the full buffer is merged into the buckets
and since the buffer size is $\Theta(k)$, there are $O(n/k)$ executions of \consolidate\ which merge buffer into buckets.
For any such execution, we use a similar argument as in the streaming setting to show that
only $O(\log \alpha)$ splits during one execution of \consolidate\ affect the error of a single threshold. Furthermore, each such split represented by a node $t$ increases the error by at most $O(n_t / k)$.
Let $t_1, \dots, t_\lambda$ be all merge-nodes $t_i$ such that the buffer is merged into buckets during the merge operation represented by $t_i$, and let $n(t_i)$ be the number of items summarized by the sketch represented by $t_i$. Note that for any split represented by $t$ during the execution of \consolidate\ after merging represented by $t_i$ satisfies $n_t = n(t_i)$.
We claim that $\sum_{i = 1}^{\lambda} n(t_i) = O(n)$; combining this with the bounds above implies that the total error of a threshold of the final sketch is $O(\log \alpha \cdot n / k)$.
Indeed, using the assumption on balanced merging, on any leaf-to-root path,
there are only $O(1)$ merges of the buffer to buckets and each node $t_i$ representing such a merge
satisfies $n(t_i)\in \Theta(k)$. Therefore, the number of items in the resulting sketch can be charged to the contents of the buffer so that each items gets charged only $O(1)$ times (as it is added to the buckets from the buffer only once). This implies the claim that $\sum_{i = 1}^{\lambda} n(t_i) = O(n)$,
and concludes the proof.
\end{proof}

\paragraph{Full mergeability.}
Finally, we consider the most general setting of mergeability.
The previous versions of this paper contain claims that SplineSketch satisfies full mergeability~\cite{AgarwalCHPWY13} with the same error bounds as in the streaming setting, 
i.e., that the error guarantees hold even if the sketch is built by an arbitrary sequence of merge operations.
Unfortunately, the provided argument for this claim has a flaw; specifically, it does not account for the error from taking the union of buckets properly.
Here, we outline a ``counterexample'', i.e., a way of merging SplineSketches in an unbalanced way so that the worst-case error may be $\Omega(n)$,
instead of the claimed $O(\log \alpha\cdot n/k)$.
Then we show a simple (and standard) way to provide full mergeability at the cost of an additional $O(\log (n/k))$ factor in the space bound,
i.e., using space $O(k\log (n/k))$ memory words, while still achieving error $O((\log \alpha + \log (n/k))\cdot n/k)$.

Intuitively, the reason why merging sketches with substantially different number of items is problematic is as follows:
Suppose we somehow create a ``large'' sketch $S$ summarizing $n/2$ items. Next, we perform $k/2$ merges with ``smaller'' sketches, each summarizing $n/k$ items. 
Each of these merges may split a particular bucket of the large sketch and thus increase the worst-case error by $\Omega(n / k)$ for that bucket, leading to the worst-case error of $\Omega(n)$
at a particular threshold after these merges\footnote{
	One could potentially avoid splits of the large sketch and not take the union of bucket thresholds, only using the thresholds of the larger sketch;
	nevertheless, many splits of a single bucket in the large sketch may be enforced as the bucket bound needs to be maintained (this is an issue related to merging the buffer items in the streaming setting).}.
The $O(\log \alpha)$ bound on the number of splits of a single bucket does not apply
as splits induced by merging may be very close to one of the bucket's thresholds and not in the middle.

More generally, the error increase due to the merge represented by each merge-node $t$ of the merge tree can be $\Omega(n_t / k)$,
where $n_t$ is the number of items summarized by the sketch.
For balanced merging, summing over all merge-nodes gives overall worst-case error $\Omega(n \log(n/k) / k)$.
However, for an unbalanced merge tree, the error from merging may be up to $\Omega(\min(n, n^2 / k^2))$.

The simple way to fix it is to maintain $O(\log (n/k))$ instances of SplineSketch,
with the $i$-th one summarizing $(2^{i-1}\cdot k, 2^i\cdot k]$ items for $i = 1, \dots, \lfloor\log_2 (n/k) \rfloor$.
Then the merge operation only merges instances containing almost the same number of items, i.e., we have reduced to (perfectly) balanced merging.
Therefore, in space $O(k\log (n/k))$, we can guarantee error of $O((\log \alpha + \log (n/k)) \cdot n/k)$ by Theorem~\ref{thm:balancedMergeability}.

\end{document}